\documentclass{article}
\usepackage{amsmath}
\usepackage{hyperref}
\usepackage{amssymb}
\usepackage{booktabs}

\usepackage{algorithmicx}
\usepackage[ruled]{algorithm}
\usepackage{algpseudocode}
\usepackage{framed}
\usepackage{enumitem}
\usepackage{listings}
\usepackage{subfig}
\usepackage{macros}
\usepackage{ioa_code_small-1}
\algnewcommand\algorithmicforeach{\textbf{for each}}  
\algdef{S}[FOR]{ForEach}[1]{\algorithmicforeach\ #1\ \algorithmicdo}

\usepackage{amsthm}
\theoremstyle{definition}
\newtheorem{definition}{Definition}[section]

\newtheorem{theorem}{Theorem}[section]
\newtheorem{corollary}{Corollary}[section]
\newtheorem{lemma}{Lemma}[section]

\usepackage[safe]{ tipa }
\usepackage{ dsfont }
\usepackage{mathtools} 
\DeclarePairedDelimiter{\floor}{\lfloor}{\rfloor}
\DeclarePairedDelimiter{\ceil}{\lceil}{\rceil}

\usepackage{glossaries} 
\newacronym{LN}{LN}{Lightning Network}
\newacronym{lnd}{lnd}{LND}
\newacronym{cl}{c-lightning}{C-lightning}
\newacronym{eclair}{eclair}{Eclair}
\newacronym{p2p}{P2P}{peer-to-peer}

\usepackage{tikz}
\usetikzlibrary{shapes,arrows,positioning,automata,calc}
\usetikzlibrary{decorations.pathreplacing,calligraphy}
\tikzset{node/.style={circle,draw,inner sep=0pt, minimum size=2em}}
\pgfdeclarelayer{background}
\pgfdeclarelayer{foreground}
\pgfsetlayers{background,main,foreground}
\usetikzlibrary{patterns}

\usepackage{xcolor}

\usepackage{ tipa } 
\usepackage{upgreek} 

\usepackage{enumitem} 

\newcommand{\cref}[1]{{Section~\ref{#1}}}
\newcommand{\mypar}[1]{\medskip\noindent\textbf{#1}}
\usepackage{esvect} 
\usepackage{flushend} 
\usepackage{xspace}

\newcommand{\solution}{ASMR\xspace}
\newcommand{\solutionlong}{Accountable SMR\xspace}
\newcommand{\blockchain}{ZLB\xspace}
\newcommand{\blockchainlong}{Zero-Loss Blockchain\xspace}
\newcommand{\blockchainproblem}{LLB\xspace}
\newcommand{\blockchainlongproblem}{Longlasting Blockchain\xspace}
\newcommand{\component}{BM\xspace}
\newcommand{\componentlong}{Blockchain Manager\xspace}

\newcommand{\myproperty}{active accountability\xspace} 
\newcommand{\Myproperty}{Active accountability\xspace} 
\newcommand{\mypropertyadj}{actively accountable\xspace} 
\newcommand{\Mypropertyadj}{Actively accountable\xspace} 
\newcommand{\problem}{actively accountable consensus\xspace}

\newcommand{\ECHO}{\text{\sc echo}\xspace}
\newcommand{\POF}{\text{\sc pof}\xspace}
\newcommand{\EST}{\text{\sc est}\xspace}
\newcommand{\COORD}{\text{\sc coord}\xspace}
\newcommand{\MSG}{\text{\sc msg}\xspace}
\newcommand{\BVALECHO}{\text{\sc bvecho}\xspace}
\newcommand{\BVALREADY}{\text{\sc bvready}\xspace}
\newcommand{\INIT}{\text{\sc init}\xspace}
\newcommand{\READY}{\text{\sc ready}\xspace}

\algnewcommand{\LeftComment}[1]{\Statex \(\triangleright\) #1}
\newcommand{\markerthree}{\raisebox{0.5pt}{\tikz{\node[scale=0.3,regular polygon, regular polygon sides=3,fill=blue!70,rotate=0](){};}}}
\newcommand{\yes}{\ding{51}}
\newcommand{\no}{\ding{55}}

\newcommand{\mconf}{w\xspace}

\usepackage{balance}
\makeatletter
\newcommand*{\inlineequation}[2][]{%
  \begingroup
    \refstepcounter{equation}%
    \ifx\\#1\\%
    \else
      \label{#1}%
    \fi
    \relpenalty=10000 %
    \binoppenalty=10000 %
    \ensuremath{%
      #2%
    }%
    ~\@eqnnum
  \endgroup
}
\makeatother

\usepackage{pifont} 

\makeatletter
\newcommand{\setword}[2]{%
  \phantomsection
  #1\def\@currentlabel{\unexpanded{#1}}\label{#2}%
}
\makeatother

\AtBeginDocument{%
  \providecommand\BibTeX{{%
    \normalfont B\kern-0.5em{\scshape i\kern-0.25em b}\kern-0.8em\TeX}}}

\title{ZLB, a Blockchain Tolerating Colluding Majorities}

\author{Alejandro Ranchal-Pedrosa \thanks{Protocol Labs, University of Sydney} \\ \texttt{alejandro@protocol.ai}  \and Vincent Gramoli \thanks{Red-Bellow Network, University of Sydney} \\ \texttt{vincent.gramoli@sydney.edu.au}}
\date{\today}

\begin{document}
\maketitle 
\newcommand{\arcomm}[1]{\todo[color=green,bordercolor=black,linecolor=black]{\textsf{\scriptsize\linespread{1}\selectfont ARP: #1}}}
\newcommand{\arcommin}[1]{\todo[inline,color=green,bordercolor=green,linecolor=green]{\textsf{ARP: #1}}}
\newcommand{\system}{[system name] }

\newcommand{\boxedtext}[1]{\fbox{\scriptsize\bfseries\textsf{#1}}}

\newcommand{\greenremark}[2]{
   \textcolor{pinegreen}{\boxedtext{#1}
      {\small$\blacktriangleright$\emph{\textsl{#2}}$\blacktriangleleft$}
    }}

  \definecolor{burntorange}{rgb}{0.8, 0.33, 0.0}
  \definecolor{blue}{rgb}{0.0, 0.0, 0.5}
  \newcommand{\changeremark}[2]{
   \textcolor{burntorange}{\boxedtext{#1}
      {\small$\blacktriangleright$\emph{\textsl{#2}}$\blacktriangleleft$}
}}
\newcommand{\myremark}[2]{
      {  \lowercase{\color{blue}}\boxedtext{#1}
      {\small$\blacktriangleright$\emph{\textsl{#2}}$\blacktriangleleft$}
    }}
  \newcommand{\myremarknew}[2]{
   \textcolor{violet}{\boxedtext{#1}
      {\small$\blacktriangleright$\emph{\textsl{#2}}$\blacktriangleleft$}
}}
\newcommand{\NewRemark}[2]{
   \textcolor{violet}{
      {\small$\blacktriangleright$\emph{\textsl{#2}}$\blacktriangleleft$}
}}

\newcommand{\redremark}[2]{
   \textcolor{red}{\boxedtext{#1}
      {\small$\blacktriangleright$\emph{\textsl{#2}}$\blacktriangleleft$}
    }}

\newcommand\ARP[1]{\myremark{ARP}{#1}}
\newcommand\ARPN[1]{\myremarknew{A}{#1}}
\newcommand\vincent[1]{{\color{red}{VG: #1}}}
\newcommand\vincentR[2]{{\color{red}{#1 \sout{#2}}}}
\newcommand{\warning}[1]{\redremark{\fontencoding{U}\fontfamily{futs}\selectfont\char 66\relax}{#1}}
\newcommand\NEW[1]{\NewRemark{NEW}{#1}}
\newcommand\TODO[1]{\greenremark{TODO}{#1}}
\newcommand\CHANGE[1]{\changeremark{CHANGE(?)}{#1}}

\begin{abstract}
The problem of Byzantine consensus has been key to designing secure distributed systems.
However, it is particularly difficult, mainly due to the presence of Byzantine processes
that act arbitrarily and the unknown message delays in general networks.

In the general setting, consensus cannot be solved if an adversary
controls a third of the system. Yet, blockchain participants typically
reach consensus ``eventually'' despite an adversary controlling a
minority of the system. Exceeding this $\frac{1}{3}$ cap is made
possible by tolerating transient disagreements, where distinct
participants select distinct blocks for the same index, before
eventually agreeing to select the same block.  Until now, no
blockchain could tolerate an attacker controlling a majority of the
system.

Although it is well known that both safety and liveness are at risk as
soon as $n/3$ Byzantine processes fail, very few works attempted to
characterize precisely the faults that produce safety violations from
the faults that produce termination violations.

In this paper, we present a new lower bound on the solvability of the
consensus problem by distinguishing deceitful faults violating safety
and benign faults violating termination from the more general
Byzantine faults, in what we call the Byzantine-deceitful-benign fault
model.  We show that one cannot solve consensus if $n\leq 3t+d+2q$
with $t$ Byzantine processes, $d$ deceitful processes, and $q$ benign
processes.

In addition, we show that this bound is tight by presenting the
Basilic class of consensus protocols that solve consensus when $n >
3t+d+2q$.  These protocols differ in the number of processes from
which they wait to receive messages before progressing. Each of these
protocols is thus better suited for some applications depending on the
predominance of benign or deceitful faults.

Then, we build upon the Basilic class in order to present
\blockchainlong (\emph{\blockchain}), the first blockchain that
tolerates an adversary controlling more than half of the system, with
up to less than a third of them Byzantine.
\blockchain is an open blockchain that combines recent theoretical 
advances in accountable Byzantine agreement to 
exclude undeniably faulty processes.

Interestingly, \blockchain does not need a known bound on the delay of
messages but progressively reduces the portion of faulty processes
below $\frac{1}{3}$, and reaches consensus.  Geo-distributed
experiments show that \blockchain outperforms HotStuff and is almost
as fast as the scalable Red Belly Blockchain that cannot tolerate
$n/3$ faults.
\end{abstract}




\maketitle

\section{Introduction}
Blockchain systems~\cite{Nak08} promise to track ownership of assets
without a central authority and thus rely heavily on distributed nodes
agreeing on a unique block at the next index of the chain.  An
attacker can exploit a disagreement to \emph{double spend} by simply
inserting conflicting transactions in competing blocks.

Some solutions~\cite{BKM18,GAG19,BBC19,CNG21} to this problem avoid forks by
guaranteeing that no 
disagreement can ever occur, even transiently. 
Such solutions typically adopt an open permissioned model where permissionless clients can issue
transactions that $n$ permissioned servers (or \emph{processes}) encapsulate in blocks they agree upon.
These solutions typically assume partial synchrony~\cite{DLS88} or that there exists an unknown bound on the time it takes to deliver any message.
Unfortunately, it is well-known~\cite{PSL80} that consensus cannot be solved as soon as $\frac{1}{3}$ of these processes experience a Byzantine fault. More specifically, in these blockchains an attacker can exploit a disagreement to double spend if it controls $\frac{1}{3}$ of these processes.

Other solutions, popularized by classic
blockchains~\cite{Nak08,Woo15,GKL15,CS20}, assume that the adversary
controls only a minority of the processes, typically expressed as
computational power or stake.  The tolerance to an adversary
controlling more than $\frac{1}{3}$ but less than $\frac{1}{2}$ of the
processes is made possible by accepting forks and tolerating transient
disagreements that eventually get resolved in an ``eventual''
consensus.  Unfortunately, as soon as the adversary controls a
majority of the system, then safety gets violated.

It has recently been shown that blockchains that provide deterministic guarantees must solve consensus~\cite{ADT}. The problem of Byzantine consensus has been key to designing secure
distributed systems~\cite{singh2009zeno,CKL09,KWQ12,LVC16}.
This problem is particularly difficult to solve because a Byzantine participant acts arbitrarily~\cite{LSP82} and message delays are generally unpredictable~\cite{DLS88}. 
Any consensus protocol would fail in this general setting if the number of Byzantine participants is $t\geq n/3$~\cite{DLS88}, where $n$ is the total number of participants.
In some executions, $\lceil n/3\rceil$ Byzantine participants can either prevent the termination of the consensus protocol by stopping or by sending unintelligible messages.
In other executions, $\lceil n/3\rceil$ can violate the agreement property of the consensus protocol by sending conflicting messages.

Interestingly, various research efforts were devoted to increase the fault tolerance of consensus protocols in closed networks (e.g., datacenters) 
by distinguishing the type of failures~\cite{CKL09,KWQ12,LVC16,LLM19}.
Some works overcome the $t<n/3$ bound by tolerating a greater number of omission than commission faults~\cite{singh2009zeno,CKL09}. 
These works are naturally well-suited for closed networks where processes are protected from intrusions by a firewall: their processes are supposedly 
more likely to crash than to be corrupted by a malicious adversary. 
In this sense, these protocols favor tolerating a greater number of faults for liveness than for safety.

Unfortunately, fewer research efforts were devoted to explore the fault tolerance of consensus protocols in open networks (e.g., blockchains).
In such settings, participants are likely to cause a disagreement if they can steal valuable assets.
This is surprising given that attacks are commonplace in blockchain systems as illustrated by 
the recent losses of $\mathdollar 70,000$\footnote{\href{https://news.bitcoin.com/bitcoin-gold-51-attacked-network-loses-70000-in-double-spends/}{https://news.bitcoin.com/bitcoin-gold-51-attacked-network-loses-70000-in-double-spends/}} and $\$18$ million\footnote{\href{https://news.bitcoin.com/bitcoin-gold-hacked-for-18-million/}{https://news.bitcoin.com/bitcoin-gold-hacked-for-18-million/}} in Bitcoin Gold, and of $\mathdollar 5.6$ million in Ethereum Classic\footnote{\href{https://news.bitcoin.com/5-6-million-stolen-as-etc-team-finally-acknowledge-the-51-attack-on-network/}{https://news.bitcoin.com/5-6-million-stolen-as-etc-team-finally-acknowledge-the-51-attack-on-network/}}.
Comparatively, some blockchain participants, called miners, are typically monitored continuously so as to ensure they provide some rewards to their owners, hence making it less likely to prevent termination.
To our knowledge, only alive-but-corrupt (abc) processes~\cite{MNR19} characterize the processes that violate consensus safety. 
Unfortunately, abc processes are restricted to only try to cause a
disagreement if the coalition size is sufficiently large to succeed at
the attempt, which is impossible to predict in blockchain systems.

\mypar{Our result.} To this end, we present the
\emph{Byzantine-deceitful-benign} (BDB) failure model, introducing two
new types of processes, characterized by the faults they commit.
First, a \emph{deceitful} process is a process that sends some
conflicting messages (messages that contribute to a violation of
agreement).  Second, a \emph{benign} process is a faulty process that
never sends any conflicting messages, contributing to
non-termination. For example, a benign process can crash or send stale
messages, or even equivocate as long as its messages have no effect on
agreement. These two faults lie at the core of the
consensus problem.

We present a new lower bound on the solvability of the Byzantine
consensus problem by precisely exploring these two additional types of
faults (that either prevent termination or agreement when $t\geq
n/3$). Our lower bound states that there is no protocol solving
consensus in the partially synchronous model~\cite{DLS88} if $n\leq
3t+d+2q$ with $t$ Byzantine processes, $d$ deceitful processes, and
$q$ benign processes.

Furthermore, we show that this lower bound is tight, in that we
present the Basilic\footnote{The name ``Basilic'' is inspired from the
Basilic cannon that Ottomans used to break through the walls of
Constantinople. Much like the cannon, our Basilic protocol provides a
tool to break through the classical bounds of Byzantine fault
tolerance.}  class of protocols that solves consensus with $n>
3t+d+2q$. Basilic builds upon recent advances in the context of accountability~\cite{CGG21} by taking into account key messages only if they are cryptographically 
signed by their sender. If they are properly signed, the recipient stores these messages and progresses in the consensus protocol execution.
Recipients also cross-check the messages they received with other recipients, based on the assumption that signatures cannot be forged.
Once conflicting messages are detected, 
they constitute an undeniable proof of fraud to exclude the faulty sender before continuing the protocol execution. 
Thanks to this exclusion, Basilic satisfies a new property, \textit{\myproperty}, which guarantees that deceitful processes can not prevent termination. 


Basilic is a class of consensus protocols, each parameterized by a
different \emph{voting threshold} or the number of distinct processes
from which a process receives messages in order to progress. For a
voting threshold of $h\in(n/2,n]$, Basilic satisfies termination if
$h\leq n-q-t$, and agreement if $h>\frac{d+t+n}{2}$. This means that
for just one threshold, say $h=2n/3$, Basilic tolerates multiple
combinations of faulty processes: it can tolerate $t<n/3,\,q=0$ and
$d=0$; but also $t=0,\,q<n/3$ and $d<n/3$; or even $t<n/6,\,q<n/6$ and
$d<n/6$. This voting threshold can be modified by an application in
order to tolerate any combination of $t$ Byzantine, $d$ deceitful and
$q$ benign processes satisfying $n>3t+d+2q$.

The generalization of Basilic to any voting threshold $h$ thus allows
us to pick the best suited protocol depending on the application
requirements. If, on the one hand, the application runs in a closed
network (e.g., datacenter) dominated by benign processes, then the
threshold will be lowered to ensure termination. If, on the other
hand, the application runs in an open network (e.g., blockchain)
dominated by deceitful processes, then the threshold will be raised to
ensure agreement.

After presenting Basilic, we propose the \emph{\blockchainlong} (or
\emph{\blockchain} for short), the first blockchain that
simultaneously solves consensus in the presence of a Byzantine
adversary controlling up to less than a third of the processes and eventually
solves consensus in the presence of an adversary controlling more than
half of the processes with the purpose of causing a disagreement, both
without assuming synchrony. The key breakthrough of \blockchain is
that it falls back to eventual consensus only for a finite amount of
time, after which consensus can be solved again until the adversary
changes the processes it corrupts. We call this property
convergence. More specifically, \blockchain solves consensus for
$t+d<n/3$ while also falling back to the eventual consensus problem
only in a bounded amount of unlucky cases where $n/3 \leq t+d < 2n/3$.

To demonstrate the efficiency of \blockchain we implement it with
Bitcoin transactions, and compare its performance to modern blockchain
systems. We show that, on 90 machines spread across distinct
continents, \blockchain outperforms by $5.6$ times the
HotStuff~\cite{YMR19} state machine replication that inspired Facebook
Libra's~\cite{BBC19}, and obtains comparable performance to the recent
Red Belly Blockchain~\cite{CNG21}. Our empirical results also show an
interesting phenomenon in that the impact of the attacks decreases
rapidly as the system size increases, due to the increased message
delays.

Furthermore, We also develop a Zero-Loss Payment application on top of
\blockchain, in which we guarantee that the financial losses from
potential disagreements caused by attackers are cancelled out by the
deposit taken from the same attackers to fund the effects of the
disagreement.

\mypar{Summary.} In summary, in this work:
\begin{enumerate}[label=\roman*)]
\item We present the novel BDB failure model, compare it with previous models and justify it.
\item We extend the classical impossibility bounds of Byzantine fault-tolerant (BFT) consensus to the BDB model.
\item We introduce the Basilic class of consensus protocols that we prove resilient-optimal in both the BFT and the BDB model.
\item We show that protocols of the Basilic class are optimal in terms of the communication complexity.
\item We introduce the \blockchainlongproblem (\blockchainproblem) problem designed to solve the blockchain problem in situations where the adversary can cause a disagreement.
\item We present the \blockchainlong (\blockchain), the first blockchain to solve \blockchainproblem, that uses the Basilic class.
\item We build \blockchain and compare its performance with the state of the art, showing that it is faster than Facebook's Libra blockchain, and competitive with the recent Red-Belly blockchain that is not accountable.
\item We build a zero-loss payment application on top of \solution in which no honest process or client loses any funds resulting from disagreement attacks.
\end{enumerate}

\mypar{Roadmap.} Section~\ref{sec:relw} illustrates the related work. In Section~\ref{sec:model} we introduce our BDB model
and other assumptions. Section~\ref{sec:imp} shows the new
impossibility bounds of consensus in the BDB model. In
Section~\ref{sec:prot} we present the Basilic class of protocols,
prove its correctness and complexities, and that it also solves eventual consensus. Section~\ref{sec:zlblockchain} presents the
\blockchainproblem problem and \blockchain, shows \blockchain's
correctness and proofs and its experimental evaluation with previous
works. Section~\ref{sec:zlbpayment} shows the zero-loss payment
application in which no honest process or client loses any fund
resulting from temporary disagreements. We finally conclude in
Section~\ref{sec:zlbsum}.

\section{Related Work}
\label{sec:relw}

\mypar{Consensus.} Accountability has been proposed for distributed systems in
PeerReview~\cite{HKD07} and particularly for the problem of consensus
in Polygraph~\cite{CGG20}. This work leverages accountability to
replace deceitful processes by new processes. Unfortunately, they
require deceitful processes to eventually
stop trying to cause a disagreement. Flexible BFT~\cite{MNR19} offers a failure model and theoretical results to
tolerate $\lceil 2n/3\rceil -1$ alive-but-corrupt (abc) processes. Their fault tolerance requires a commitment
from clients to not tolerate a single Byzantine fault in order to
tolerate $\lceil 2n/3\rceil -1$ abc faults, or to instead tolerate no
abc faults if clients decide to tolerate $t=\lceil n/3\rceil -1$ Byzantine faults.
Neu et al.'s ebb-and-flow system~\cite{ebbnflow} is available in
partial synchrony for $t<n/3$ and satisfies finality in synchrony for
$t<n/2$. They also motivate the need for a model like BDB in their recent
availability-accountability
dilemma~\cite{neu2021availability}. Sheng et al.\cite{forensics}
characterize the forensic support of a variety of Blockchains. Unfortunately,
none of these works tolerate $q=\ceil{\frac{n}{3}}-1$ benign and even
$d=1$ deceitful faults, or $d=\ceil{\frac{n}{3}}-1$ and even $q=1$
benign fault, a direct consequence of them not satisfying \myproperty.

Upright~\cite{CKL09} tolerates $n=2u+r+1$ faults, where $u$ and $r$
are the numbers of commission and omission faults,
respectively. Upright tolerates $n/3$ commission faults or instead
$n/2$ omission faults, falling short of Basilic's $q+d<2n/3$
deceitful and benign faults or $t<n/3$ Byzantine faults
tolerated. Upright does also not tolerate more faults for commission
than the lower bound for BFT consensus. Anceaume et
al.~\cite{anceaume2020finality} tolerate $t<n/2$ Byzantine faults for
the problem of eventual consensus, at the cost of not tolerating even
$t=1$ Byzantine fault for deterministic consensus. Our Basilic class also tolerates this case if the initial voting threshold $h_0$ is set to $h_0=\floor{\frac{n}{2}}+1$, being this part of the Basilic class. 

Basilic is, to the best of our knowledge, the first protocol
tolerating $n>3t+d+2q$ in the BDB model, thanks to the property of
\myproperty. However, previous works already try to discourage
misbehavior by threatening with slashing a deposit or removing a
faulty process from the committee, or both. Shamis et al.~\cite{shamis2021pac} store signed
messages in a dedicated ledger so as to punish processes in case of
misbheavior. The Casper~\cite{Casper} algorithm
incurs a penalty in case of double votes but does not ensure
termination when $t<n/3$. Although Tendermint~\cite{BKM18} aims at slashing processes, it is not accountable. SUNDR~\cite{LKM04} requires cross-communication between non-faulty clients
to detect failures. FairLedger~\cite{levari2019fairledger} requires synchrony to detect faulty processes. Balance~\cite{HGG19} adjusts the size of the deposit to avoid over collateralizing but we are not aware of any system that implements it.
Polygraph~\cite{CGG21} solves accountable consensus without slashing. FairLedger~\cite{levari2019fairledger} assumes synchrony in order to detect faulty processes. Sheng et al.~\cite{forensics} consider forensics support as the ability to make processes accountable for their actions to clients. We do not consider this model in which they show accountability cannot be achieved with $2n/3$ faults.
Freitas de Souza et al.~\cite{SKRP21} provide an asynchronous implementation of an accountable lattice agreement protocol that reconfigures processes in a lattice agreement after detection.
Shamis et al.~\cite{shamis2021pac} propose, in a concurrent work with ours, 
to store signed messages in a dedicated ledger so as to punish processes 
in case of misbehavior. It however needs $2n/3$ honest 
processes to guarantee progress. 


\mypar{Blockchains.} Several works tried to circumvent the upper bound on the number of Byzantine failures~\cite{LSP82} to reach agreement.
As opposed to permissionless blockchains~\cite{Nak08}, open permissioned blockchains
try to rotate the consensus participants to cope with an increasing amount of colluding processes 
without perfect synchrony~\cite{VG19b,BAS20}. The notion of accountability has originally been applied to distributed systems in PeerReview~\cite{HKD07} and to consensus protocols in Polygraph~\cite{CGG20}, however, not to recover from inconsistencies.

%
 
\mypar{Fault models.} Traditionally, closed distributed systems consider that omission faults (omitting messages) are more frequent than commission faults (sending wrong messages)~\cite{CKL09,KWQ12,LLM19}.
Zeno~\cite{singh2009zeno} guarantees eventual consistency by decoupling requests into weak (i.e., requests that may suffer reordering) and strong requests. 
\blockchain could not be built upon Zeno because Zeno requires wrongly ordered transactions to be rolled back, whereas blockchain transactions can have irrevocable side effects like the shipping of goods to the buyer.
BFT2F~\cite{Li} offers fork* consistency, which forces the adversary to keep correct clients in one fork, while also allowing accountability. Stewart et al.~\cite{grandpa} provide a finality gadget similar to our confirmation phase, however, it does not recover from disagreements. Flexible BFT's abc processes~\cite{MNR19} behave maliciously only if they know they can violate
safety, and correctly otherwise. This is an stronger
assumption than our deceitful faults. 

The BAR model~\cite{AAC05}, and its Byzantine, altruistic, rational classification 
is motivated by multiple administrative domains and 
corresponds better to the blockchain open networks without distinguishing a-b-c faults. Various non-blockchain systems already refined the types of failures to strengthen guarantees.
Upright~\cite{CKL09} proposes a system that supports $n=2u+r+1$
faults, where $u$ and $r$ are the numbers of commission and omission
faults, respectively. They can either tolerate $n/3$ commission faults
or $n/2$ omission faults. Ranchal-Pedrosa et
al.~\cite{trap2022ranchalpedrosa} propose a game theoretical model and
a consensus protocol tolerating up to
$n>\max\left(\frac{3}{2}k+3t,2(k+t)\right)$ where $t$ are the number
of Byzantine and $k$ the number of rational processes, estimating
utilities for rational processes by the gain they would get from
causing a disagreement in a blockchain application. This result is
however subject to rational processes being bounded by these
utilities, whereas deceitful processes will always be interested in
causing a disagreement.

Some hybrid failure models tolerate crash failures and Byzantine
failures but prevent Byzantine failures from partitioning the
network~\cite{LVC16}. Others aim at guaranteeing that well-behaved
quorums are responsive~\cite{LLM19} or combine crash-recovery with
Byzantine behaviors to implement reliable broadcast~\cite{BC03}.
%
\section{Model}
\label{sec:model}
We consider a committee as a set $N=\{p_0, ..., p_{n-1}\}$ of $|N|=n$
processes. These processes communicate in a partially synchronous
network, meaning there is a known bound $\Delta$ on the communication
delay that will hold after an unknown Global Stabilization Time
(GST)~\cite{DLS88}. Processes communicate through standard all-to-all
reliable and authenticated communication channels~\cite{kuznesov2021},
meaning that messages can not be duplicated, forged or lost, but they
can be reordered. A process that follows the protocol is
\emph{honest}. Faulty processes can be Byzantine, deceitful or benign,
as we detail in Section~\ref{sec:bdbmodel}. We will be dealing with
two bounds for fault tolerance, that we denote
$t_\ell=\ceil{\frac{n}{3}}-1$ and $t_s=\ceil{\frac{2n}{3}}-1$.
Processes communicate through private pairwise channels.

\subsection{Cryptography} We assume a public-key infrastructure (PKI)
in that each party has a public key and a private key, and any
party’s public key is known to all~\cite{xue2021}. As with other protocols that use this standard assumption~\cite{xue2021,abraham2021reach}, we do not require the use of revocation lists (we will remove processes from the committee, but not their keys from the PKI). We refer to
$\lambda$ as the security parameter, i.e., the number of bits of the
keys. As our claims and proofs require cryptography, they hold except
with $\epsilon(\lambda)$ negligible
probability~\cite{backes2003reliable}.  We formalize negligible
functions measured in the security parameter $\lambda$, which are
those functions that decrease asymptotically faster than the inverse
of any polynomial. Formally, a function $\epsilon(\kappa)$ is
negligible if for all $c>0$ there exists a $\kappa_0$ such that
$\epsilon(\kappa) < 1/\kappa^c$ for all $\kappa >
\kappa_0$~\cite{backes2003reliable}.

\subsection{Adversary} We model processes as probabilistic
polynomial-time interactive Turing machines
(ITMs)~\cite{Dumbo2020,cachin2005random,cachin2001}. A process is an
ITM defined by the following protocol: it is activated upon receiving
an incoming message to carry out some computations, update its states,
possibly generate some outgoing messages, and wait for the next
activation. The adversary $\mathcal{M}$ is a probabilistic ITM that
runs in polynomial time (in the number of message bits generated by
honest parties). $\mathcal{M}$ can control the network to read or
delay messages, but not to drop them. It can also take control and
corrupt a coalition of processes, learning its entire state
(stored messages, signatures, etc.). It takes control of receiving and
sending all their messages. Furthermore, it can deliver the messages
from honest processes and users instantly, and its messages can be
delivered instantly by any honest process or user.

\subsection{Send, receive and deliver} Messages can
be sent and received, but we also consider broadcast primitives that
contain two functions: a broadcast function that allows process $p_i$ to
send messages through multiple channels accross the network, and a deliver
function that is invoked at the very end of the broadcast primitive to
indicate that the recipient of the message has received and processed
the message. There could be however multiple message
exchanges before the delivery can happen. As we will specify some of
these broadcast primitives, we attach the name of the protocol as a
prefix to the broadcast and deliver function to refer to a message
broadcast or delivered using that protocol, such as AARB-broadcast,
AARB-deliver, ABV-broadcast and ABV-deliver, as we detail later in
this work.

\subsection{Consensus}
A protocol executed by a committee of processes solves the consensus
problem if the following three properties are satisfied by the
protocol:
\begin{itemize}
\item {\bf Termination.} Every non-faulty process eventually decides on a value.
\item {\bf Agreement.} No two non-faulty processes decide on different values.
\item {\bf Validity.} If all non-faulty processes propose the same value, no other value can be decided.
\end{itemize}

\mypar{Blockchains.} A blockchain system~\cite{Nak08} is a distributed system maintaining a sequence of blocks that contains \emph{valid} (cryptographically signed) and non-conflicting transactions indicating how assets are exchanged between \emph{accounts}.\footnote{Note that in~\cref{sec:solution} we will implement Bitcoin's transactions where ``valid'' implies ``non-conflicting'' as requested transactions cannot be valid if their UTXOs are already  consumed.}

\subsection{Byzantine state machine replication} A Byzantine State Machine Replication (SMR)~\cite{CL02,KADC07} is a replicated service that accepts deterministic commands from clients and totally orders these commands using a consensus protocol so that, upon execution of these commands, every honest process ends up with the same state despite \emph{Byzantine} or faulty processes. The instances of the consensus execute in sequence, one after the other, starting from index 0. We refer to the consensus instance at index $i$ as $\Gamma_i$.
%

Traditionally, 
given that honest processes propose a value, the Byzantine consensus problem~\cite{PSL80} is for every honest process to eventually decide a value (consensus termination), for no two honest processes to decide different values (agreement) and for the decided value to be one of the proposed values (validity).
In this paper, we consider however a variant of Byzantine consensus (Def.~\ref{def:sbc}) %
useful for blockchains~\cite{MSC16,DRZ18,CNG21} where 
the validity requires the decided value to be a subset of the union of the proposed values, hence allowing us to commit more proposed blocks per consensus instance.
 
\begin{definition}[Set Byzantine Consensus]\label{def:sbc}
Assuming that each honest process proposes a set of transactions, the \emph{Set Byzantine Consensus} (SBC) problem is for each of them to decide on a set in such a way that the following properties are satisfied:
\begin{itemize}
\item SBC-Termination: every honest process eventually decides a set of transactions;
\item SBC-Agreement: no two honest processes decide on different sets of transactions;
\item SBC-Validity: a decided set of transactions is a non-conflicting set of valid transactions taken from the union of the proposed sets; 
\item SBC-Nontriviality:  
if all processes are honest and propose a common valid non-conflicting set of transactions, then this set is the decided set. 
\end{itemize}
\end{definition}
SBC-Termination and SBC-Agreement are common properties to many Byzantine consensus definition variants, while 
SBC-Validity 
states that 
transactions proposed by Byzantine proposers could be decided as long as they are valid and \emph{non-conflicting} (i.e., they do not withdraw more assets from one account than its balance); and SBC-Nontriviality is necessary to prevent trivial algorithms that decide a pre-determined value from solving the problem. 
%
As a result, we consider that a consensus instance $\Gamma_i$ outputs a set of enumerable decisions $out(\Gamma_i)=d_i,\; |d_i|\in\mathds{N}$ that all  $n$ processes replicate. We refer to the state of the SMR at the $i$-th consensus instance $\Gamma_i$ as all decisions of all instances up to the $i$-th consensus instance. %



\subsection{Accountability}
The processes of a blockchain system are, by default, not accountable in that their faults often go undetected. For example, when a process creates a fork, it manages to double spend after one of the blockchain branches where it spent coins vanishes. This naturally prevents other processes from detecting frauds and from holding this process accountable for its misbehavior.
%
%
Recently, Polygraph~\cite{CGG20,CGG21} 
introduced accountable consensus (Def.~\ref{def:accountability})
as the problem of solving consensus if $f<n/3$ and eventually detecting $f_d \geq n/3$ faulty processes in the case of a disagreement.

\begin{definition}[Accountable Consensus]\label{def:accountability}
The problem of \emph{accountable consensus} is: (i) to solve consensus if the number of Byzantine faults is $f<n/3$, and (ii) for every honest process to eventually output at least $f_d\geq n/3$ faulty processes if two honest processes output distinct decisions.
\end{definition}
  \subsection{Byzantine-deceitful-benign fault model}
  \label{sec:bdbmodel}
We introduce in this section formal definitions needed for our novel Byzantine-deceitful-benign (BDB) fault model.

\mypar{Conflicting messages.} Basilic detects and removes faulty processes that try to cause a disagreement, even if they do not succeed at causing the disagreement. For this reason, Basilic must be able to detect processes that send distinct messages to different processes where they
were expected to broadcast the same message to different processes~\cite{abraham2021}, we
refer to these messages as conflicting. Given a protocol $\sigma$, we say that a message, or set of messages, $\ms{msg}$ sent by process $p$ \textit{conforms} to an execution $\sigma_E$ of the protocol $\sigma$, if $\sigma_E$ belongs to the set of all possible executions where $p$ sent $m$ and $p$ is an honest process. Also, a faulty process $p$ sending two messages $\ms{msg},\ms{msg}'$ \textit{contributes} to a disagreement if there is an execution $\sigma_E$ of $\sigma$ such that (i) sufficiently many faulty processes sending $\ms{msg},\,\ms{msg}'$ (and possibly more messages) to a disjoint subset of honest processes, one to each, leads to a disagreement, and (ii) $\sigma_E$ does not lead to a disagreement without $p$ sending $\ms{msg},\,\ms{msg}'$. Two messages $\ms{msg},\, \ms{msg}'$ are \textit{conflicting} with respect to $\sigma$ if:
\begin{enumerate}
\item $\ms{msg},\,\ms{msg}'$ individually conform to algorithm $\sigma$
  for some execution $\sigma_E$, $\sigma_{E'}$, respectively, $\sigma_E\neq \sigma_{E'}$,
\item there is no execution $\sigma_{E''}$ of $\sigma$
such that both messages together conform to $\sigma_{E''}$, and
\item if $p$ sending $\ms{msg},\ms{msg}'$ to a disjoint subset of honest 
processes, one to each, contributes to a disagreement.
\end{enumerate}

When combined in one message and signed by the sender, conflicting
messages constitute a Proof-of-Fraud (PoF), thanks to 
accountability. An
example of two conflicting messages is a faulty process sending two
different proposals for the same round (the proposer should only
propose one value per round).

Our definition of conflicting messages
differs from previous similar concepts in that conflicting messages
allow for any process $p$ to verify if two messages are conflicting: an
honest process can always construct a PoF from two conflicting
messages alone, but it cannot do so with all mutant
messages~\cite{KLM03}, as $p$ would need to also learn the entire
execution, or with messages sent from an equivocating
process~\cite{CFAR12}, as these do not necessarily contribute to
disagreeing.

\mypar{Fault model.} There are three mutually exclusive classes of
faulty processes: Byzantine, deceitful and
benign, in what we refer to as the
\textit{Byzantine-deceitful-benign} (BDB) failure model. Each faulty
process belongs to only one of these classes. Byzantine, deceitful and
benign processes are characterized by
the faults they can commit. A fault is \textit{deceitful} if it
contributes to breaking agreement, in that it sends conflicting
messages violating the protocol in order to lead two or more
partitions of processes to a disagreement. We allow deceitful
processes to constantly keep sending conflicting messages, even if
they do not succeed at causing a disagreement, but instead their
deceitful behavior prevents termination. As deceitful processes model
processes that try to break agreement, we assume also that a deceitful
fault does not send conflicting messages for rounds or phases of the
protocol that it has already terminated at the time that it sends the
messages. Deceitful processes can alternate between sending
conflicting messages and following the protocol, but cannot deviate in
any other way. A \textit{benign} fault is any fault that does not ever
send conflicting messages. Hence, benign faults cover only faults that
can break termination, e.g. by crashing, sending stale messages, etc.

As usual, Byzantine processes can act arbitrarily. Thus, Byzantine
processes can commit benign or deceitful faults, but they can also
commit faults that are neither deceitful nor benign. A fault that
sends conflicting messages and crashes afterwards is, by these
definitions, neither benign nor deceitful. We denote $t,\,d,$ and $q$
as the number of Byzantine, deceitful, and benign processes,
respectively. We assume that the
adversary is static, in that the adversary can choose up to $t$
Byzantine, $d$ deceitful and $q$ benign processes at the start of the
protocol, known only to the adversary. The total number of faulty
processes is thus $f=t+d+q$.

In order to distinguish benign (resp. deceitful) processes
from Byzantine processes that commit a benign (resp. deceitful) fault
during a particular execution of a protocol, we formalize
fault tolerance in the BDB model. Let $E_\sigma(t,d,q)$ denote the set
of all possible executions of a protocol $\sigma$ given that there are
up to $t$ Byzantine, $d$ deceitful and $q$ benign processes. We say
that a protocol $\sigma$ for a particular problem $P$ is
\textit{$(t,d,q)$-fault-tolerant} if $\sigma$ solves $P$ for all
executions $\sigma_E\in E_\sigma(t,d,q)$. We abuse notation by
speaking of a $(t,d,q)$-fault-tolerant protocol $\sigma$ as a protocol
that tolerates $t,\,d$ and $q$ Byzantine, deceitful and benign
processes, respectively.

Note that, given a protocol $\sigma$, then $E_\sigma(0,d+k,q)\subset
E_\sigma(k,d,q)$ by definition. Thus, if $\sigma$ is
$(k,d,q)$-fault-tolerant then $\sigma$ is $(0,d+k,q)$-fault tolerant, and also $(0,d,q+k)$-fault-tolerant. However, the contrary is not
necessarily true: a protocol $\sigma$ that is
$(0,d+k,q)$-fault-tolerant is not necessarily $(k,d,q)$-fault
tolerant, as $E_\sigma(k,d,q) \nsubseteq E_\sigma(0,d+k,q)$, because
Byzantine participants can commit more faults than deceitful or
benign. 

Compared to commission and omission faults, notice that not all
commission faults contribute to causing disagreements. For example,
some commission faults broadcast an invalid message that can be
discarded. In our BDB model, this type of fault would categorize as
benign, and not deceitful, since invalid messages never contribute to
a disagreement, but can instead prevent termination (by only sending
invalid messages that are discarded). All omission faults are however
benign faults, while the contrary is also not true (as per the same
aforementioned example). Compared to the alive-but-corrupt failure
model~\cite{MNR19}, deceitful faults are not restricted to only contribute to a
disagreement if they know the disagreement will succeed, but instead 
we let them try forever, even if they do not succeed. This means that
while a protocol might tolerate $d<n/3$ abc faults along with $q<n/3$
benign faults, it would not necessarily tolerate $d<n/3$ deceitful
faults along with $q<n/3$ benign faults. The contrary direction always
holds.

We believe thus the BDB model to be better-suited for consensus, as it
establishes a clear difference in the types of faults depending on the
type of property that the fault jeopardizes (agreement for deceitful,
termination for benign), without restricting the behavior of these
faults to the cases where they are certain that they will cause a
disagreement. We restate that the property of validity is defined only
to rule out trivial solutions of consensus in which all processes decide
a constant, and this property can be locally checked for correctness.

\section{Impossibility of consensus in the BDB model}
\label{sec:imp}

In this section, we extend Dwork et al.'s impossibility results~\cite{DLS88} on the
number of honest processes necessary to solve the Byzantine
consensus problem in partial synchrony by adding deceitful and benign processes. First, we prove
in Section~\ref{sec:impgen} lower bounds on the size of the
committee of any consensus protocol. Then,  we prove in
Section~\ref{sec:impthre} lower bounds depending on the
voting threshold of that protocol, which we define in the same
section.

\subsection{Impossibility bounds}
\label{sec:impgen}
First, we consider the case where $t=0$, i.e., there are only
deceitful and benign processes. In particular, we show in
Lemma~\ref{lem:imp} that if a protocol solves consensus then it
tolerates at most $d<n-2q$ deceitful processes and $q< n/2$ benign
processes. The intuition for the proof is analogous to the classical
impossibility proof of consensus in partial synchrony in the presence
of $t_\ell+1$ Byzantine processes. Lemma~\ref{lem:imp} extends to the 
BDB model the classical lower bound for the BFT model~\cite{DLS88}, by
tolerating a stronger adversary than the classical bound (e.g. an
adversary causing $d=t_\ell$ deceitful faults and
$q=t_\ell$ benign faults). By contradiction, we show that in the
presence of a greater number of faulty processes than bounded by
Lemma~\ref{lem:imp}, in some executions all processes would either not
terminate, or not satisfy agreement, if maintaining validity.

\begin{lemma}
  \label{lem:imp}
  Let a protocol $\sigma$ and let $\sigma$ solve consensus for all executions $\sigma_E\in E_\sigma(t,d,q)$ for some $t,d,q>0$. Then, $d+t< n-2(q+t)$.
\end{lemma}
\begin{proof}

  First, we show $q<n/2$ by contradiction, as done by previous work
for omission faults~\cite{DLS88}. Suppose $q\geq n/2,\,d=0,\,t=0$ and
consider processes are divided into a disjoint partition $P,Q$ such
that $P$ contains between $1$ and $q$ processes and $Q$ contains
$n-|P|$. First, consider scenario A: all processes in $P$ are benign
and the rest honest, and all processes in $Q$ propose value
$0$. Then, by validity all processes in $Q$ decide $0$. Then, consider
scenario B: all processes in $Q$ are benign and the rest honest, and
all processes in $P$ propose value $1$. Then, by validity all
processes in $P$ decide $1$. Now consider scenario C: no process is
benign, and processes in $P$ propose all $1$ while processes in $Q$
propose all $0$. For processes in $P$ scenario C is
indistinguishable from scenario B, while for processes in $Q$ scenario
C is indistinguishable from scenario A. This yields a
contradiction.

It follows that $q<n/2$. Hence, for $n=2$, and since $q<1$, it is
immediate that for $d+t\geq 2$ it is impossible to solve consensus. As
such, we have left to consider $d+t\geq n-2(q+t)$ with $n\geq 3$. We will
prove this by contradiction.

Consider processes are divided into three disjoint partitions $P, Q, R$, such
that $P$ and $Q$ contain between $1$ and $q+t$ processes each, and $R$
contains between $1$ and $d+t$. First consider the following scenario A:
processes in $P$ and $R$ are honest and propose value $0$, and
processes in $Q$ are benign. It follows that $P\cup R$ must decide
value $0$ at some time $T_A$, for if they decided $1$ there would be a
scenario in which processes in $Q$ are honest and also propose $0$,
but messages sent from processes in $Q$ are delivered at a time
greater than $T_A$, having processes in $P\cup R$ already decided
$1$. This would break the validity property. Also, they must decide
some value to satisfy termination tolerating $q+t$ benign faults.

Consider now scenario B: processes in $P$ are benign, and processes in
$R$ and $Q$ are honest and propose value $1$. By the same approach,
$R\cup Q$ decide $1$ at a time $T_B$.

Now consider scenario C: processes in $P$ and $Q$ are honest, and
processes in $R$ are deceitful, the messages sent from
processes in $Q$ are delivered by processes in $P$ at a time greater
than $\max(T_A,T_B)$, and the same for messages sent from processes in
$P$ to processes in $Q$. Processes in $P$ propose $0$, processes in $Q$ propose $1$, and processes in $R$ propose $0$ to those in $P$ and $1$ to those in $Q$. Then, for processes in $P$ this scenario is
identical to scenario A, deciding $0$, while for processes in $Q$ this
is identical to scenario B, deciding $1$, which leads to a
disagreement. This yields a contradiction.
\end{proof}

\begin{corollary}[Impossibility of consensus with $t=0$]
  \label{cor:impct0}
  It is impossible for a consensus protocol $\sigma$ to tolerate $d$ deceitful and $q$ benign processes if $d\geq n-2q$ or $q\geq n/2$.
\end{corollary}
\begin{proof}
  This is immediate from Lemma~\ref{lem:imp} since $\sigma$ is $(0,d,q)$-fault-tolerant if $\sigma$ solves $P$ for all
executions $\sigma_E\in E_\sigma(0,d,q)$. 
\end{proof}

We prove the impossibility result of Theorem~\ref{thm:imp} by extending the result of Corollary~\ref{cor:impct0}: it is
impossible to solve consensus in the presence of $t$ Byzantine, $q$
benign and $d$ deceitful processes unless $n>3t+d+2q$.
\begin{theorem}[Impossibility of consensus]
  \label{thm:imp}
      It is impossible for a consensus protocol to tolerate $t$ Byzantine, $d$ deceitful and $q$ benign processes if $n\leq 3t+d+2q$.
\end{theorem}
\begin{proof}
    This is immediate from Lemma~\ref{lem:imp} since $\sigma$ is $(t,d,q)$-fault-tolerant if $\sigma$ solves $P$ for all executions $\sigma_E\in E_\sigma(t,d,q)$.
\end{proof}

\subsection{Impossibility bounds per voting threshold}
\label{sec:impthre}
The proofs for the impossibility results of Section~\ref{sec:impgen}
(and for the classical impossibility results~\cite{DLS88}) derive a
trade-off between agreement and termination. In some scenarios,
processes must be able to terminate without delivering messages from a
number of processes that may commit benign faults. In other scenarios,
processes must be able to deliver messages from enough processes
before terminating in order to make sure that no disagreement caused
by deceitful faults is possible. We prove in this section the
impossibility results depending on this trade-off.

A protocol that satisfies both
agreement and termination in partial synchrony must thus state a
threshold that represents the number of processes from which to
deliver messages in order to be able to terminate without compromising
agreement. If this threshold
is either too small to satisfy
agreement, or too large to satisfy termination, then the protocol does
not solve consensus. We refer to this threshold as the \textit{voting
threshold}, and denote it with $h$. Typically, this threshold is
$h=n-t_\ell= \ceil{\frac{2n}{3}}$ to tolerate $t_\ell=\ceil{\frac{n}{3}}-1$
Byzantine faults~\cite{crain2018dbft,CGG21, KADC07, YMR19}. We prove however in Lemma~\ref{lem:impagr} and
Corollary~\ref{cor:impagr} that $h>\frac{d+t+n}{2}$ with $h\in(n/2,n]$
for safety.

\begin{lemma}[Impossibility of Agreement ($t=0$)]
  \label{lem:impagr}
  Let $\sigma$ be a protocol with voting threshold $h\in(n/2,n]$ that satisfies agreement. Then $\sigma$ tolerates at most $d<2h-n$ deceitful processes.
\end{lemma}
\begin{proof}
    The bound $h\in(n/2,n]$ derives trivially: if $h\leq n/2$ then two
subsets without any faulty processes can reach the threshold for
different values (Lemma~\ref{lem:imp}).
  We calculate for which cases it is possible to cause a
disagreement. Hence, we have two disjoint partitions of honest processes such
that $|A|+|B|\leq n-d$. Suppose that processes in $A$ and in $B$ decide each a
different decision $v_A,\,v_B,\, v_A\neq v_B$. This means that both $|A|+d\geq h$ and
$|B|+d\geq h$ must hold. Adding them up, we have
$|A|+|B|+2d\geq 2h$ and since $|A|+|B|\leq n-d$ we have
$n+d\geq 2h$ for a disagreement to occur. This means that
if $h>\frac{n+d}{2}$ then it is impossible for $d$ deceitful
processes to cause a disagreement.
\end{proof}
The proof of Lemma~\ref{lem:impagr} can be straightforwardly extended to
include Byzantine processes, resulting in
Corollary~\ref{cor:impagr}. 
\begin{corollary}
  \label{cor:impagr}
  Let $\sigma$ be a protocol with voting threshold $h\in(n/2,n]$ that satisfies agreement. Then $\sigma$ tolerates at most $d+t<2h-n$ deceitful and Byzantine processes.
\end{corollary}

Next, in Lemma~\ref{lem:impter} and Corollary~\ref{cor:impter} we show the analogous results for
the termination property. That is, we show that if a protocol solves
termination while $t=0$, then it tolerates at most $q\leq n-h$ benign
processes, or $q+t\leq n-h$ benign and Byzantine processes.

\begin{lemma}[Impossibility of Termination ($t=0$)]
  \label{lem:impter}
  Let $\sigma$ be a protocol with voting threshold $h$ that satisfies termination. Then $\sigma$ tolerates at most $q\leq n-h$ benign processes.
\end{lemma}
\begin{proof}
  If $n-q<h$, then termination is not
  guaranteed, since in this case termination would require the votes from some benign processes. This is impossible if
  $h\leq n-q$, as it guarantees that the threshold is lower than
  all processes minus the
  $q$ benign processes.
\end{proof}

\begin{corollary}
  \label{cor:impter}
  Let $\sigma$ be a protocol with voting threshold $h$ that satisfies termination. Then, $\sigma$ tolerates at most $q+t\leq n-h$ benign and Byzantine processes. 
\end{corollary}
Combining the results of corollaries~\ref{cor:impagr} and~\ref{cor:impter}, one can derive an impossibility bound for a consensus protocol given its voting threshold. We show this result in Corollary~\ref{cor:impteragr}.
\begin{corollary}
  \label{cor:impteragr}
  Let $\sigma$ be a protocol that solves the consensus problem with voting threshold $h\in(n/2,n]$. Then, $\sigma$ tolerates at most $d+t<2h-n$ and $q+t\leq n-h$ Byzantine, deceitful and benign processes. 
\end{corollary}
We show in Figure~\ref{fig:fig2} the threshold $h$ to tolerate a
number $d$ of deceitful and $q$ of benign processes. For example, for a
threshold $h=\ceil{\frac{5n}{9}}-1$, then $d<\frac{n}{9}$ for
safety and $q<\frac{4n}{9}$ for liveness, with $t=0$. The maximum
number of Byzantine processes tolerated with $d=q=0$ is the minimum of
both bounds, being for example $t<\frac{n}{9}$ for
$h=\ceil{\frac{5n}{9}}-1$. In the remainder of this work, we assume
the adversary satisfies the resilient-optimal bounds of $h\leq n-q-t$ and $h>\frac{d+t+n}{2}$, given a
particular voting threshold $h$. The result of Theorem~\ref{thm:imp}
holds regardless of the voting threshold. Thus, a protocol that
satisfies both $h\leq n-q-t$ and $h>\frac{d+t+n}{2}$ can set its voting
threshold $h\in(n/2,n]$ in order to solve consensus for any
combination of $t$ Byzantine, $q$ benign and $d$ deceitful processes, as long as $n>3t+d+2q$ holds.

\begin{figure}[ht] \center
\includegraphics[width=.45\textwidth]{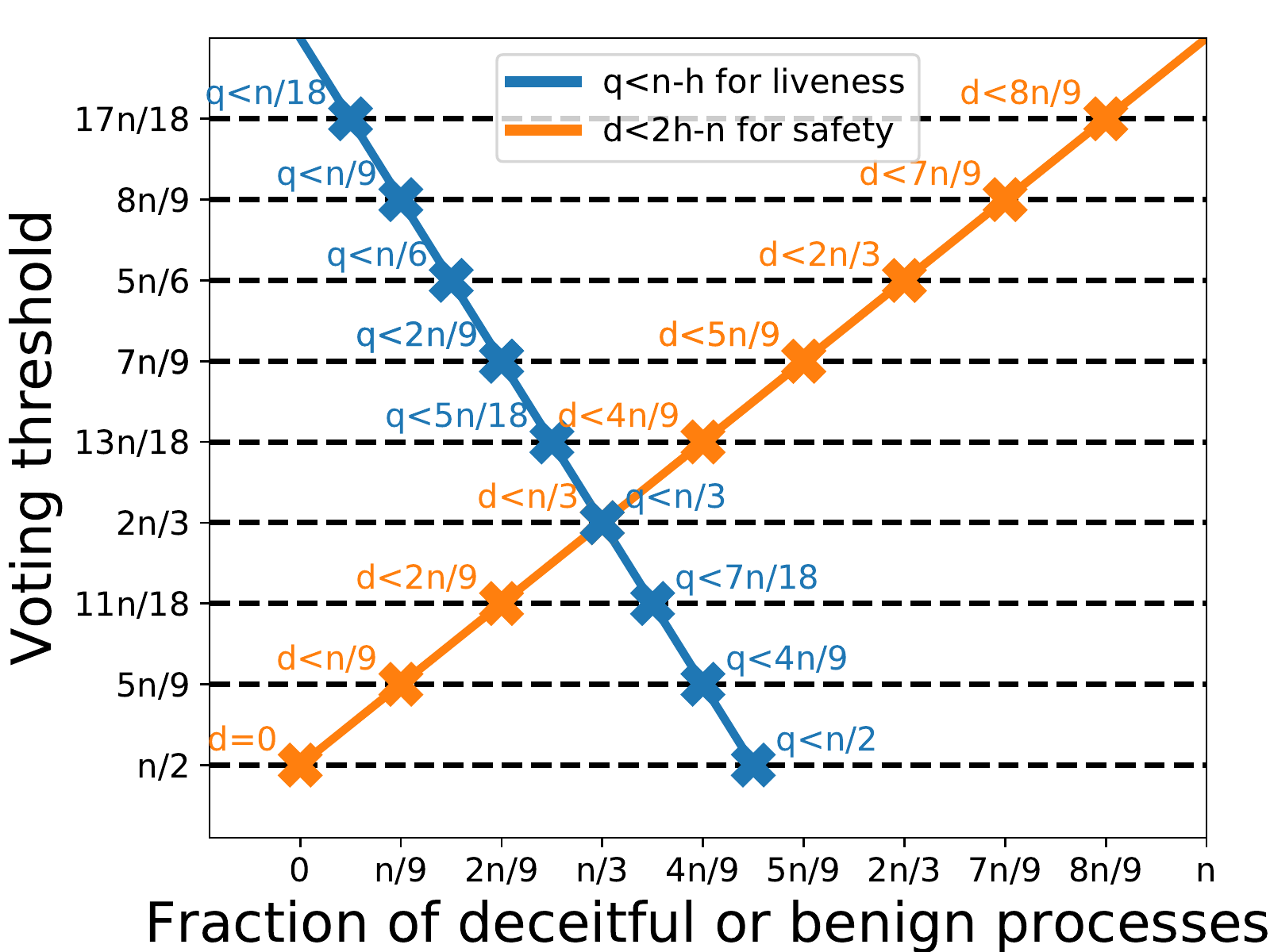}
  \caption{Number of
deceitful processes $d$ and benign processes $q$ tolerated for safety and liveness, respectively, per voting threshold $h$ and with $t=0$ Byzantine processes.}
  \label{fig:fig2}
\end{figure}

  \section{Basilic, resilient-optimal consensus in the BDB model}
  \label{sec:prot}
  In this section, we introduce the Basilic class of protocols, 
  a class of resilient-optimal protocols that solve, for different voting thresholds, the \problem problem in the BDB model. 
  In particular, all protocols within the Basilic class tolerate $t$ Byzantine, $d$ deceitful and $q$ benign processes satisfying $n>3t+d+2q$, and, given a particular protocol $\sigma(h)$ of the class uniquely defined by a voting threshold $h\in(n/2,n]$, then $\sigma(h)$ tolerates a number $n$ of processes satisfying $d+t<2h-n$ and $q+t\leq n-h$. In this section, we first need to introduce few definitions in Section~\ref{sec:model2}. Second, we present the overview of the Basilic class and show its components in Section~\ref{sec:overview}.
  \subsection{ Actively accountable consensus problem}
    \label{sec:model2}
The accountable consensus problem~\cite{CGG21} includes the property of
accountability in order to provide guarantees in the event that
deceitful and Byzantine processes manage to cause a disagreement. This
property is however insufficient for the purpose of Basilic. We need
an additional property that identifies and removes all deceitful
behavior that prevents termination. Faulty processes can break agreement
in a finite number of conflicting messages, but once they send a pair
of these conflicting messages, they leave a trace that can result in
their exclusion from the system. Our goal is to exploit this trace to make sure that
deceitful processes cannot contribute to breaking liveness. As a result,
we include the property of \myproperty, stating that deceitful faults do not
prevent termination of the protocol.

\begin{definition}[{\expandafter\MakeUppercase\problem} problem]
  \label{def:aac}
A protocol $\sigma$ with voting threshold $h$ solves the \problem
problem if the following properties are satisfied:
\begin{itemize}
\item {\bf Termination.} Every honest process eventually decides on a value.
\item {\bf Validity.} If all honest processes propose the same value, no other value can be decided.
\item {\bf Agreement.} If $d+t<2h-n$ then no two honest processes decide on different values.
\item {\bf Accountability.} If two honest processes output
disagreeing decision values, then all honest processes eventually
identify at least $2h-n$ faulty processes responsible for that
disagreement.
\item{\bf \expandafter\MakeUppercase\myproperty.} Deceitful behavior
  does not prevent liveness.
\end{itemize}
\end{definition}
We generalise the previous definition of accountability~\cite{CGG21} by including the voting threshold $h$. That is, the previous definition of accountability is the one we present in this work for the standard voting threshold of $h=2n/3$.



\subsection{Basilic Internals}
\label{sec:overview}

Basilic is a class of consensus protocols, all these protocols follow the same 
pseudocode (Algorithms~\ref{alg:prot}--\ref{alg:gen}) but differ by their voting threshold $h\in(n/2,n]$.
The structures of these protocols follow the classic reduction~\cite{BCG93} from the consensus problem, 
which accepts any ordered set of input values, to the binary consensus problem, which accepts binary
input values.

\subsubsection{Basilic Overview}
More specifically, Basilic has at its core the binary consensus
protocol called \textit{\mypropertyadj binary consensus} or AABC for
short (Alg.~\ref{alg:prot}--\ref{alg:helper}) and presented in
Section~\ref{sec:aabc}. We show in Figure~\ref{fig:example} an example
execution with $n=4$ processes in the committee. First, each process
$p_i$ selects their input value $v_i$, which they share with everyone
executing an instance of a reliable broadcast protocol called
\textit{\mypropertyadj reliable broadcast} or AARB for short (Alg.~\ref{alg:arb}). Then,
processes execute one instance $AABC_i$ of the binary consensus
protocol for each process $p_i$ to decide whether to select their associated input value from
process $p_i$. Finally, processes locally process the minimum input
value from the values whose associated AABC instance output $1$.

This Basilic binary consensus protocol shares similarities with Polygraph~\cite{CGG21}, as it also detects
guilty processes, but goes further, by excluding these detected processes and adjusting its 
voting threshold at runtime to solve consensus even in cases where Polygraph cannot ($ n/3\leq t+q+d$). We summarize the comparison of Basilic with the state of the art in Table~\ref{tab:bigtab}. Similarly to Polygraph, Basilic can perform the superblock optimization~\cite{crain2018dbft,CNG21} to solve SBC simply by deciding the union of both $v_0$ and $v_2$ in the example, instead of the minimum. This provides a better normalized communication complexity of the protocol (per decision). 
Finally, the rest of the reduction is depicted in Alg.~\ref{alg:gen} and invokes $n$ \mypropertyadj reliable 
broadcast instances or AARB (Alg.~\ref{alg:arb}) described in Section~\ref{sec:aarb}, 
followed by $n$ of the aforementioned AABC instances.

\tikzstyle{n}= [circle, fill, minimum size=10pt,inner sep=0pt, outer
sep=0pt] \tikzstyle{mul} = [circle,draw,inner sep=-1pt]
\newcounter{y}
\begin{figure*}[t]
  \hspace{-6em}
  \begin{tikzpicture}[yscale=0.5, xscale=1.2, node distance=0.3cm,
    auto]


    \def\varn{3}
    \def\varz{-1}
    \def\varf{0}
    \def\vars{2.75}
    \def\vart{5.75}
    \def\varmv{7}
    \def\varfo{9}
    \def\varfif{12}
    \def\vertscale{1}
    \foreach \i in {0,...,\varn}{
      \node (0-\i) at (\varf,-\i*\vertscale) {$p_\i: v_\i$};
      \node (1-\i) at (\vars,-\i*\vertscale) {AARB$_\i:v_\i$};
    }

    \draw [decorate, 
    decoration = {calligraphic brace,
        raise=5pt,
        amplitude=5pt}] (0-0.north)+(-0.25,0) --  (1-0.north)+(-0.25,0)
      node[pos=0.4,above=10pt,black]{reliably broadcast proposals};
      
    \foreach \i/\y in {0/1,1/0,2/1,3/0}{
      \node (2-\i) at (\vart,-\i*\vertscale) {AABC$_\i:\y$};
    }

    \node (3-i) at (\varfo,-1.5*\vertscale) {\scriptsize $\{v_0:1,\,v_1:0,\,v_2:1,\,v_3:0\}$};
    \node (4-i) at (\varfif,-1.5*\vertscale){$v_0$};
    
    \draw [decorate, 
    decoration = {calligraphic brace,
        raise=5pt,
        amplitude=5pt}] (1-0.north)+(0.25,0) -- (2-0.north)
      node[pos=0.5,above=10pt,black]{binary consensus};
      
    \node[] (MV) at (\varmv,-1.5*\vertscale) {};
    \draw[-latex'] (3-i.east) -- node[above, midway]{\scriptsize $min(v_0,v_2)$} (4-i.west);
    \draw[-latex'] (MV.center) -- (3-i.west);

    \draw [decorate, 
    decoration = {calligraphic brace,
        raise=5pt,
        amplitude=5pt}] (2-0.north)+(0.25,0) -- (10.4,0.5)
      node[pos=0.5,above=10pt,black]{bits and associated proposals};

      \draw [decorate, 
    decoration = {calligraphic brace,
        raise=5pt,
        amplitude=5pt}] (10.5,0.5) -- (12.2,0.5)
      node[pos=0.5,above=10pt,black]{decide one};

      \draw [decorate, 
    decoration = {calligraphic brace,
        raise=5pt,
        amplitude=5pt}] (-0.25,2) -- (12.2,2)
      node[pos=0.5,above=10pt,black]{Basilic's multi-valued consensus};
    
    \foreach \i in {0,...,\varn}{
      \draw (2-\i.east) -| (MV.center);        
      
      \foreach \j in {0,...,\varn}{
        \draw (0-\i.east) -- (1-\j.west);
        \draw (1-\i.east) -- (2-\j.west);

        
      }
    }

  \end{tikzpicture}
  \caption[Basilic execution example for a committee of $n=4$ processes.]{Basilic execution example for a committee of
$n=4$. First, each process $p_i$ selects their input value $v_i$,
which they share with everyone executing their respective instance
$AARB_i$ of $AARB$. Then, processes execute one instance AABC$_i$ of
the binary consensus protocol to decide whether to select their
associated input value from process $p_i$. Finally, processes locally
process the minimum input value from the values whose associated AABC
instance output $1$.}
  \label{fig:example}
\end{figure*}
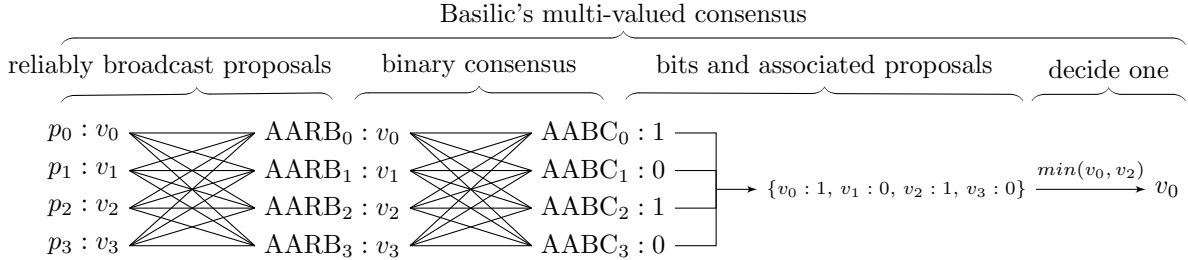


\mypar{Certificates and transferable authentication.} Basilic uses certificates in order to validate or discard a message, and also to detect deceitful processes
by cross-checking certificates. A certificate is a list of previously
delivered and signed messages 
that justifies the content of the message on which the certificate is piggybacked.
Thus, honest processes perform
transferable authentication~\cite{CFAR12}. That is, process $p_i$ can
deliver $\ms{msg}$ from $p_j$ by verifying the signature of
$\ms{msg}$, even if $\ms{msg}$ was received from $p_k$, for
$k\neq i\neq j$.

\mypar{Detected deceitful processes.} A key novelty of Basilic is to
remove detected deceitful processes from the committee at
runtime. For this reason, we refer to $d_r$ as the number of detected
deceitful processes, and define a voting threshold $h(d_r)$ that
varies with the number of detected deceitful processes. Therefore, processes start Basilic with an
initial voting threshold $h(d_r=0)=h_0$, e.g., $h_0=\ceil{\frac{2n}{3}}$, but then
update the threshold by removing detected deceitful processes,
i.e. $h(d_r)=h_0-d_r$. This way, detected deceitful processes break
neither liveness nor safety, as we will show. Certificates must always
contain $h(d_r)$ signatures from
distinct processes justifying the message (after filtering out up to
$d_r$ signatures from detected deceitful processes), or else they will
be discarded. Recall that the adversary is thus constrained to the bounds from Corollary~\ref{cor:impteragr} depending on the voting threshold. As Basilic uses a threshold that updates at runtime starting from an initial threshold $h(d_r)=h_0-d_r$, we restate these bounds 
applied to the initial threshold $h_0\leq n-q-t$ and $h_0>\frac{d+t+n}{2}$, or to the updated threshold $h(d_r)\leq n-q-t-d_r$ and
$h(d_r)>\frac{d+t+n}{2}-d_r$.


\subsubsection{The general Basilic protocol}
\label{sec:gen}
We bring together the $n$ instances of the AABC binary consensus
protocol with the $n$ instances of the AARB reliable broadcast protocol in
Algorithm~\ref{alg:gen}, where we show the general Basilic
protocol. The protocol derives from Polygraph's general protocol~\cite{civit2019techrep,CGG21},
which in turn derives from DBFT's multi-valued consensus protocol~\cite{crain2018dbft}.

Honest processes first start their respective AARB instance (for which they 
each are the source) by proposing a value in line~\ref{line:genaarb}. Delivered proposals are stored in $\ms{msgs}$ with the index
 corresponding to the source of the proposal. A binary consensus
at index $k$ is started with input value 1 for each index $k$ where a
proposal has been recorded (line~\ref{line:genaabc}). Notice
that we can guarantee to decide $1$ on at most $h(d_r)$ proposals
(line~\ref{line:gendec1}), where $d_r$ can be up to $d$ and is set by $\lit{update-committee}$ in Algorithm~\ref{alg:helper}, meaning that,
for the standard threshold $h(d_r)=\ceil{\frac{2n}{3}}-d_r$, the
maximum number of decided proposals is $\ceil{\frac{n}{3}}$, since
$d_r<\frac{n}{3}$. Once honest processes decide 1 on at least $h(d_r)$ AABC instances, honest processes start the remaining AABC
instances with input value 0 (line~\ref{line:genaabc0}), without
having to wait to AARB-deliver their respective values.

Finally, once all AABC instances have terminated
(line~\ref{line:genaabcter}), honest processes can output a decision. As such, processes take as input a list of AARB-delivered values and their associated index and output a decision selecting the AARB-delivered value with the lowest associated index whose binary consensus with the same index output 1 (line~\ref{line:genuni}).


\begin{algorithm}[th]
  \caption{The general Basilic with initial threshold $h_0$.}
  \label{alg:gen}
  \smallsize{
    \begin{algorithmic}[1]
      \Part{\smallsize $\lit{Basilic-gen-propose_{h_0}(}v_i\lit{)}$}{
        \State $\ms{msgs}\gets \lit{AARB-broadcast(}\EST, \langle v_i,i\rangle)$
        \Comment{Algorithm~\ref{alg:arb}}\label{line:genaarb}
        \Repeat{}
        \SmallIf{$\exists v,k:(\EST,\langle v,k\rangle)\in\ms{msgs}$}{}
        \Comment{proposal AARB-delivered}
        \SmallIf{$\lit{BIN-CONSENSUS}[k]$ not yet invoked}{}\Comment{Algorithm ~\ref{alg:prot}}
        \State \hspace{-0.5em}$\ms{bin-decisions}[k]\gets\lit{BIN-CONSENSUS}[k].\lit{AABC-prop}(1)$\label{line:genaabc}
        \EndSmallIf
        \EndSmallIf
        \EndRepeat
        \Until{$|\ms{bin-decisions}[k]=1|\geq h(d_r)$} \Comment{decide $1$ on at least $h(d_r)$}\label{line:gendec1}\EndUntil
        \For{\textbf{all} $k$ such that $\lit{BIN-CONSENSUS}[k]$ not yet invoked}

        \State $\ms{bin-decisions}[k]\gets\lit{BIN-CONSENSUS}[k].\lit{AABC-prop}(0)$

        \label{line:genaabc0}
        \EndFor
        \WUntil{\textbf{for all} $k,\,\ms{bin-decisions}[k]\neq \bot$}\EndWUntil\label{line:genaabcter}
            
        \State $j\gets \min\{k:\ms{bin-decisions}[k]=1\} \lit{)} $\label{line:genuni1}        
        \WUntil{$\exists v:(\EST, \langle v,j\rangle)\in\ms{msgs}$}\EndWUntil
        \State \textbf{decide} $v$\label{line:genuni}        
      }\EndPart
      \algstore{alg:gen}
    \end{algorithmic}
  }
\end{algorithm}

\subsubsection{\Mypropertyadj binary consensus}
\label{sec:aabc}
We show in Algorithm~\ref{alg:prot} the Basilic
\textit{\mypropertyadj binary consensus} (AABC) protocol with initial threshold $h_0\in(n/2,n]$, along with some additional components and functions in Algorithm~\ref{alg:helper}. First, note that all delivered messages are correctly signed 
(as wrongly signed messages are discarded)
and
stored in $\ms{sig\_msgs}$, along with all sent messages (as we
detail in Rule~\ref{item:signed} of Alg.~\ref{alg:prot}).

\begin{algorithm}[htp]
  \caption{Basilic's AABC with initial threshold $h_0$ for $p_i$.
  }
  \label{alg:prot}
  \smallsize{
    \begin{algorithmic}[1]
      \algrestore{alg:gen}
      \Part{\smallsize $\lit{AABC-prop_{h_0}(}v_i\lit{)}$}{
        \State $\ms{est}\gets v_i$\label{line:estimate}
        \State $\ms{r}\gets 0$
        \State $\ms{timeout}\gets 0$
        \State $\ms{cert}[0]\gets \emptyset$
        \State $\ms{bin\_vals} \gets \emptyset$
        \Repeat{}
        \State $r\gets r+1$
        \State $\ms{timeout}\gets \Delta$\Comment{set timer}
        \State $coord\gets ((r-1)mod\,n)+1$\Comment{rotate coordinator}\label{line:rotcoord}
        \EndRepeat
        \Part{$\blacktriangleright$ Phase 1}{
          \State $timer\gets \lit{start-timer}(\ms{timeout})$ \Comment{start timer}
          \State $\lit{abv-broadcast}(\EST[r],\ms{est},\ms{cert}[r-1],i,\ms{bin\_vals})$
          
          \label{line:abv-broadcast}
          \SmallIf{$i=coord$}{}
          \label{line:coordbroad1}
          \WUntil{$\ms{bin\_vals}[r]=\{w\}$}          \EndWUntil
          \State $\lit{broadcast}(\COORD[r],w)$
          \EndSmallIf\label{line:coordbroad2}
          \WUntil{$\ms{bin\_vals}[r]\neq \emptyset \wedge \ms{timer}$ expired} \EndWUntil\label{line:wuntiphase1}
          
        }\EndPart
        \Part{$\blacktriangleright$ Phase 2}{
          \State $\ms{timer}\gets \ms{timeout}$ \Comment{reset timer}
          \SmallIf{$(\COORD[r],w)\in\ms{sig\_msgs} \wedge w\in\ms{bin\_vals}[r]$}{}%
          \label{line:delcoord}
          \State $\ms{aux}\gets \{w\}$\Comment{prioritize coordinator's value}
          \EndSmallIf
          \SmallElse{$\ms{aux}\gets \ms{bin\_vals}[r]$}\Comment{else use any received value}\label{line:binvalues1}
          \EndSmallElse
          \State $\lit{broadcast}(\ECHO[r],\ms{aux})$\Comment{broadcast signed \ECHO message}\label{line:broadecho}
          \WUntil{($\ms{vals} = \lit{comp-vals}(\ms{sig\_msgs},\ms{bin\_vals},\ms{aux}))\neq \emptyset \wedge \ms{timer}$ expired}\EndWUntil\label{line:callcomputeval}
        }\EndPart
        \Part{$\blacktriangleright$ Decision phase}{
          \SmallIf{$|\ms{vals}|=1$}{$\ms{est}\gets\ms{vals}[0]$}
          \Comment{if only one, adopt as estimate}\label{line:dec1}
          \SmallIf{$\ms{est}=(r\,mod\,2)\,\wedge\,p_i$ not decided before}{}
          \State $\lit{decide}(\ms{est}); {\bf \lit{return }} \ms{est}$ \Comment{if parity matches, decide the estimate}\label{line:decide}
          \EndSmallIf
          \EndSmallIf
          \SmallElse{$\ms{est} \gets (r\,mod\,2)$}\Comment{otherwise, the estimate is the round's parity bit}\label{line:adoptestimate2}
          \EndSmallElse
          \State $\ms{cert}[r]\gets \lit{compute-cert}(\ms{vals},\ms{est},r,\ms{bin\_vals},$ $\ms{sig\_msgs})$
          \label{line:dec2}
          
          }\EndPart
        }\EndPart
        \Part{\textbf{Upon} receiving a signed message $\ms{s\_msg}$}{\label{line:dr1}
          \State $\ms{pofs}\gets \lit{check-conflicts}(\ms{\{s\_msg\}},\,\ms{sig\_msgs})$
          \Comment{returns $\emptyset$ or PoFs}\label{lin:smcc}
          \State $\lit{update-committee}(\ms{pofs})$\label{lin:smuc}\Comment{remove fraudsters}
          }\EndPart
          \Part{Upon receiving a certificate $\ms{cert\_msg}$}{
            \State $\ms{pofs}\gets \lit{check-conflicts}(\ms{cert\_msg},\,\ms{sig\_msgs})$          \Comment{returns $\emptyset$ or PoFs} 
            \State $\lit{update-committee}(\ms{pofs})$\Comment{remove fraudsters}
          }\EndPart
          \Part{Upon receiving a list of PoFs $\ms{pofs\_msg}$}{
            \SmallIf{$\lit{verify-pofs}(\ms{pofs\_msg})$}{}\Comment{if proofs are valid then}
            \State $\lit{update-committee}(\ms{pofs\_msg})$
            \Comment{remove fraudsters from committee}
            \EndSmallIf\label{line:dr2}
          }\EndPart
        \Part{Rules}{\label{lin:rul}
          \begin{enumerate}[leftmargin=* ,wide=\parindent]
          \item Every message that is not properly signed by the sender is discarded.
          \item Every message that is sent by $\lit{abv-broadcast}$ without a valid certificate after Round $1$, except for messages with value $1$ in Round $2$, are discarded.
          \item Every signed message received is stored in $\ms{sig\_msgs}$, including messages within certificates.\label{item:signed}
          \item Every time the timer reaches the timeout for a phase, and if that phase cannot be terminated, processes broadcast their current delivered signed messages for that phase (and all messages received for future phases and rounds) and reset the timer for that phase. These messages are added to the local set of messages and cross-checked for PoFs on arrival.\label{item:timer} 
          \end{enumerate}
         }\EndPart
         \algstore{alg:prot}
    \end{algorithmic}
  }
\end{algorithm}

The Basilic's AABC protocol is divided in two phases, after which a decision is taken. A key difference with Polygraph is that
when a timer for one of the two phases reaches its timeout, if a
process cannot terminate that phase yet, then it broadcasts its set of
signed messages for that phase and resets the timer, as detailed in Rule~\ref{item:timer}. This allows Basilic
to prevent deceitful processes from breaking termination by trying to
cause a disagreement and never succeeding. For example, for $n=4$ and
$h=\ceil{2n/3}=3$, if $q=1$ and $d=1$, the deceitful process could
prevent the $2$ honest processes from terminating by constantly
sending them conflicting messages, even if none of these honest
will reach the threshold for the disagreeing values. Thus,
once the timer is reached, processes exchange their known set messages
and can update the committee removing processes that sent conflicting
messages.
It is important that processes wait for this timer before taking a
decision for the phase, or before exchanging signed messages, since
only waiting for that timer guarantees that all sent messages will be
received before the timer reaches its timeout, after GST.
Each process maintains an estimate (line~\ref{line:estimate}), initially given as input, and then proceeds in rounds executing the following phases:
\begin{enumerate}[leftmargin=* ,wide=\parindent]
  \item In the first phase, each process
broadcasts its estimate (given as input) via an accountable binary value reliable
broadcast (ABV-broadcast) (line~\ref{line:abv-broadcast}), which
we present in Algorithm~\ref{alg:helper}, lines~\ref{line:abvb-start}--\ref{line:abvb-end} and discuss in Section~\ref{sec:aabc}.
Decision and
$\lit{abv-broadcast}$ messages are discarded unless they come with a
certificate justifying them.

The protocol also uses a rotating coordinator (line~\ref{line:rotcoord}) per round which carries
a special \COORD message (lines~\ref{line:coordbroad1}-\ref{line:coordbroad2}). All processes wait until they deliver at
least one message from the call to $\lit{abv-broadcast}$ and until the
timer, initially set to $\Delta$, expires (line~\ref{line:wuntiphase1}). 
(Note that the bound on the message delays remains unknown due to the unknown GST.)
If a process delivers a message from
the coordinator (line~\ref{line:delcoord}), then it broadcasts an \ECHO message with the
coordinator's value and signature in the second phase (line~\ref{line:broadecho}). Otherwise, it
echoes all the values delivered in phase $1$ as part of the call
to $\lit{abv-broadcast}$ (line~\ref{line:binvalues1}).

\item In the second phase, processes wait till they receive $h(d_r)$ \ECHO
messages, as shown in the call to $\lit{comp-vals}$ (line~\ref{line:callcomputeval}), which returns
the set of values that contain these $h(d_r)$ signed \ECHO
messages. Function $\lit{comp-vals}$ is depicted in
Algorithm~\ref{alg:helper} (lines~\ref{line:comp-val-start}--\ref{line:comp-val-end}). Processes then
try to come to a decision in lines~\ref{line:dec1}-\ref{line:dec2}. As it was the case for phase $1$, when the
timer expires in phase $2$, all processes broadcast their current set
of \ECHO messages. Then, they update their committee if they detect deceitful processes through PoFs (lines~\ref{line:dr1}-\ref{line:dr2}) and recheck if
they reach the updated $h(d_r)$ threshold, after which they reset the
timer.

\item During the decision phase, if there is just one value returned by
$\lit{comp-vals}$ and that value's parity matches with the
round's parity, process $p_i$ decides it (line~\ref{line:decide}) and broadcasts the associated
certificate in the call to $\lit{compute-cert}$. If the parity
does not match then process $p_i$ simply adopts the value as the
estimate for the next round (line~\ref{line:dec1}). If instead there is more than one value
returned by $\lit{comp-vals}$ then $p_i$ adopts the round's
parity as next round's estimate
(line~\ref{line:adoptestimate2}). Adopting the parity as next round's
estimate helps with convergence in the next round, in this case where
processes are hesitating between two values. The call to
$\lit{compute-cert}$ (depicted at lines~\ref{lin:com-cer}--\ref{lin:com-cer-end} of Algorithm~\ref{alg:helper}) gathers the signatures justifying the
current estimate and broadcasts the certificate if the estimate was
decided in this round. 

\end{enumerate}
\begin{algorithm}[htbp]
  \caption{Helper components.}
  \label{alg:helper}
  \smallsize{
    \begin{algorithmic}[1]
    \algrestore{alg:prot}
      \Part{$\lit{update-committee}(\ms{new\_pofs})$}{\Comment{function that removes fraudsters}\label{line:update-com-start}
        \SmallIf{$\ms{new\_pofs}\neq \emptyset\,\wedge\, \ms{new\_pofs}\not\subseteq \ms{local\_pofs}$}{}
        \State $\ms{new\_pofs}\gets \ms{new\_pofs}\backslash\ms{local\_pofs}$\Comment{consider only new PoFs}
        \State $\ms{local\_pofs}\gets \ms{local\_pofs}\cup\ms{new\_pofs}$\Comment{store new PoFs}
        \State $\lit{broadcast}(\POF,\,\ms{new\_pofs})$\Comment{broadcast new PoFs}
        \State $\ms{new\_deceitful}\gets \ms{new\_pofs}.\lit{get\_processes()}$\Comment{get deceitful from PoFs}
        \State $\ms{new\_deceitful}\gets \ms{new\_deceitful}\backslash\ms{local\_deceitful}$
        \State $\ms{local\_deceitful}\gets \ms{local\_deceitful}\cup\ms{new\_deceitful}$      
        \State $N\gets N\backslash \{\ms{new\_deceitful}\}$; $n\gets |N|$   \Comment{remove new deceitful}
        \State $d_r\gets |\ms{local\_deceitful}|$\Comment{update number of detected deceitful}
        \State $h(d_r)\gets \lit{recalculate-threshold}(N,\,d_r)$        
        \State $\lit{recheck-certs-termination}()$
        \Comment{check termination of current phase}\label{line:recheccerts}
          \State $\lit{reset-current-timer()}$\Comment{reset timer of current phase}\label{line:resetcurt}
          \EndSmallIf \label{line:update-com-end}
        }\EndPart
        \Statex 
        \Part{$\lit{abv-broadcast}(\MSG,\ms{val},\ms{cert},i,\ms{bin\_vals})$}{ \label{line:abvb-start}
          \State $\lit{broadcast}(\BVALECHO,\langle \ms{val},\ms{cert},i\rangle)$\Comment{broadcast message}\label{line:bvechobroad1}
          \SmallIf{ $r=3$ \textbf{or} $(r=2$ \textbf{and} $\ms{val}=1)$}{ discard all messages received without a valid certificate}
          \EndSmallIf
          \Upon{receipt of $(\BVALECHO,\langle v,\cdot,j\rangle )$}
          \SmallIf{$(\BVALECHO,\langle v,\cdot,\cdot\rangle)$ received from $\floor{\frac{n-q-t}{2}}-d_r+1$ processes \textbf{and} $\BVALECHO,\langle v,\cdot, i\rangle)$ not broadcast}{}
          \State Let $\ms{cert}$ be any valid certificate $cert$ received in these messages
          \State $\lit{broadcast}(\BVALECHO,\langle v,\ms{cert},i\rangle)$
          \label{line:bvechobroad2}
          \EndSmallIf
          \SmallIf{$(\BVALECHO,\langle v,\cdot,\cdot\rangle)$ received from $h(d_r)$ processes \textbf{and} $(\BVALREADY,\langle v,\cdot, \cdot\rangle)$ not yet broadcast}{}
          \State Let $\ms{cert}$ be any valid certificate $cert$ received in these messages

          \State Construct $\ms{bv\_cert}$ a certificate with $h(d_r)$ signed $\BVALECHO$

          \label{line:bvreadycons}
          \State $\ms{bin\_vals}\gets\ms{bin\_vals}.\lit{add}(\BVALREADY,\langle v,\ms{cert},j,\ms{bv\_cert}\rangle)$\label{line:bvdel1}

          \State $\lit{broadcast}(\BVALREADY,\langle v,\ms{cert},j,\ms{bv\_cert}\rangle)$
          \EndSmallIf
          \SmallIf{$(\BVALREADY,\langle v,cert,j,\ms{bv\_cert}\rangle)$ received from $1$ process}{}\label{line:bvreadyrec}
          \State $\ms{bin\_vals}\gets\ms{bin\_vals}.\lit{add}(\BVALREADY,\langle v,\ms{cert},j,\ms{bv\_cert}\rangle)$\label{line:bvdel2}

          \SmallIf{$(\BVALREADY,\langle v,cert,j,\ms{bv\_cert}\rangle)$ not yet broadcast}{}
          \State $\lit{broadcast}(\BVALREADY,\langle \ms{val},\ms{cert},i,\ms{bv\_cert}\rangle)$
          \EndSmallIf
          \EndSmallIf
          \EndUpon  \label{line:abvb-end}
        }\EndPart
        \Statex 
          \Part{$\lit{comp-vals}(\ms{msgs},\ms{b\_set},\ms{aux\_set})$}{ \label{line:comp-val-start}
            \Comment{check for termination of phase $2$}
            \State \textbf{If }$\exists S\subseteq \ms{msgs}$ where the following conditions hold:
            \State $\>\>$ $(i)\;|S|$ contains $h(d_r)$ distinct $\ECHO[r]$ messages        
            \State $\>\>$ $(ii)\;\ms{aux\_set}$ is equal to the set of values in $S$    \Comment{$h(d_r)$ with same est}
            \State $\>$\textbf{then return}$(\ms{aux\_set})$
            \State \textbf{Else If }$\exists S\subseteq \ms{msgs}$ where the following conditions hold:

            \State $\>\>$ $(i)\;|S|$ contains $h(d_r)$ distinct $\ECHO[r]$ messages
            \State $\>\>$ $(ii)\;$Every value in $S$ is in $\ms{b\_set}$    \Comment{$h(d_r)$ messages with different est}
            \State $\>$\textbf{then return}$(V=$ the set of values in $S)$
            \State \textbf{Else return}$(\emptyset)$ \Comment{else not ready to terminate} \label{line:comp-val-end}
            
          }\EndPart
          \Statex 
          \Part{$\lit{compute-cert}(\ms{vals},\ms{est},r,\ms{bin\_vals},\ms{msgs})$}{\label{lin:com-cer} \Comment{compute and send cert}

            \SmallIf{$\ms{est}=(r\,mod\,2)$}{} 
            \SmallIf{$r>1$}{}
            \State $\ms{to\_return}\gets(\ms{cert}$ : $(\EST[r],\langle v, \ms{cert},\cdot\rangle)\in\ms{bin\_vals})$
            \EndSmallIf
            \SmallElse{$\ms{to\_return}\gets(\emptyset)$}
            \EndSmallElse
            \EndSmallIf
            \SmallElse{ $\ms{to\_return}\gets(h(d_r)$ signed msgs containing only $\ms{est})$}
            \EndSmallElse
            \SmallIf{$\ms{vals}=\{(r\,mod\,2)\}\,\wedge\,$no previous decision by $p_i$}{}
            \State $cert[r]\gets h(d_r)$ signed messages containing only $r\,mod\,2$
            \State $\lit{broadcast}(\ms{est},r,i,\ms{cert}[r])$\Comment{broadcast decision}
              \EndSmallIf
              \State $\textbf{return}(\ms{to\_return})$ \label{lin:com-cer-end}
            }\EndPart 
    \algstore{alg:helper}
    \end{algorithmic}
  }
\end{algorithm}

\mypar{Detecting and removing deceitful processes.} Upon receiving a signed message, honest processes check if the
received message conflicts with some previously delivered message in
storage in $\ms{sig\_msgs}$ by calling $\lit{check-conflicts}$
(line~\ref{lin:smcc}). This function returns $\ms{pofs}=\emptyset$ if
there are no conflicting messages, or a list $\ms{pofs}$ of PoFs otherwise. Then, at line~\ref{lin:smuc},
honest processes call $\lit{update-committee}$ (depicted at lines~\ref{line:update-com-start}--\ref{line:update-com-end} of
Algorithm~\ref{alg:helper}) to remove the $|\ms{pofs}|$ detected deceitful
processes at runtime. In the call to
$\lit{update-committee}$, process $p_i$ removes all processes that are
proven deceitful via new PoFs, and updates the committee $N$, its size
$n$, and the voting threshold $h(d_r)$. After that, $p_i$ rechecks all
delivered messages in that phase in case it can now terminate the
phase with the new threshold $h(d_r)$ (and after filtering out
messages delivered by the $d_r$ removed deceitful processes) by calling $\lit{recheck-certs-termination}()$ in line~\ref{line:recheccerts} of Algorithm~\ref{alg:helper}. Finally,
it resets the timer for the current phase by calling $\lit{reset-current-timer()}$ in line~\ref{line:resetcurt} of Algorithm~\ref{alg:helper}.

\mypar{Termination and agreement of Basilic's AABC.} We
show the detailed proofs of agreement and termination in
Lemmas~\ref{lem:aabc-agr} and~\ref{lem:aabc-ter}. The idea is that 
removing deceitful processes has no effect on agreement, while it
facilitates termination, since the threshold $h(d_r)=h_0-d_r$
decreases the initial threshold $h_0$ with the number of removed
deceitful processes. Also, since all honest
processes broadcast their delivered PoFs and thanks to the property of accountability, eventually all honest
processes agree on the same set of removed deceitful processes.

Then, if a
process $p_i$ terminates broadcasting certificate $cert_i$ while
another process $p_j$ already removed newly detected deceitful
processes $new\_d_r$ present in $cert_i$, then
$|cert_i|-new\_d_r\geq h(d_r+new\_d_r)$ by construction. As such, either an honest process terminates and then all
subsequent honest processes can terminate, even after removing more deceitful
processes, or honest processes eventually reach a scenario where all deceitful
processes are detected $d_r=d$ and removed, after which honest processes terminate.



Note that removing processes at runtime can result in rounds whose
coordinator is already removed. For the sake of correctness, we do not
change the coordinator for that round even if it has already been
removed. This guarantees that all honest processes eventually reach a
round in which they all agree on the same coordinator, which is an
honest process. If this round is the first after GST and after all
deceitful processes have been removed from the committee, then honest
processes will reach agreement.

\mypar{Accountable binary value broadcast.} The
ABV-broadcast that we present in Algorithm~\ref{alg:helper} is inspired from the E
protocol presented by Malkhi et al.~\cite{malkhi1997secure} and the binary
broadcast presented in Polygraph~\cite{civit2019techrep,CGG21}. If honest processes add a value $v$ to $\ms{bin\_vals}$ (lines~\ref{line:bvdel1} and~\ref{line:bvdel2}) as a result of the ABV-broadcast, we say that they \textit{ABV-deliver} $v$. Processes exchange \BVALECHO and \BVALREADY messages during ABV-broadcast. \BVALECHO messages are
signed and must come with a valid certificate $\ms{cert}_i$ justifying
the value, as shown in lines~\ref{line:bvechobroad1} and~\ref{line:bvechobroad2}. \BVALREADY messages carry the same information as \BVALECHO
messages plus an additional certificate $\ms{bv\_cert}$ containing
$h(d_r)$ \BVALECHO messages justifying the \BVALREADY message, constructed in line~\ref{line:bvreadycons}. This way, as
soon as a process receives a \BVALREADY message with a value (line~\ref{line:bvreadyrec}), it already
obtains $h(d_r)$ \BVALECHO messages too, meaning it can ABV-deliver that
value adding it to $\ms{bin\_vals}$ (lines~\ref{line:bvdel1} and~\ref{line:bvdel2}). Honest processes broadcast signed \BVALECHO messages for their estimate (line~\ref{line:bvechobroad1}) and for all values for which they receive at least $\floor{\frac{n-q-t}{2}}-d_r+1$ signed \BVALECHO messages from distinct processes. Waiting for this many
\BVALECHO messages for a value $v$ guarantees that all honest
processes ABV-deliver $v$, as we show in Section~\ref{sec:proofs}.

In particular, we show that our ABV-broadcast satisfies the following
properties: (i) ABV-Termination, in that every honest process
eventually adds at least one value to $\ms{bin\_vals}$; (ii)
ABV-Uniformity, in that honest processes eventually add the same
values to $\ms{bin\_vals}$; (iii) ABV-Obligation, in that if
$\floor{\frac{n-q-t}{2}}-d_r+1$ honest processes ABV-broadcast a value
$v$, then all honest processes ABV-deliver $v$; (iv)
ABV-Justification, in that if an honest process ABV-delivers a value
$v$ then $v$ was ABV-broadcast by an honest process; and (v)
ABV-Accountability, in that every ABV-delivered value contains a valid
certificate from the previous round.

\subsubsection{\Mypropertyadj reliable broadcast}
\label{sec:aarb}
Algorithm~\ref{alg:arb} shows Basilic's \textit{\mypropertyadj reliable broadcast} (AARB). The protocol is analogous to the secure broadcast presented in previous work~\cite{malkhi1997secure}, with the difference that we also introduce a timer that honest processes use to periodically broadcast their set of delivered \ECHO messages, in order to detect deceitful processes. The protocol starts when the source broadcasts an \INIT message with its proposed value $v$ (line~\ref{line:aarbecho1}). Upon delivering that message, all honest processes also broadcast a signed \ECHO message with $v$ (line~\ref{line:aarbecho2}). Then, once a process $p_i$ delivers $h(d_r)$ distinct signed \ECHO messages for the same value $v$, $p_i$ first broadcasts a \READY message (line~\ref{line:aarb-broadcastReady1}) with a certificate containing the $h(d_r)$ \ECHO messages justifying $v$ (constructed in line~\ref{line:aarbconstructcert1}), and then AARB-delivers the value (line~\ref{line:aarb-deliver1}). The same occurs if instead a process delivers just one valid \READY message containing a valid certificate justifying it in lines~\ref{line:aarb-read1}-\ref{line:aarb-read2}.

As it occurs with Basilic's AABC protocol presented in Algorithms~\ref{alg:prot} and~\ref{alg:helper}, upon cross-checking newly received signed messages with previously delivered ones (lines~\ref{line:crosscheck1} and~\ref{line:crosscheck2}), honest processes can detect deceitful faults and update the committee (lines~\ref{line:aarbupdcom1} and~\ref{line:aarbupdcom2}), removing them at runtime, by calling $\lit{update-committee}$. This can also occur when receiving a list of PoFs (line~\ref{line:aarbpofs}). Note that this is the same call to the same function as in the AABC protocol shown in Algorithm~\ref{alg:prot}, because honest processes update the committee across the entire Basilic protocol, and not just for that particular instance of AARB or AABC where the deceitful process was detected.
We show in Section~\ref{sec:proofs} that Basilic's AARB protocol satisfies the following properties of \mypropertyadj reliable broadcast:
\begin{itemize}[leftmargin=* ,wide=\parindent]
  \item {\bf AARB-Unicity.} honest processes AARB-deliver at most one value.
  \item {\bf AARB-Validity.} honest processes AARB-deliver a value if it was previously AARB-broadcast by the source.
  \item {\bf AARB-Send.} If the source is honest and AARB-broadcasts $v$, then honest processes AARB-deliver $v$.
  \item {\bf AARB-Receive.} If an honest process AARB-delivers $v$, then all honest processes AARB-deliver $v$.
  \item {\bf AARB-Accountability.} If two honest processes AARB-deliver distinct values, then all honest processes receive PoFs of the deceitful behavior of at least $2h(d_r)-n$ processes including the source.
  \item{\bf AARB-\expandafter\MakeUppercase\myproperty.} Deceitful behavior does not prevent liveness.
  \end{itemize}




\begin{algorithm}[th]
  \caption{Basilic's AARB with initial threshold $h_0$.}
  \label{alg:arb}
  \smallsize{
    \begin{algorithmic}[1]
    \algrestore{alg:helper}
      \Part{\smallsize $\lit{AARB-broadcast_{h_0}(}v_i\lit{)}$}{\Comment{executed by the source}
        \State $\lit{broadcast(}\INIT, v_i)$ \Comment{broadcast to all}\label{line:aarbecho1}
        \Upon{\textbf{receiving $(\INIT, v_i)$ from $p_j$ \textbf{and}} \textbf{not having sent $\ECHO$:}}
        \State $\lit{boadcast(}\ECHO,v,j)$\Comment{echo value to all}\label{line:aarbecho2}
        \EndUpon
        \Upon{\textbf{receiving $h(d_r)$ $(\ECHO, v, j)$ \textbf{and}}\textbf{ not having sent a \READY:}}
        \State Construct $cert_i$ containing at least $h(d_r)$ signed msgs $(\ECHO, v, j)$
        \label{line:aarbconstructcert1}
        \State $\lit{broadcast(}\READY,v,cert_i,j)$\Comment{broadcast certificate}\label{line:aarb-broadcastReady1}
        \State $\lit{AARB-deliver(}v,j\lit{)}$ \Comment{AARB-deliver value}\label{line:aarb-deliver1}
        \EndUpon
        \Upon{\textbf{receiving} $(\READY,v,cert,j)$, \textbf{and} \textbf{not} \textbf{having sent a} $\READY$:}\label{line:aarb-read1}
        \SmallIf{$\lit{verify}(cert)=False$}{\textbf{continue}}\EndSmallIf
        \State Set $cert_i$ to be one of the valid certs received $(\READY,v,cert,j)$
        \State $\lit{broadcast(}\READY,v,cert_i,j)$\Comment{broadcast certificate}
        \State $\lit{AARB-deliver(}v,j\lit{)}$\Comment{AARB-deliver value}\label{line:aarb-read2}
        \EndUpon

        \Part{\textbf{Upon} receiving a signed message $\ms{s\_msg}$}{
          \State $\ms{pofs}\gets \lit{check-conflicts}(\ms{\{s\_msg\}},\,\ms{sig\_msgs})$\Comment{returns $\emptyset$ or PoFs} 

          \label{line:crosscheck1}
          \State $\lit{update-committee}(\ms{pofs})$\Comment{remove fraudsters}\label{line:aarbupdcom1}
          }\EndPart
          \Part{Upon receiving a certificate $\ms{cert\_msg}$}{
            \State $\ms{pofs}\gets \lit{check-conflicts}(\ms{cert\_msg},\,\ms{sig\_msgs})$\Comment{returns $\emptyset$ or PoFs} \label{line:crosscheck2}
            \State $\lit{update-committee}(\ms{pofs})$\Comment{remove fraudsters}\label{line:aarbupdcom2}
          }\EndPart
          \Part{Upon receiving a list of PoFs $\ms{pofs\_msg}$}{\label{line:aarbpofs}
            \SmallIf{$\lit{verify-pofs}(\ms{pofs\_msg})$}{}\Comment{if proofs are valid then}
            \State$\lit{update-committee}(\ms{pofs\_msg})$
            \Comment{exclude from committee}
            \EndSmallIf
          }\EndPart
      }\EndPart
      \Part{Rules}{
          \begin{enumerate}[leftmargin=* ,wide=\parindent]
          \item Processes broadcast their current delivered signed
$\INIT$ and $\ECHO$ messages once a timer $timer$, initially set to
$\Delta$, reaches $0$, and reset the timer to $\Delta$.
          \end{enumerate}
         }\EndPart
    \end{algorithmic}
  }
\end{algorithm}

\subsection{Basilic's fault tolerance in the BDB model}
We show in Figure~\ref{fig:fig31} the combinations of Byzantine, deceitful
and benign processes that Basilic tolerates, depending on the
initial threshold $h_0$. The solid lines represent the variation in tolerance to benign and
deceitful processes as the number of Byzantine processes varies for a
particular threshold. For example, for $h_0=\frac{2n}{3}$,
if $t=0$ then $d<\frac{n}{3}$ and $q<\frac{n}{3}$. As $t$
increases, for example to $t=\ceil{\frac{n}{6}}-1$, then
$d<\frac{n}{6}$ and $q<\frac{n}{6}$.
\begin{figure}[htp] \center
  \subfloat[Combinations of benign, deceitful and Byzantine processes that Basilic tolerates, for an initial threshold $h_0$.\label{fig:fig31}]{
    \includegraphics[width=.45\textwidth]{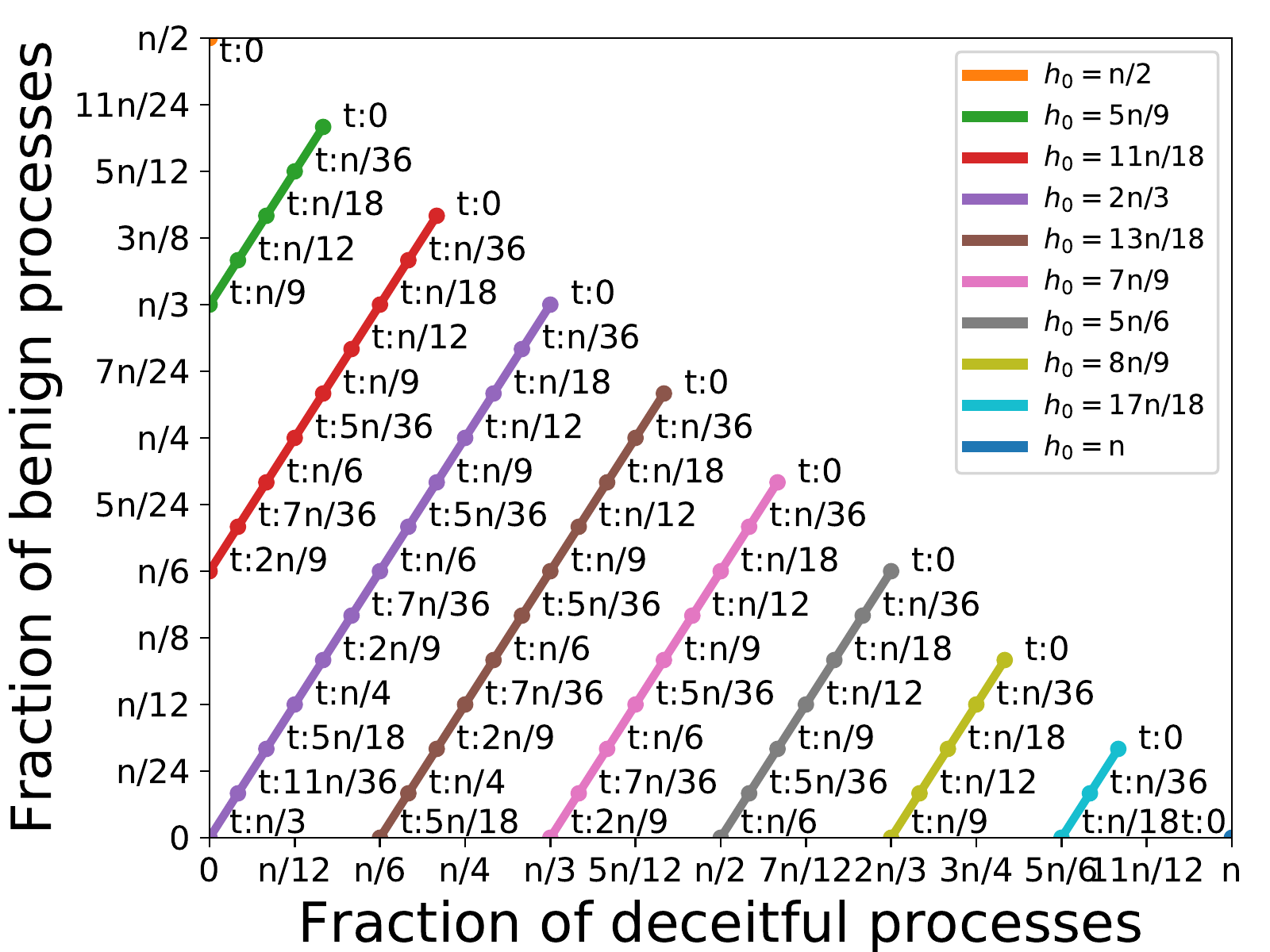}
 }\hfill
 \subfloat[Fraction of faulty processes, compared with fraction of Byzantine processes, for a particular initial threshold $h_0$ of the general Basilic protocol, compared with other works.\label{fig:fig32}]{
   \includegraphics[width=.45\textwidth]{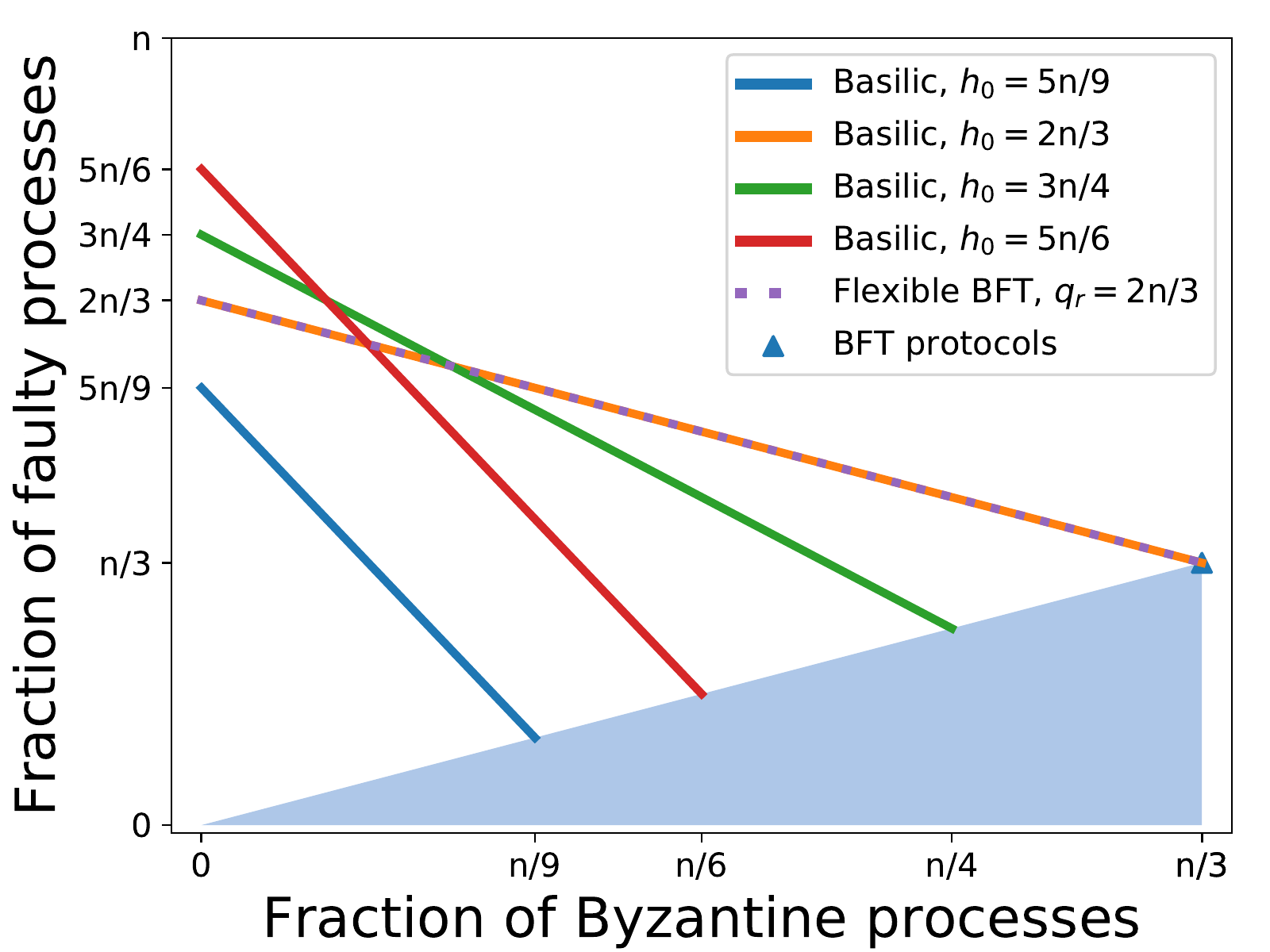}
 }
\end{figure}

We compare our Basilic's fault tolerance with that of previous works
in Figure~\ref{fig:fig32}. In particular, we represent multiple values
of the initial threshold
$h_0\in\{5n/9,\,2n/3,\,3n/4,\,5n/6\}$ for Basilic. First, we show that classical
Byzantine fault-tolerant (BFT) protocols tolerate only the case $t<n/3$ with a blue triangle dot (\markerthree) in the
figure. This is the case of most partially synchronous BFT consensus protocols~\cite{crain2018dbft,CGG21, KADC07, YMR19}. 
Second, we
represent Flexible BFT~\cite{MNR19} in their greatest fault tolerance
setting in partial synchrony. As we can see, such setting overlaps
with Basilic's initial threshold of $h_0=2n/3$. However, the difference lies in that
while Basilic tolerates all the cases in the solid line $h_0=2n/3$,
Flexible BFT only tolerates a particular dot of the line, set at the discretion of
each client. That is, Flexible BFT's clients must decide, for example,
whether they tolerate either $\ceil{2n/3}-1$ total faults, being none
of them Byzantine, or instead tolerate $\ceil{n/3}-1$ Byzantine faults, not tolerating any
additional fault. Basilic can however tolerate any range satisfying
both $h_0>\frac{n+d+t}{2}$ for safety and $h_0\leq n-q-t$ for
liveness, which allows our clients and processes to tolerate significantly more combinations of faults for one particular threshold $h_0\in(n/2,n]$. For this reason, we
represent the line of Flexible BFT as a dashed line, whereas
Basilic's lines are solid. For each initial voting threshold $h_0$, the maximum
number of Byzantine processes Basilic tolerates is $t<\min(2h_0-n,n-h_0)$,
which is obtained by setting $q=d=0$ and resolving both bounds for
safety and liveness.
  

\subsection{Basilic's correctness}
\label{sec:proofs}
In this section, 
we prove the rest of the properties of Basilic, including its ABV-broadcast, AABC and AARB protocols.
\subsubsection{Accountable binary value broadcast}

We first start with the properties that ABV-broadcast satisfies. We say process $p_i$ ABV-broadcasts value $v$ to refer to $p_i$ sending a \BVALECHO message containing $v$ and a valid certificate justifying $v$.
We prove ABV-termination in Lemma~\ref{lem:aabv-ter}, ABV-uniformity in Lemma~\ref{lem:aabv-uni}, ABV-obligation in Lemma~\ref{lem:aabv-obl}, ABV-justification in Lemma~\ref{lem:aabv-jus}, and ABV-accountability in Lemma~\ref{lem:aabv-acc}.

\begin{lemma}[ABV-Termination]
  \label{lem:aabv-ter}
   Every non-faulty process eventually adds at least one value to $\ms{bin\_vals}$.
 \end{lemma}
 \begin{proof}
   Note that all non-faulty processes broadcast a \BVALECHO message with
value $v$ when they receive $\floor{\frac{n-q}{2}}-d_r+1$ \BVALECHO messages
with $v$. First, let us consider that $t=d=0$, in that case,
non-faulty processes broadcast a \BVALECHO message with $v$ if they receive
$\floor{\frac{n-q-t}{2}}+1$ \BVALECHO messages with $v$. Also recall that $v\in\{0,1\}$. As such,
let us consider a partition of non-faulty processes $A,\,B\subseteq N$
such that $A\cap B=\emptyset$, and let us consider that processes in
$A$ initially sent a \BVALECHO message with $v=0$ while processes in
$B$ sent a \BVALECHO message with $v=1$. It is clear that $|A|+|B|\geq
n-q-t$ and thus either $|A|\geq \floor{\frac{n-q-t}{2}}+1$ or $|B|\geq
\floor{\frac{n-q-t}{2}}+1$. W.l.o.g. let us assume that $|A|\geq
\floor{\frac{n-q-t}{2}}+1$, then processes in $|B|$ eventually receive enough
\BVALECHO messages with value $v=0$ to also broadcast a \BVALECHO
message with $v=0$. Thus, since $h(d_r)\leq n-q-t-d_r$, eventually
all non-faulty processes receive enough \BVALECHO messages to add at
least the value $0$ to $\ms{bin\_vals}$.

Suppose instead that $d>0$ and $t=0$. Then, if the $d_r\leq d$ deceitful
processes that behave deceitful at a particular phase are enough to prevent termination, this
means that $d_r$ processes have sent at least two conflicting messages
to at least two non-faulty processes. As such, when the timer expires
and non-faulty processes broadcast their received signed \BVALECHO
messages, all non-faulty processes will eventually receive enough \BVALECHO messages to send a \BVALECHO message. Thus, the
case $d>0$ is analogous to the case $d=0$ since \BVALECHO messages are relayed when timer expires, and
we have proven in the previous paragraph that termination is
guaranteed in that case. The same analogy takes place if $t>0$.

Note additionally that if $d_r$ detected deceitful processes have been removed, then the thresholds
decrease by the same factor $d_r$, preserving termination.

\end{proof}

\begin{lemma}[ABV-Uniformity]
  \label{lem:aabv-uni}
  If a non-faulty process $p_i$ adds value $v$ to the set $\ms{bin\_vals}$, then all other non-faulty processes also eventually add $v$ to their local set $\ms{bin\_vals}$.
\end{lemma}
\begin{proof}
  This proof is straightforward: $p_i$ adds $v$ to the set $\ms{bin\_vals}$ if it holds $h(d_r)$ signed \BVALECHO messages with $v$. In that case, it also constructs a certificate $\ms{bv\_cert}$ with these messages and broadcasts $\ms{bv\_cert}$ as part of the \BVALREADY with $v$ before adding $v$ to $\ms{bin\_vals}$. Therefore, all other non-faulty processes will eventually receive $p_i$'s \BVALREADY message along with $\ms{bv\_cert}$ containing enough \BVALECHO messages to also add $v$ to their local $\ms{bin\_vals}$. Finally, recall that all non-faulty processes broadcast their \BVALREADY message before adding $v$ to $\ms{bin\_vals}$, which solves the case that $p_i$ is faulty and sends \BVALREADY only to a subset of the non-faulty processes.
  \end{proof}
\begin{lemma}[ABV-Obligation]
  \label{lem:aabv-obl}
  If $\floor{\frac{n-q-t}{2}}-d_r+1$ non-faulty processes ABV-broadcast a value $v$, then all non-faulty processes ABV-deliver $v$.
\end{lemma}
\begin{proof}
This proof is analogous to that of Lemma~\ref{lem:aabv-ter}.
\end{proof}
\begin{lemma}[ABV-Justification]
  \label{lem:aabv-jus}
  If process $p_i$ is non-faulty and ABV-delivers $v$, then $v$ has been ABV-broadcast by some non-faulty process.
\end{lemma}
\begin{proof}
  Assume first $t=0$ and suppose the contrary: $p_i$ ABV-delivers $v$ and all non-faulty
processes ABV-broadcast $v',\,v\neq v'$. Since benign processes may
either send $v'$ to a subset of the non-faulty processes or nothing at
all, this means that $d-d_r > \floor{\frac{n-q }{2}}-d_r+1$ for
deceitful alone to be able to make $p_i$ ABV-deliver $v$. But using
the bound $d-d_r<n-h(d_r)$ we obtain that $q\geq 2h(d_r)-n$, which
contradicts our assumption on the number of benign faults (i.e. the
bound $q<2h(d_r)-n$). As a
result, it follows that at least some non-faulty process must have ABV-broadcast $v$.
The prove is analogous if $t>0$.
\end{proof}
 \begin{lemma}[ABV-Accountability]
   \label{lem:aabv-acc}
   If process $p_i$ adds value $v$ to $\ms{bin\_vals}$ then associated with $v$ is a valid certificate $\ms{cert}$ from the previous round.
 \end{lemma}
 \begin{proof}
   Since every \BVALECHO and \BVALREADY message without a valid
certificate is discarded, it follows immediately that when a value $v$
is added to $\ms{bin\_vals}$ then $p_i$ has access to a valid
certificate.
   \end{proof}
   \subsubsection{\Mypropertyadj reliable broadcast}
   In this section, we prove the properties of Basilic's reliable broadcast, AARB. We prove AARB-unicity in Lemma~\ref{lem:aarb-uni}, AARB-validity in Lemma~\ref{lem:aarb-val}, AARB-send in Lemma~\ref{lem:aarb-sen}, AARB-Receive in Lemma~\ref{lem:aarb-rec}, AARB-accountability in Lemma~\ref{lem:aarb-acc} and AARB-\myproperty in Lemma~\ref{lem:aarb-aac}.
\begin{lemma}[AARB-Unicity]
  \label{lem:aarb-uni}
  Non-faulty processes AARB-deliver at most one value. 
\end{lemma}
\begin{proof}
  By construction all non-faulty processes AARB-deliver at most one value.
\end{proof}
\begin{lemma}[AARB-Validity]
  \label{lem:aarb-val}
  If non-faulty process $p_i$ AARB-delivers $v$, then $v$ was AARB-broadcast by $p_s$. 
\end{lemma}
\begin{proof}
  Process $p_i$ AARB-delivers $v$ if it receives $h(d_r)$ messages $\langle
\ECHO, v,\cdot,\cdot\rangle$. Non-faulty processes only send an \ECHO message
for $v$ if they receive $\langle \INIT,\,v\rangle$. Thus, since $d+t<h(d_r)$,
$p_s$ AARB-broadcast $v$ to at least one non-faulty process.
\end{proof}
\begin{lemma}[AARB-Send]
  \label{lem:aarb-sen}
  If $p_s$ is non-faulty and AARB-broadcasts $v$, then all non-faulty
processes eventually AARB-deliver $v$.
\end{lemma}
\begin{proof}
  Deceitful processes either broadcast $v$
or multicast $v'$ to a partition $A$ and $v$ to a partition $B$. In
the first case (in which all deceitful behave like non-faulty processes),
since the number of benign and Byzantine processes is $q+t\leq n-h(d_r)$ it follows that
at least $h(d_r)$ non-faulty processes will echo $v$, being that
enough for all processes to eventually AARB-deliver it.

Consider instead some $d_r\leq d+t$ deceitful processes behave deceitful
echoing different messages to two different partitions each containing
at least one non-faulty process. Then when the timer expires and
non-faulty processes exchange their delivered \ECHO messages, all
processes will update their committee removing the $d_r$ detected
deceitful. Thus, since processes also recalculate the thresholds and
recheck them after updating the committee, this case becomes the
aforementioned case where no deceitful process behaves deceitful. The
same occurs if one of the partitions AARB-delivers a value while the
other does not and reaches the timer (Lemma~\ref{lem:aarb-rec}).
\end{proof}
\begin{lemma}[AARB-Receive]
  \label{lem:aarb-rec}
  If a non-faulty process AARB-delivers $v$ from $p_s$, then all non-faulty processes eventually AARB-deliver $v$ from $p_s$.
\end{lemma}
\begin{proof}
  First, since $d+t<2h(d_r)-n$ it follows that deceitful and Byzantine processes can
not cause two non-faulty processes to AARB-deliver different
values (analogously to Lemma~\ref{lem:impagr}). Then, before a process $p_i$ AARB-delivers a value $v$, it
broadcasts a \READY message containing the certificate that justifies
delivering $v$. Thus, when $p_j$ receives that \READY message, it also
AARB-delivers v.
\end{proof}
 \begin{lemma}[AARB-Accountability]
   \label{lem:aarb-acc}
   If two non-faulty processes $p_i$ and $p_j$ AARB-deliver $v$ and $v'$,
respectively, such that $v\neq v'$, then all non-faulty processes
eventually receive PoFs of the deceitful behavior of at least $2h(d_r)-n$ processes (including $p_s$).

\end{lemma}
\begin{proof}
  Non-faulty processes broadcast the certificates of the values they
AARB-deliver, containing $h(d_r)$ signed \ECHO messages from distinct
processes. Therefore, analogous to Lemma~\ref{lem:impagr}, at least
$2h(d_r)-n$ processes must have sent conflicting $\ECHO$ messages, and
they will be caught upon cross-checking the conflicting
certificates. Also, some non-faulty processes must have received
conflicting signed $\INIT$ messages from $p_s$ in order to reach the
threshold $h(d_r)$ to AARB-deliver conflicting messages, meaning that
$p_s$ is also faulty.
\end{proof}
 \begin{lemma}[AARB-\Myproperty]
   \label{lem:aarb-aac}
   The Basilic's AARB protocol satisfies \myproperty.
\end{lemma}
\begin{proof}
  We prove here that if a number of faulty processes send conflicting messages to two
subsets $A,\,B\subseteq N$, each containing at least one non-faulty
process, then:
\begin{itemize}
\item eventually all non-faulty processes terminate without removing the faulty processes,
  or
\item eventually all non-faulty processes receive a PoF for these faulty processes and remove them from the committee, after which, if the source is non-faulty, they terminate.
\end{itemize}

W.l.o.g. we consider just $p_A\in A$ and $p_B\in B$. If they both
terminate despite the conflicting messages, we are finished. Suppose
instead a situation in which only one of them, for example $p_A$,
terminated AARB-delivering a value $v$. Then $p_A$ broadcast a \READY
message with enough $h(d_r)$ \ECHO messages in the certificate
$\ms{cert}$ for $p_B$ to also AARB-deliver $v$ and terminate. Let us
consider w.l.o.g. only one faulty process $p_i$. If a signature from
$p_i$ in $\ms{cert}$ conflicts with a local signature from $p_i$
stored by $p_B$, then $p_B$ constructs and broadcasts a PoF for $p_i$,
and then updates the committee and the
threshold. Then, it rechecks the certificate filtering out the
signature by $p_i$, which would cause $p_B$ to also AARB-deliver $v$
(since the threshold also decreased accordingly).

Suppose neither $p_A$ nor $p_B$ has terminated yet. Then, when the
timer is reached and they both broadcast the \INIT and \ECHO messages
they delivered, they will both be able to construct a PoF for $p_i$,
after which they update the committee and the threshold. Then, if the
source was non-faulty, non-faulty processes can terminate analogously
to the previous case.
\end{proof}

\subsubsection{Basilic binary consensus}
We focus in this section on the properties of Basilic's binary
consensus, AABC. We first prove that if all non-faulty processes
start a round $r$ with the same estimate $v$, then all non-faulty processes
decide $v$ in round $r$ or $r+1$ in Lemma~\ref{lem:aabc-aux}. Then, we prove AABC-\myproperty in Lemma~\ref{lem:aabc-aac}, AABC-agreement in Lemma~\ref{lem:aabc-agr}, AABC-strong validity in Lemma~\ref{lem:aabc-strongva} and AABC-validity as Corollary~\ref{cor:aabc-val} of Lemma~\ref{lem:aabc-strongva}, AABC-termination in Lemma~\ref{lem:aabc-ter}, and AABC-accountability in Lemma~\ref{lem:aabc-acc}. This thus makes AABC the first \mypropertyadj binary consensus protocol, as we show in Theorem~\ref{thm:aabc-cons}.
\begin{lemma}
  \label{lem:aabc-aux}
  Assume that each non-faulty process begins round $r$ with the
estimate $v$. Then every non-faulty process decides $v$ either at the
end of round $r$ or round $r+1$.
\end{lemma}
\begin{proof}
By Lemma~\ref{lem:aabv-obl}, $v$ is eventually delivered to every
non-faulty process. By Lemma~\ref{lem:aabv-jus}, $v$ is the only value
delivered to each non-faulty process. As such, $v$ is the only value in
$\ms{bin\_vals}$ and the only value echoed by non-faulty processes,
since deceitful processes that prevent termination are removed from the
committee when the timer expires (and the threshold is updated). This
means that $v$ will be the only value in $\ms{vals}$. If $v=r\mod
2$ then all non-faulty processes decide $v$. Otherwise, by the same
argument every non-faulty process decides $v$ in round $r+1$.
\end{proof}
We show in Lemma~\ref{lem:aabc-aac} that Basilic's AABC satisfies
AABC-\myproperty.

\begin{lemma}[AABC-\Myproperty]
  \label{lem:aabc-aac}
  Basilic's AABC satisfies \myproperty.
\end{lemma}
\begin{proof}
  We show that if a faulty process $p_i$ sends two conflicting messages to two
subsets $A,\,B\subseteq N$, each containing at least one honest
process, then eventually all honest processes terminate, or instead they
receive a PoF for $p_i$ and remove it from the committee, after which
they all terminate.

First, we observe that no process gets stuck in some round.
Process $p_i$ cannot get stuck in phase $1$ since,
by ABV-Termination (Lemma~\ref{lem:aabv-ter}), every honest process
eventually ABV-delivers a value.

A process also does not get stuck waiting on phase $2$. First, notice
that every value that is included in an \ECHO message from an honest
process is eventually delivered to $\ms{bin\_vals}$. Then, note that all
honest processes eventually deliver $h(d_r)$ \ECHO messages, or
instead, when the timer expires, processes will exchange their \ECHO
messages and be able to construct PoFs and remove $d_r$ deceitful
processes that are preventing termination. In the latter case, after
removing all deceitful processes from the committee and updating the
threshold, they will deliver enough \ECHO messages to terminate
phase $2$, since $h(d_r)\leq n-q-t-d$ for $d_r=d$.

Then, we show that all honest processes always hold a valid certificate
to broadcast a proper message, which could otherwise prevent
termination during the ABV-broadcast in phase $1$. For an
estimate whose parity is the same as that of the finished round $r-1$,
process $p_i$ must have received a valid certificate for the round
(otherwise it would not have terminated such round). If the parity
matches, then it can always construct a valid certificate from the
delivered estimates in round $r-1$.

As a result, all processes always progress infinitely in every round.
Consider the first round $r$ after GST where (i) the coordinator is
honest and (ii) all deceitful processes have been detected and removed by all
honest processes. In this case, every honest
process will prioritize the coordinator’s value, adopting it
as their \ECHO message adding only that value. Hence, every
process adopts the same value, and decides either in round $r$ or
round $r+1$ (by Lemma~\ref{lem:aabc-aux}).

\end{proof}
\begin{lemma}[AABC-Agreement]
  \label{lem:aabc-agr}
  If $d+t\leq 2h-n$, no two non-faulty processes decide different values.
\end{lemma}
\begin{proof}
  W.l.o.g. assume that the non-faulty process $p_i$ decides $v$ in round $r$. This
means that $p_i$ received $h(d_r)$ \ECHO messages in round $r$, and
that $vals=\{v\}$. Consider the \ECHO messages received by non-faulty
process $p_j$ in the same round. If $v$ is in $p_j$'s $vals$
then $p_j$ adopts estimate $v$ because $v=r \mod 2$. If instead
$p_j$'s $vals=\{w\},\,w\neq v$, then $p_j$ received $h(d_r)$ \ECHO messages
containing only $w$.

Analogously to Lemma~\ref{lem:impagr}, it is impossible for $p_j$ and for $p_i$ to receive
$h(d_r)$ \ECHO messages for $v$ and for $w$, respectively. We then conclude, by Lemma~\ref{lem:aabc-aux}, that every non-faulty
process decides value $v$ in either round $r+1$ or round $r+2$.
\end{proof}

\begin{lemma}[AABC-Strong Validity]
  \label{lem:aabc-strongva}
  If a non-faulty process decides $v$, then some non-faulty process proposed $v$.
\end{lemma}
\begin{proof}
  This proof is identical to Polygraph's proof of strong validity~\cite{civit2019techrep,CGG21}.
\end{proof}
\begin{corollary}[AABC-Validity]
  \label{cor:aabc-val}
  If all processes are non-faulty and begin with the same value, then
that is the only decision value.
\end{corollary}

\begin{lemma}[AABC-Termination]
  \label{lem:aabc-ter}
Every non-faulty process eventually decides on a value.
\end{lemma}
\begin{proof}
  This proof derives directly from Lemma~\ref{lem:aabc-aac}.
\end{proof}
\begin{lemma}[AABC-Accountability]
  \label{lem:aabc-acc}
  If two non-faulty processes output
disagreeing decision values, then all non-faulty processes eventually
identify at least $2h-n$ faulty processes responsible for that
disagreement.
\end{lemma}
\begin{proof}
  This proof is identical to Polygraph's proof of
accountability~\cite{civit2019techrep,CGG21}, with the a
generalization to any threshold $h(d_r)$ analogous to the one we make
in Lemma~\ref{lem:aarb-acc}.
\end{proof}

\begin{theorem}
  \label{thm:aabc-cons}
  Basilic's AABC solves the \mypropertyadj binary consensus problem.
\end{theorem}
\begin{proof}
   Corollary~\ref{cor:aabc-val} and
Lemmas~\ref{lem:aabc-agr},~\ref{lem:aabc-aac},~\ref{lem:aabc-ter},
and~\ref{lem:aabc-acc} prove AABC-validity, AABC-agreement,
AABC-\myproperty, AABC-termination and AABC-accountability,
respectively.
  \end{proof}
\subsubsection{General Basilic protocol}

We gather all the results together in this section, showing the proofs for the general Basilic protocol. We prove \myproperty in Lemma~\ref{lem:aac}, validity in Lemma~\ref{lem:val}, termination in Corollary~\ref{cor:ter}, agreement in Lemma~\ref{lem:agr}, and accountability in Lemma~\ref{lem:aac}. Finally, we prove that Basilic solves the \problem problem in Theorem~\ref{thm:consf}.

\begin{corollary}[Termination]
  \label{cor:ter}
  The Basilic protocol satisfies termination.
\end{corollary}
\begin{proof}
  Trivial from Lemma~\ref{lem:aac}.
  \end{proof}

\begin{lemma}[Validity]
  \label{lem:val}
  Basilic satisfies validity.
\end{lemma}
\begin{proof}
  This is trivial by Corollary~\ref{cor:aabc-val} and the proofs of
AARB. Suppose all processes begin Basilic with value $v$. If all
processes are non-faulty then every proposal AARB-delivered was
AARB-sent by a non-faulty process, and since all processes AARB-send
$v$, only $v$ is AARB-delivered.

Since initially processes only start an AABC instance for which they
can propose $1$, this means that eventually all processes start one
AABC instance proposing $1$. By Corollary~\ref{cor:aabc-val}, this
instance will terminate with all processes deciding $1$. Since the rest
of the AABC instances will eventually terminate by Lemma~\ref{lem:aabc-ter},
this means that processes will terminate at least one instance of AABC
outputting $1$. Upon calculating the minimum of all values (which are
all $v$) whose associated bit is set to $1$, all processes will decide
$v$.
\end{proof}

\begin{lemma}[Agreement]
  \label{lem:agr}
  The Basilic protocol satisfies agreement.
\end{lemma}
\begin{proof}
  The proof is immediate having Lemmas~\ref{lem:aabc-agr} and~\ref{lem:aarb-rec}.
\end{proof}

\begin{lemma}[Accountability]
  \label{lem:acc}
  If two non-faulty processes output
disagreeing decision values, then all non-faulty processes eventually
identify at least $2h-n$ faulty processes responsible for that
disagreement.
\end{lemma}
\begin{proof}
  The proof is immediate from Lemmas~\ref{lem:aabc-aac} and~\ref{lem:aarb-aac}.
\end{proof}

We show in Lemma~\ref{lem:aac} that Basilic satisfies
\myproperty. 
\begin{lemma}[\Myproperty]
  \label{lem:aac}
  Basilic satisfies \myproperty.
\end{lemma}
\begin{proof}
  We show that if a faulty process $p_i$ sends two conflicting
messages to two subsets $A,\,B\subseteq N$,
each containing at least one honest process, then eventually all
honest processes terminate, or instead they receive a PoF for $p_i$
and remove it from the committee, after which they all terminate.

  First, analogously to Lemma~\ref{lem:aabc-aac}, all conflicting messages that can be sent
in Basilic are messages of Basilic's AARB or AABC, that already
satisfy \myproperty (see Lemmas~\ref{lem:aabc-aac}
and~\ref{lem:aarb-aac}). This means that if $d_r>0$, then honest
processes eventually update the committee and threshold, after which they recheck
if they hold enough signed messages to terminate. Next, we prove termination. By the
AARB-Send property (Lemma~\ref{lem:aarb-sen}), all honest processes will eventually deliver
the proposals from honest processes. Eventually all honest
processes propose $1$ in all binary consensus whose index corresponds
to an honest proposer, and by AABC-Validity decide $1$. Since
eventually $h(d_r)\leq n-q-d-t$ if enough $d_r$ prevent termination
and are thus detected and removed, we can conclude that at least
$h(d_r)$ binary consensus instances will terminate deciding $1$.

Once honest processes decide $1$ on at least $h(d_r)$ proposals, they propose
$0$ to the rest, and by AABC-Termination (Lemma~\ref{lem:aabc-ter}) all remaining
binary consensus instances will terminate. Next, we show that for
every binary consensus upon which we decided $1$, at least one honest
process AARB-delivered its associated proposal. For the sake of
contradiction, if no honest process had AARB-delivered its associated
proposal, then all honest processes would have proposed $0$, meaning by
AABC-Validity that the final decision of the binary consensus would have
been $0$, not $1$. As a result, by the AARB-Receive property (Lemma~\ref{lem:aarb-rec}),
eventually all honest processes will deliver the proposal for all
binary consensus that they decided $1$ upon. Finally, processes
decide the value proposed by the proposer with the lower index.
\end{proof}
We summarize all proofs in the result shown in
Theorem~\ref{thm:consf} to show that Basilic protocol with initial
threshold $h_0$ solves consensus if $d+t<2h_0-n$ and $q+t\leq
n-h_0$. This result translates in the Basilic class of protocols
solving consensus if $n>3t+d+2q$, as we show in
Corollary~\ref{cor:consf}.

\begin{theorem}
  \label{thm:consf}
  The Basilic protocol with initial threshold $h_0\in(n/2,n]$ solves the
\problem problem if $d+t<2h_0-n$ and $q+t\leq n-h_0$.
\end{theorem}
\begin{proof}
  Corollary~\ref{cor:ter} and Lemmas~\ref{lem:aac},~\ref{lem:val},~\ref{lem:agr}, and~\ref{lem:acc} satisfy termination, \myproperty, validity,
agreement, and accountability, respectively.
%
\end{proof}
\begin{corollary}[Corollary]
  \label{cor:consf}
  The Basilic class of protocols solves the \problem problem if $n>3t+d+2q$.
\end{corollary}
\begin{proof}
  The proof is immediate from Theorem~\ref{thm:consf} after removing $h_0$ from the system of two
inequations defined by $d+t<2h_0-n$ and $q+t\leq n-h_0$.
\end{proof}

\begin{proof}
  The proof is analogous to Theorem~\ref{thm:imp} with the difference
that deceitful processes can actually prevent termination by sending
conflicting messages. Thus, we have $n+t+d\leq 2n-2q-2t-2d$,
which means $n>3(t+d)+2q$.
\end{proof}
\subsection{Basilic's complexity}
\label{sec:comps}
In this section, we show the complexities of Basilic. We execute one instance of
Basilic's AARB reliable broadcast and of Basilic's AABC binary
consensus per process. We prove these complexities in Section~\ref{sec:extcomp}.

\subsubsection{Naive Basilic} We summarize the complexities of the three protocols without optimizations in Table~\ref{tab:gstcom}.
\begin{table}[htp]
  \centering
  \setlength{\tabcolsep}{12pt}
\begin{tabular}{lrrr} 
\toprule 
Complexity & 
AARB & AABC & Basilic \\\midrule
 Time & $\mathcal{O}(1)$ & $\mathcal{O}(n)$ & $\mathcal{O}(n)$ \\
 Message & $\mathcal{O}(n^2)$ & $\mathcal{O}(n^3)$ & $\mathcal{O}(n^4)$ \\
 Bit& $\mathcal{O}(\lambda n^3)$& $\mathcal{O}(\lambda n^4)$& $\mathcal{O}(\lambda n^5)$\\
 \bottomrule
\end{tabular}

\caption{Time, message and bit complexities of naive implementations of AARB, AABC and the
general Basilic protocol, after GST.}
\label{tab:gstcom}
\end{table}

\subsubsection{Optimized Basilic} The complexities of Basilic after GST share the same asymptotic
complexity of other recent works that are not
\mypropertyadj~\cite{civit2019techrep, CGG21, forensics}, some of them not
being accountable either~\cite{CL02}, as we show in
Table~\ref{tab:bigtab}. This is because the adversary cannot prevent
termination of any phase. Thus, after GST, all processes can continue
to the next phase or terminate the protocol by the time the timer for
that phase expires, resulting in an execution equivalent to that of
Polygraph (apart from one additional message broadcast in
ABV-broadcast). In this table, naive Basilic represents the protocol
we show in Algorithm~\ref{alg:gen}, whereas the following row,
multi-valued Basilic, shows the analogous optimizations shown in
Polygraph and applicable to the Basilic protocol as
well~\cite{CGG21}. The rows containing 'superblock' refer to the
result of applying the additional superblock
optimization~\cite{crain2018dbft,CNG21}, which consists on deciding on
the union of all $h(d_r)$ ($\mathcal{O}(n)$) proposals whose
associated AABC instance output $1$, solving SBC (Def.~\ref{def:sbc}) instead of just
consensus. This optimization is only
available to democratic protocols: processes
in which all processes provide an input~\cite{crain2018dbft,CGG21,Marko2019MirBFT,Marko2022ISS,Guathier2019Dispel} (i.e. DBFT, Polygraph and Basilic in Table~\ref{tab:bigtab}), and the
output is the result of combining these inputs. After these
optimizations, the resulting normalized bit complexity (i.e. per
decision) of Basilic is as low as those of other works that are only
accountable and not actively accountable, such as BFT
Forensics~\cite{forensics} or Polygraph~\cite{CGG21}. Furthermore,
since this is the lowest complexity to obtain
accountability~\cite{CGG21}, this means that this is also optimal
in the bit complexity. Note that other optimizations present in other
works, such as the possibility to obtain an amortized complexity of
$\mathcal{O}(\lambda n^2)$ in BFT Forensics per decision after $n$
iterations of the protocol~\cite{TG19}, are orthogonal to our
optimizations, and thus they also apply to Basilic.

An additional advantage of Basilic, as well as of other democratic
protocols~\cite{CGG21,crain2018dbft}, compared to non-democratic
protocols~\cite{forensics,TG19}, is that the distribution of proposals
scatters the bits throughout multiple channels of the network, instead
of bloating channels that have the leader as sender or recipient. That
is, while BFT forensics' normalized an amortized complexity per
network channel is $\Theta(\lambda n)$, as this is the number of bits
that must be sent through the $\Theta(n)$ channels to and from the
leader, Basilic's is instead $\Theta(\lambda)$, which are instead
sent through each of the $\Theta(n^2)$ pairwise channels of the
network.

Finally, not only are the rest of the protocols
in Table~\ref{tab:bigtab} not \mypropertyadj, but also this means that
they only solve consensus tolerating at most $t<n/3$ faults in the BDB
model, whereas Basilic with initial threshold $h_0=2n/3$ solves
consensus where $d+t<n/3$ and $q+t\leq n/3$ faults, hence tolerating
the strongest adversary among these works.

\begin{table}[htp]
  \centering
  \caption{Complexities of Basilic compared to other works.}
  \label{tab:bigtab}
  \setlength{\tabcolsep}{5pt}
  \hspace{-1em}
\begin{tabular}{lllcc}
\toprule
Algorithm & Msgs & Bits & Accountable & Actively accountable \\
\midrule
PBFT~\cite{CL02} & $\mathcal{O}(n^3)$ & $\mathcal{O}(\lambda n^4)$ & \no & \no \\
Tendermint~\cite{Buc16} & $\mathcal{O}(n^3)$ & $\mathcal{O}(\lambda n^3)$ & \no & \no \\
HotStuff~\cite{TG19} & $\mathcal{O}(n^2)$ & $\mathcal{O}(\lambda n^2)$ & \no & \no \\
DBFT superblock~\cite{crain2018dbft} & $\mathcal{O}(n^3)$ & $\mathcal{O}(n^3)$ & \no & \no \\
  
\midrule
  BFT Forensics~\cite{forensics} & $\mathcal{O}(n^2)$ & $\mathcal{O}(\lambda n^3)$ & \yes & \no \\
Polygraph's binary~\cite{CGG21} & $\mathcal{O}(n^3)$ & $\mathcal{O}(\lambda n^4)$ & \yes & \no \\  
Naive Polygraph~\cite{CGG21} & $\mathcal{O}(n^4)$ & $\mathcal{O}(\lambda n^5)$ & \yes & \no \\
Polygraph Multi-v.~\cite{CGG21} & $\mathcal{O}(n^4)$ & $\mathcal{O}(\lambda n^4)$ & \yes & \no \\
  
Polygraph superblock~\cite{CGG21} & $\mathcal{O}(n^3)$ & $\mathcal{O}(\lambda n^3)$ & \yes & \no \\
\midrule
Basilic's AABC & $\mathcal{O}(n^3)$ & $\mathcal{O}(\lambda n^4)$ & \yes & \yes \\
Naive Basilic & $\mathcal{O}(n^4)$ & $\mathcal{O}(\lambda n^5)$ & \yes & \yes \\
Multi-valued Basilic & $\mathcal{O}(n^4)$ & $\mathcal{O}(\lambda n^4)$ & \yes & \yes \\
Basilic superblock & $\mathcal{O}(n^3)$ & $\mathcal{O}(\lambda n^3)$ & \yes & \yes \\
\bottomrule
\end{tabular}
\end{table}
\subsubsection{Complexity proofs}
\label{sec:extcomp} We prove in this section the complexities of Basilic, and of Basilic's AARB and AABC, which we presented in Section~\ref{sec:comps}.
\begin{lemma}[Basilic's AARB Complexity]
  \label{lem:compaarb}
  After GST and if the source is non-faulty, Basilic's AARB protocol has time complexity
$\mathcal{O}(1)$, message complexity $\mathcal{O}(n^2)$ and bit
complexity $\mathcal{O}(\lambda\cdot n^3)$.
\end{lemma}
\begin{proof}
  After GST, all non-faulty processes will have received a message
from each non-faulty process and from each deceitful processes by the time the
timer reaches $0$. Thus, either non-faulty processes can terminate, or they broadcast
their current list of \ECHO and \INIT messages, after which they
remove the detected deceitful processes, and they can terminate
too. Thus, the time complexity is $\mathcal{O}(1)$. Then, the message
complexity is $\mathcal{O}(n^2)$, as each non-faulty process
broadcasts at least one \ECHO and \READY message, and, in some
executions, a list of \ECHO messages that they delivered by the time
the timer reaches $0$. Since both this list and \READY messages
contain $\mathcal{O}(n)$ signatures, or $\mathcal{O}(\lambda n)$ bits,
the bit complexity of Basilic's AARB is $\mathcal{O}(\lambda n^3)$.
\end{proof}
\begin{lemma}[Basilic's AABC Complexity]
  \label{lem:compaabc}
  After GST, Basilic's AABC protocol has time complexity
$\mathcal{O}(n)$, message complexity $\mathcal{O}(n^3)$ and bit
complexity $\mathcal{O}(\lambda\cdot n^4)$.
\end{lemma}
\begin{proof}
  After GST, the Basilic protocol terminates in the first round (i)
whose leader is a non-faulty process and (ii) after having removed
enough deceitful faults so that they cannot prevent termination. Since
$t+d+q<n$, we have that (i) holds in $\mathcal{O}(n)$. As for every
added round in which deceitful faults prevent termination, a non-zero
number of deceitful faults are removed, we have that (ii) holds in
$\mathcal{O}(n)$ as well. This means that Basilic terminates in
$\mathcal{O}(n)$ rounds. In each round during phase $1$ of AABC,
non-faulty processes execute an ABV-broadcast of $\mathcal{O}(n^2)$,
obtaining $\mathcal{O}(n^3)$ messages. The bit complexity is
$\mathcal{O}(\lambda n^4)$ as each message may contain up to two
ledgers of $\mathcal{O}(n)$ signatures, or
$\mathcal{O}(\lambda n)$ bits. The complexities of phase 2 are
equivalent and obtained analogously to those of phase 1, as non-faulty
processes may broadcast $\mathcal{O}(n)$ signatures if
deceitful faults prevent termination of phase 2, or a
certificate if they decide in this round.
\end{proof}
\begin{theorem}
  \label{thm:compgen}
  The Basilic protocol has time complexity $\mathcal{O}(n)$, message complexity $\mathcal{O}(n^3)$ and bit complexity $\mathcal{O}(\lambda\cdot n^5)$.
\end{theorem}
\begin{proof}
  The proof is immediate from Lemma~\ref{lem:compaabc} and
Lemma~\ref{lem:compaarb} since Basilic executes $n$ instances of
AARB followed by $n$ instances of AABC.
\end{proof}

\subsection{Solving eventual consensus with Basilic}
\label{sec:ec}
In this section, we adapt Basilic to solve eventual consensus in the
BDB model, and then prove that the Basilic protocol is resilient
optimal.
The eventual consensus ($\Diamond$-consensus) abstraction~\cite{Dubois2015} captures eventual agreement among all participants. It exports, to every process $p_i$, operations
$\lit{proposeEC}_1, \lit{proposeEC}_2, ...$ that take multi-valued arguments (non-faulty processes propose valid values) and return multi-valued responses. Assuming that, for all $j \in N$,
every process invokes $\lit{proposeEC}_j$ as soon as it returns a response to $\lit{proposeEC}_{j-1}$, the abstraction guarantees that, in every admissible run, there exists $k \in N$, such that the following properties are satisfied:
\begin{itemize}
\item {\bf $\Diamond$-Termination.} Every non-faulty process eventually returns a response to $\lit{proposeEC}_j$
for all $j\in N$.
\item {\bf $\Diamond$-Integrity.} No process responds twice to $\lit{proposeEC}_j$ for all $j\in N$.
\item {\bf $\Diamond$-Validity.} Every value returned to $\lit{proposeEC}_j$ was previously proposed to $\lit{proposeEC}_j$ for all $j\in N$.
\item {\bf $\Diamond$-Agreement.} No two non-faulty processes return different values to $\lit{proposeEC}_j$
for all $j\geq k$.
\end{itemize}

We detail thus \emph{$\Diamond$-Basilic (BEC)}, an adaptation
of Basilic for the $\Diamond$-consensus problem. Process $p_i$ executes $\Diamond$-Basilic
with the following steps:
\begin{enumerate}[leftmargin=* ,wide=\parindent]
\item BEC first executes $\lit{Basilic-gen-propose_{h_0}(}v_i\lit{)}$, whose output is returned by $p_i$ as BEC's output of $\lit{proposeEC}_0$.
\item If $p_i$ finds no disagreement between operations $k$ and $k'$, then for all operations $\lit{proposeEC}_j,\,k'>j\geq k$, the output is that of $\lit{proposeEC}_{j-1}$.
\item If $p_i$ finds a new disagreement at operation $j$ for some index $r\in[0,n-1]$, then:
  \begin{enumerate}[leftmargin=* ,wide=\parindent]
    \item If the disagreement is between AARB-delivered values, BEC resolves it as follows: let $(\EST, \langle u,r\rangle)$ be the value that differs with the locally AARB-delivered value $(\EST, \langle v,r\rangle)$, then, for $\lit{proposeEC}_j$, $p_i$ applies $y=\lit{min(}v,u\lit{)}$ to the disagreeing value. Next, if the output of $\lit{proposeEC}_{j-1}$ was $v$, $p_i$ replaces the AARB-delivered value with $y$, and outputs $y$ instead for $\lit{proposeEC}_{j}$. 
    \item If the disagreement is between values 1 and 0 decided at AABC's protocol, then $p_i$ sets $\ms{bin-decisions}[r]$ to 1. Then, $p_i$ recalculates if the minimum decided value changed after adding this binary decision (i.e., re-execute lines~\ref{line:genuni1}-\ref{line:genuni} of Algorithm~\ref{alg:gen}), and outputs this decision for $\lit{proposeEC}_{j}$.
    \item Finally, $p_i$ broadcasts the values (and certificates) of all the disagreements that $p_i$ has not yet broadcast.
    \end{enumerate}

\end{enumerate}
We show in Theorem~\ref{thm:ec} that $\Diamond$-Basilic with initial threshold $h_0$ solves the $\Diamond$-consensus
problem if $d+t<h_0$ and $q+t\leq n-h_0$, where $t,\,d$ and $q$ are the numbers
of Byzantine, deceitful and benign processes, respectively, and $h_0$ the initial threshold. This means that the $\Diamond$-Basilic class of protocols solves $\Diamond$-consensus for any combination of $t,\,d$ and $q$ Byzantine, deceitful and benign processes, respectively,
such that $2t+d+q<n$, as we show in Corollary~\ref{cor:dec1}.
\begin{theorem}[$\Diamond$-Consensus per threshold]
  \label{thm:ec}
  The $\Diamond$-Basilic protocol with initial threshold $h_0$ solves the $\Diamond$-consensus
problem if $d+t<h_0$ and $q+t\leq n-h_0$.
\end{theorem}
\begin{proof}
  $\Diamond$-Integrity is trivial. The bound $q+t\leq n-h_0$ is proven in
Corollary~\ref{cor:impter}: $\Diamond$-Basilic starts by executing
Basilic, which does not terminate unless $q+t\leq n-h_0$, satisfying
$\Diamond$-Termination. $\Diamond$-Validity derives immediately from Basilic's proof of validity (Lemma~\ref{lem:val}).

We only have left to prove $\Diamond$-Agreement. If $d+t<h_0$ then all
valid certificates contain at least one honest process. This means
that the number of disagreements is finite. Then, since honest
processes broadcast all disagreements they find (and their
corresponding valid certificates), all honest processes will
eventually find all disagreements. Also, all honest processes will find all disagreements of Basilic by its accountability property (Lemma~\ref{lem:aac}). Let us consider that all honest
processes, except $p_i$, have already found and treated all
disagreements (as specified by the $\Diamond$-Basilic protocol). Suppose
that $p_i$ finds the last disagreement at the start of operation
$\lit{proposeEC}_{k-1}$ for some $k>0$. Then, for all $j\geq k$, no
two honest processes return different values to
$\lit{proposeEC}_k$, satisfying $\Diamond$-Agreement.
\end{proof}
\begin{corollary}[$\Diamond$-Consensus]
  \label{cor:dec1}
  The Basilic class of protocols solves $\Diamond$-consensus if $n>2t+d+q$.
\end{corollary}
\begin{proof}
    The proof is immediate from Theorem~\ref{thm:ec} after removing $h_0$ from the system of two inequations defined by $d+t<h_0$ and $q+t\leq n-h_0$.
\end{proof}

\section{The Zero-Loss Blockchain}
\label{sec:zlblockchain}

Having shown our \mypropertyadj Basilic class of protocols, we now present
our Zero-loss Blockchain (ZLB). ZLB is the first blockchain that
tolerates an adversary controlling up to $t_s$ processes trying to
cause a disagreement, while also tolerating instead up to $t_l$
Byzantine processes. ZLB achieves this high level of tolerance by
resolving temporary disagreements and replacing provably fraudulent
processes.
For this purpose, we first detail the
\blockchainlongproblem problem and an additional assumption of the adversary
fitting for the long-lasting nature of ZLB.

In particular, solving the longlasting blockchain problem is to solve
consensus when possible ($n> 3t+d+2q$, Corollary~\ref{cor:consf}), and to recover from a situation
where consensus is violated ($n\leq 3t+d+2q$) by excluding faulty
processes, resolving this violation, and preventing future ones
($n> 3t'+d'+2q$).

\subsection{\blockchainlongproblem}
A \blockchainlongproblem (\blockchainproblem) is a Byzantine fault tolerant SMR that 
allows for some consensus instances to reach a disagreement before fixing the disagreement by merging the branches of the resulting fork and deciding the union of all the past decisions using SBC (Def.~\ref{def:sbc}). As a result, we consider that a consensus instance $\Phi_i$ outputs a set of enumerable decisions $out(\Phi_i)=s_i,\; |s_i|\in\mathds{N}$ that all $n$ processes replicate. We refer to the state of the SMR at the $i$-th consensus instance $\Phi_i$ as all decisions of all instances up to the $i$-th consensus instance.

More formally, an SMR is an \blockchainproblem if it ensures termination, agreement and convergence:

\begin{definition}[Longlasting Blockchain Problem]\label{def:properties} 

An SMR is an \blockchainproblem if all the following properties are satisfied:
\begin{enumerate}[leftmargin=*,wide=\parindent]
\item {\bf Termination:}
For all $k>0$, consensus instance $\Phi_k$ solves eventual consensus.
\item {\bf Agreement:}
If $d+t<2h-n$ when $\Phi_k$ starts, then honest processes executing $\Phi_k$ reach agreement.
\item {\bf Convergence:}
  There is a finite number of consensus instances that solve eventual consensus, after which all consensus instances solve consensus.
\end{enumerate}

\end{definition}

Termination does not imply agreement from among honest processes in
the first output of the same consensus instance, but it implies that
all instances terminate with an output that may change to eventually
reach agreement, whereas agreement is the classic property of
consensus. Convergence guarantees that there is a limited number of
disagreements before reaching agreement.
\subsection{Slowly-adaptive adversary}
\label{sec:slowly-adaptive}
Considering an SMR rather than single-shot consensus requires to
cope with attacks in which the corrupted processes that the adversary
chooses change over time. As a result, we consider that the adversary
that controls these faulty processes is adaptive in that $f$ can
change over time. However, we assume that the adversary is
\emph{slowly-adaptive}~\cite{LNZ16}, as in previous blockchain systems
with dynamic membership~\cite{LNZ16}, in that the adversary
experiences \emph{static periods} during which Byzantine, deceitful, benign and
honest processes remain so. We assume that these static periods are
long enough for honest processes to discover and replace the faulty
processes, for the sake of convergence.
In particular, each static period $\mathfrak{t}$ is assigned a consensus instance
$\Phi_k$ and $\mathfrak{t}$ starts when $\Phi_k$ starts.  To cope with
pipelined consensus instances, $\mathfrak{t}$ may not end exactly when $\Phi_k$
ends, as there can be $\Phi_\ell$, \dots, $\Phi_m$ instances
running at this time, in which case $\mathfrak{t}$ ends (and static period $\mathfrak{t}+1$
starts) as soon as $\Phi_\ell$, \dots, $\Phi_m$ and $\Phi_k$
have all ended with agreement.
\subsection{Pool of process candidates}
As our system will perform a membership change that excludes some
processes and includes new ones, we model all users that can join as
processes in the system by assuming that there exists a large pool of
$m$ users among which at least $2n/3$ are honest users ($m$ can be
much greater than $n$) and the rest are deceitful. This pool simply
intends to represent the lowest possible requirement that from the
entire world of users that will ever be proposed to be included as
processes, at least $2n/3$ honest ones will eventually be proposed by
honest processes. Notice this is a significantly weaker assumption
than assuming, for example, that honest processes always propose
honest users to be included. For simplicity and w.l.o.g., we assume
that no user from this pool is proposed twice if it has been a process
before, within the same static period of the adversary.

\begin{figure*}[ht]
  \centering 
  \includegraphics[scale=0.5]{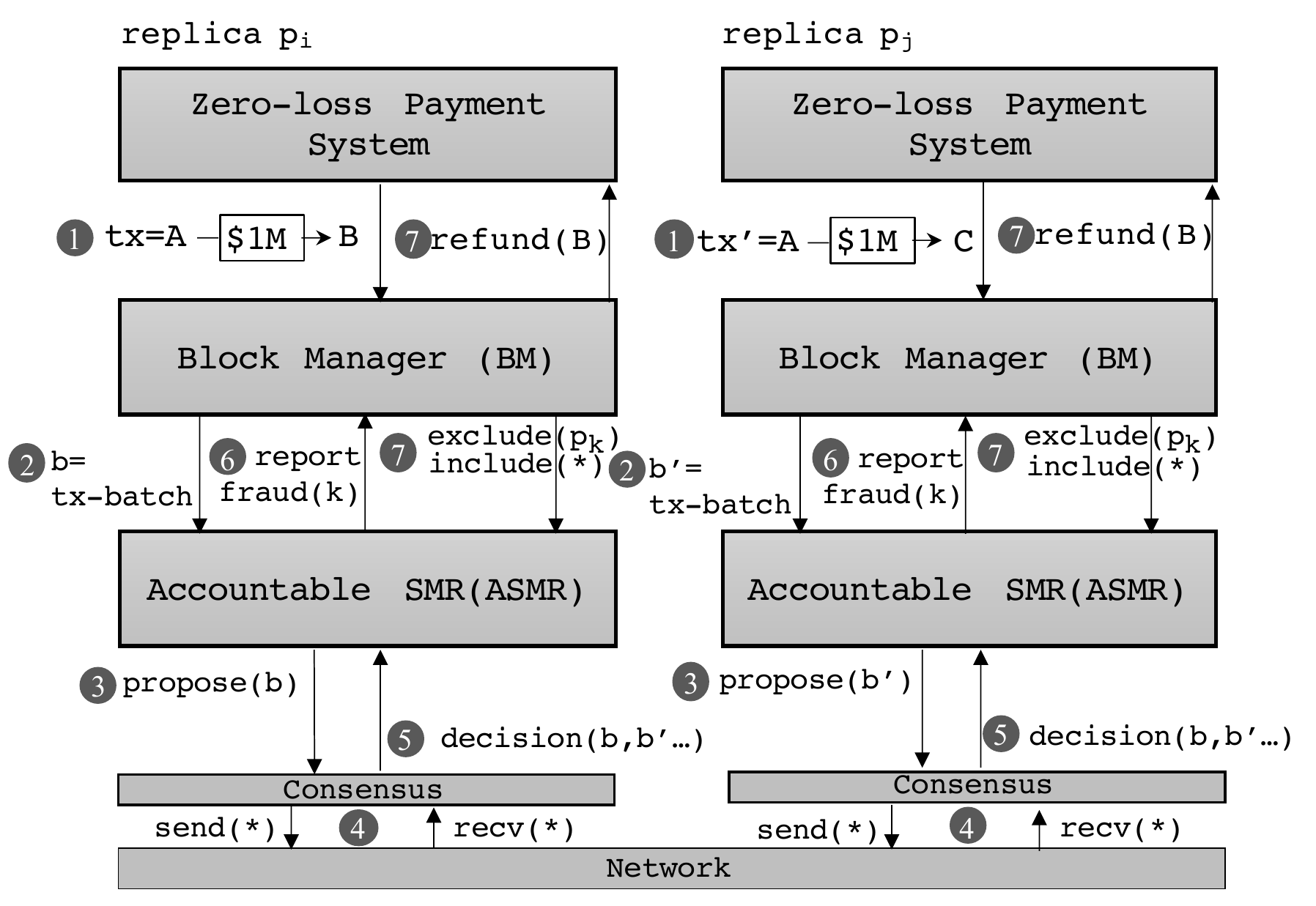} 
\caption[\blockchain architecture.]{The distributed architecture of our \blockchain system relies on \solutionlong (\solution), \component and the payment system. {\color{darkgray}\ding{203}} Each process batches some payment requests illustrated with {\color{darkgray}\ding{202}} a transfer $\ms{tx}$ (resp. $\ms{tx}'$) of \$1M from Alice's account (A) to Bob's (B) (resp. Carol's (C)). Consider that Alice has \$1M initially and attempts to double spend by modifying the code of a process $p_k$ under her control so as to execute a coalition attack. 
{\color{darkgray}\ding{204}--\ding{206}} The \solution component detects the deceitful process $p_k$ that tried to double spend, the associated transactions $\ms{tx}$ and $\ms{tx'}$ and account $A$ with insufficient funds. It uses A's balance to fund transaction $\ms{tx}$, {\color{darkgray}\ding{207}} notifies \component that {\color{darkgray}\ding{208}} excludes or replaces process $p_k$ and {\color{darkgray}\ding{208}} funds $\ms{tx}'$ with $p_k$'s slashed deposit.\label{fig:architecture}
}
\end{figure*}

\subsection{The \blockchainlong}\label{sec:solution}
In this section we detail our system. 
Its two main ideas are (i)~to replace deceitful processes undeniably responsible for a fork by new processes to converge towards a state where consensus can be reached, and (ii)~to refund conflicting transactions that were wrongly decided. 
%
We will show that \blockchain solves the \blockchainlongproblem problem. 
%
As depicted in Figure~\ref{fig:architecture}, we present below the components of our \blockchain system, namely the \solutionlong (\solution) (\cref{ssec:asmr}) and the \componentlong(\component) (\cref{ssec:blockchain}) 
 but we defer the zero loss payment application
(\cref{sec:zlbpayment}).

As long as new requests are submitted by a client to a process, the payment system component of the process converts them into payments that are passed to the \component component. As depicted in Fig.~\ref{fig:architecture}, when sufficiently many 
payment requests have been received, the \component issues a batch of requests to the \solutionlong (\solution) that, in turn, proposes 
it to the consensus component. The consensus component exchanges messages through the network for honest processes to agree. If a disagreement is detected, then the account of the deceitful process is slashed. 
Consider that Alice (A) attempts to double spend by (i)~spending her \$1M with 
both Bob (B) and Carol (C) in $\ms{tx}$ and $\ms{tx}'$, respectively, and (ii)~hacking the code of process $p_k$ that commits deceitful faults to produce a disagreement. Once the \solution{} detects the disagreement, \component is notified, 
process $p_k$ is excluded or replaced and $tx'$ is funded with $p_k$'s slashed deposit.

\subsubsection{\solutionlong (\solution)}\label{ssec:asmr}

In order to detect deceitful processes, we
now present, as far as we know, the first accountable state machine replication, called \solution.
\solution consists of running an infinite sequence of five actions: \ding{192}~the \mypropertyadj consensus (Def.~\ref{def:aac}) that tries to decide upon a new set of transactions, \ding{193}~a confirmation that aims at confirming that the agreement was reached, \ding{194}--\ding{195}~a membership change that aims at replacing detected faulty processes responsible for a disagreement by new processes and \ding{196}~a reconciliation phase that combines all the decisions of the disagreement, as depicted in Figure~\ref{fig:phases}.

\begin{figure}[t]
  \hspace{-4em}
  \includegraphics[scale=0.46]{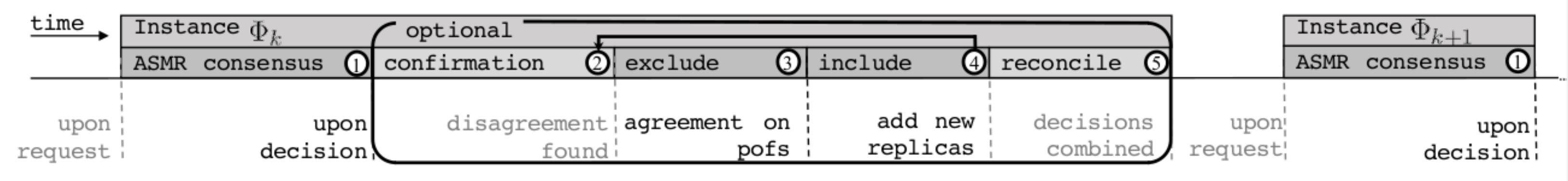}
\caption[The phases of \solution.]{If there are enqueued requests that wait to be served, then a process starts a new instance $\Phi_k$ by participating in an \solution consensus phase \ding{192}; a series of phases may follow: \ding{193}~the process tries to confirm this decision to make sure no other honest process disagrees, \ding{194}~it invokes an exclusion protocol if faulty processes caused a disagreement, \ding{195}~it then includes new processes to compensate for the exclusion, and \ding{196}~merges the two batches of decided transactions. Some of these phases complete upon consensus termination (in black) whereas other phases terminate upon simple notification reception (in grey). The process starts a new instance $\Phi_{k+1}$ without waiting for phases \ding{193}-\ding{196} to terminate, as this is not always guaranteed.\label{fig:phases}
}
\end{figure}

Our \solution's consensus is Basilic, which already removes processes at runtime. However, it is important to note that honest processes do not permanently remove processes removed within a consensus instance of Basilic (in that once the consensus instance terminates, these processes are re-added at the start of the next consensus instance). Honest processes store however the PoFs, and will eventually trigger a membership change that will permanently replace faulty processes by new processes.

\paragraph{The phases of \solution}
For each index, \solution first executes the \mypropertyadj consensus phase \ding{192} to try to agree on a set of transactions. Then, it may terminate four subsequent phases \ding{193}--\ding{196} either to recover from a possible disagreement or to confirm that no disagreement took place.
\begin{enumerate}[leftmargin=* ,wide=\parindent]
\item[\ding{192}] {\bf \solution consensus:} Honest processes propose
a set of transactions, which they received from clients, to an
instance of the Basilic class of \mypropertyadj consensus protocols,
with initial threshold $h_0=2n/3$, in the hope to reach
agreement. When the consensus terminates, all honest processes agree
on the same decision or some honest processes disagree: they decide
distinct sets of transactions.
\item[\ding{193}] {\bf Confirmation:} As honest processes could be temporarily unaware of a disagreement if the adversary controls $d+t\geq n/3$ deceitful and Byzantine processes, 
they enter a confirmation phase waiting for messages coming from more distinct processes than what BFT consensus requires. 
If faulty processes caused a disagreement, then the confirmation terminates and leads honest processes to detect disagreements, 
i.e., honest processes receive certificates supporting distinct decisions.
Otherwise, this phase may not terminate, as an honest process needs to
deliver messages from more than
$(\updelta + 1/3)\cdot n$ processes, where $\updelta$ is the ratio of
potential deceitful faults $\updelta = (d+t)/n$, so as to guarantee that
no disagreement was possible by a \textit{deceitful ratio} $\updelta$. In
particular, with $q=0$ (i.e. $f=t+d$), honest processes
need to receive agreeing messages from $n-x$ processes solving
$\floor{(n-x)/(f-x+1)}=1$, which translates to at least $8n/9$
processes for $f=\ceil{5n/9-1}$, or all processes for $f=\ceil{2n/3}$, as we show in Theorem~\ref{thm:alphaconf}. More specifically, we speak of a decision $v$ being
\textit{$\alpha$-confirmed}, $\alpha\in[0,2/3]$ at $\Phi_k$ if only
an adversary with a deceitful ratio $\updelta\geq \alpha$ could have
caused an honest process to decide $v'\neq v$ at $\Phi_k$.

However, $\Phi_k$ always terminates, as it proceeds in parallel with
the confirmation without waiting for its termination. If the
confirmation phase terminates, it either confirms that a block is
irrevocably final (no process disagreed), or a membership change
starts.

\begin{algorithm}[htp]
  \caption{Membership change at process $p_i$, consensus $\Phi_k$}
  \label{alg:recconsensus}
  \smallsize{
    \begin{algorithmic}[1]
      \Part{\smallsize {\bf State}}{
        \State $\Phi_k$, $k^{\ms{th}}$ instance of ASMR consensus $p_i$ participates to.
        \State $C$, set of processes forming the committee
        \State $C'$, updated set of processes, initially $C'=C$
        \State $\ms{certificates}$, received certificates during exclusion, initially $\emptyset$
        \State $\ms{pofs}$, the set of proofs of fraud (PoFs), initially $\emptyset$
        \State $\ms{new\_pofs}$, set of newly delivered PoFs, initially $\emptyset$
        \State $\ms{cons-exclude}$, the set of PoFs output by consensus, initially $\emptyset$
        \State $\ms{cons-include}$, the set of new processes output by consensus, initially $\emptyset$
        \State $\ms{pool}$, the pool of process candidates from which to propose new processes
        \State $\ms{deceitful}\in I$, the identity of an agreed deceitful process, initially $\emptyset$
        \State $f_d$, the threshold of proofs of fraud to recover, $\lceil n/3 \rceil$ by default
      }\EndPart
      
      \Statex \rule{0.45\textwidth}{0.4pt}
      \Part{\smallsize \textbf{Upon receiving a list of proofs of fraud} $\ms{\_pofs}$}{
        \SmallIf{$\lit{verify(}\_\ms{pofs}\lit{)}$}{} \label{line:verify-pof}\Comment{if PoFs are correctly signed}
        \State $\ms{new\_pofs}\gets \ms{\_pofs}\backslash \ms{pofs}$
        \State $\ms{pofs}.\lit{add(}\ms{new\_pofs}{)}$ \Comment{add PoFs on distinct processes}
        \SmallIf{$\lit{ex-propose}$ \textbf{not started}}{} \label{line:if-ex-propose-not-start}
        \SmallIf{$\lit{size(}\ms{pofs}\lit{)}\geq f_d$}{} \label{line:pof-size}\Comment{enough to change members}
        \SmallIf{$\Phi_k$ started \textbf{and not} finished}{}
        $\Phi_k$.\lit{stop}() \label{line:instance-paused}
        \EndSmallIf
        \State $C'\gets C'\backslash \ms{pofs}.\lit{processes()}$ \label{line:discardstart}
        \State $\lit{ex-propose.update\_committee(}pofs\lit{)}$\label{line:update-committee-1}\Comment{update committee}
        \State $\lit{ex-propose.start(}\ms{pofs}\lit{)}$ \label{line:slashing-consensus-1}\Comment{exclusion consensus}
        
        \EndSmallIf
        \EndSmallIf
        \SmallElseIf{$\ms{new\_pofs}\neq \emptyset$ \textbf{and} $\lit{ex-propose}$ \textbf{not} finished}{$C'\gets C'\backslash \ms{new\_pofs}.\lit{processes()}$} \label{line:upcom1}\label{line:update-committee-2}
        \State $\lit{ex-propose.update\_committee(}new\_pofs\lit{)}$\Comment{update committee}\label{line:upcom2}
        \State $\lit{broadcast(}\ms{new\_pofs}\lit{)}$\Comment{broadcast new PoFs}\label{line:brpofs}
        \State $\lit{ex-propose.check\_certificates(}\ms{certificates}\lit{)}$\Comment{recheck certificates}\label{line:checkcert1}
        \EndSmallElseIf
        
        \EndSmallIf}\EndPart
      \Statex \rule{0.45\textwidth}{0.4pt}
      \Part{\smallsize \textbf{Upon receiving a certificate $\ms{ex-cert}$ of the exclusion protocol}}{
        \SmallIf{$\ms{ex-cert}\not\in\ms{certificates}$ \textbf{and} $\lit{verify\_certificate}(\ms{ex-cert})$}{        $\ms{certificates.add(}\ms{ex-cert}\lit{)}$}\label{line:upcom3}
        
        \EndSmallIf
        \State $\lit{ex-propose.check\_certificates(}\{\ms{ex-cert}\}\lit{)}$\Comment{check certificate with current $C'$}\label{line:checkcert2}
      }\EndPart
      \Statex \rule{0.45\textwidth}{0.4pt}
      \Part{\textbf{function} $\lit{ex-propose.check\_certificates(}\ms{certs}\lit{)}$}{
        \For{all $\ms{cert}\in\ms{certs}$}
        \SmallIf{$\lit{verify\_certificate(}\ms{cert}\lit{)}$}{}
        \SmallIf{$|\ms{cert}\lit{.processes()}\cap C'|\geq \frac{2|C'|}{3}$}{} \label{line:certdec-0} \Comment{current threshold}
        \State $\lit{ex-propose.cert\_decide(}\ms{cert}\lit{)}$ \Comment{decide certificate's decision}\label{line:certdec}
        \EndSmallIf
        
        \EndSmallIf
        
        \EndFor

      }\EndPart
      \Statex \rule{0.45\textwidth}{0.4pt}
      \Part{\textbf{Upon deciding a list of proofs of fraud $\ms{cons-exclude}$ in $\lit{ex-propose}$}}{
        \State $\lit{detected-fraud(}\ms{cons\_exclude}.\lit{get\_processes()}\lit{)}$\label{line:punish} \Comment{application  punishment} 
        \State $\ms{pofs}\gets \ms{pofs} \setminus \ms{cons-exclude.}\lit{get\_pofs()}$ \Comment{discard the treated pofs}
        \State $\ms{C}\gets \ms{C} \setminus \ms{cons-exclude}.\lit{get\_deceitfuls()}$ \Comment{exclude deceitful}\label{line:exclude}
        \State $\ms{inc-prop}\gets \ms{pool}\lit{.take}(|\ms{cons-exclude}|)$\Comment{take processes from the pool}\label{line:pool}
        \State $\lit{inc-propose.start(}\ms{inc-prop}\lit{)}$ \label{line:slashing-consensus-2}\Comment{inclusion cons.}

      }\EndPart
      \Statex \rule{0.45\textwidth}{0.4pt}
      
      \Part{Upon deciding a list of processes to include $\ms{cons-include}$ in $\lit{inc-propose}$}{
        \State $\ms{new\_processes}\gets \lit{choose(}|\ms{cons-exclude}|\lit{,}\ms{cons-include}\lit{)}$\Comment{deterministic}\label{lin:deter}
        \For{all $\ms{new\_process}\in\ms{new\_processes}$}  \Comment{for all new to inc.}
        \State $\lit{set-up-connection}(\ms{new\_process})$ \Comment{new process joins}
        \State $\lit{send-catchup}(\ms{new\_process})$\label{line:cu}\Comment{get latest state}
        \EndFor
        \State $C\gets C\cup \ms{new\_processes}$
        \SmallIf{$\Phi_k$ stopped}{}
        {\bf goto} \ding{192} of Fig.~\ref{fig:phases} \label{line:instance-restarted} \Comment{restart cons.}
        \EndSmallIf
      }\EndPart

    \end{algorithmic}
  }
\end{algorithm}

\item[\ding{194}-\ding{195}]\textbf{Membership change:} 
Our membership change (Alg.~\ref{alg:recconsensus}) consists of two 
consecutive consensus algorithms: one that excludes deceitful processes (line~\ref{line:slashing-consensus-1}), and another that adds newly joined processes (line~\ref{line:slashing-consensus-2}).
We
separate inclusion and exclusion in two consensus instances to avoid
deciding to exclude and include processes proposed by the same
process. 
Process $p_i$ maintains a series of
variables: the current consensus instance $\Phi_k$, the
$\ms{deceitful}$ processes among the whole set $C$  of current process
ids, a set $C'$ of process ids that is updated at runtime for the exclusion protocol, the pool of process candidates
$\ms{pool}$, a set of certificates $\ms{certificates}$, a set of
PoFs $\ms{pofs}$ and of new PoFs
$\ms{new\_pofs}$,
 a local threshold $f_d$ of detected deceitful processes, a
set $\ms{cons-exclude}$ of decided PoFs and a set $\ms{cons-include}$ of
decided new processes.
\begin{enumerate}[leftmargin=* ,wide=2\parindent]
\item[\ding{194}] \textbf{Exclusion protocol:}
If honest processes detect at least $f_d= 2h_0-n$ (i.e. $f_d= n/3$ for $h_0=2n/3$) deceitful processes
(via distinct PoFs), they stop their pending ASMR consensus (line~\ref{line:instance-paused}) before restarting it with the new set of processes (line~\ref{line:instance-restarted}). Then, honest processes start the membership change ignoring
messages from these $f_d$ processes by using instead an updated committee $C'$ that
excludes these processes (lines~\ref{line:discardstart}-\ref{line:slashing-consensus-1}). We fix $f_d=2h_0-n$ for the remaining of this work. Honest processes propose in line~\ref{line:slashing-consensus-1} their set of PoFs at the start of the exclusion protocol $\lit{ex-propose}$ by invoking the Basilic 
\mypropertyadj consensus protocol. We will use $h'(d'_{r})=h'_0-d'_r$
to refer to the voting threshold of the exclusion protocol, and $C'$
to refer to the updated committee of the exclusion protocol. We will discuss specific values of $h'_0$ later in this work.

The key novelty of our exclusion protocol is for processes to exclude
other processes, and thus update their committee $C'$, at runtime upon
reception of new valid PoFs
(lines~\ref{line:upcom1}-\ref{line:upcom2}). Note that Basilic already
removes detected faulty processes at runtime, but starting with $d'_r=f_d$
removed processes gives an advantage to honest processes from the
start. Hence, upon delivering a certificate (line~\ref{line:checkcert2}),
honest processes verify that the certificate contains a threshold
$h'(d'_r)$ of signatures from processes that have not been detected
faulty (line~\ref{line:certdec-0}) and decide the proposals that the
certificate justifies at line~\ref{line:certdec}. Upon updating their committee,
honest processes re-check all their certificates
(line~\ref{line:checkcert1}) and re-broadcast their PoFs
(line~\ref{line:brpofs}). As our exclusion protocol solves the SBC problem (Def.~\ref{def:sbc}),
it maximizes the number of excluded processes by deciding at least
$h'(d'_r)$ proposals at once.

Note that instead of waiting for $f_d$ PoFs (line~\ref{line:if-ex-propose-not-start}), 
processes could start Alg.~\ref{alg:recconsensus} as soon as they detect one deceitful process, however, waiting for at least $f_d$ PoFs guarantees that a membership change is necessary and will help remove many deceitful processes from the same coalition at once.
\item[\ding{195}] \textbf{Inclusion protocol:}
To compensate for the excluded processes, an inclusion protocol $\lit{inc-propose}$ (line~\ref{line:slashing-consensus-2}) adds new candidate processes taken from the pool of process candidates (Section~\ref{sec:slowly-adaptive}) in line~\ref{line:pool}.
This inclusion protocol is also an instance of Basilic with the same voting threshold $h'(d'_r)=h'_0-d'_r$ as the exclusion protocol, but it differs in the format and verification of the proposals: each proposal contains as many new processes as the number of processes excluded (lines~\ref{line:pool}-\ref{line:slashing-consensus-2}). By contrast with the exclusion protocol, the inclusion protocol uses the updated committee ($C$ from line~\ref{line:exclude} onward), resulting from taking the committee from the start of the membership change and excluding from it the decided processes to exclude at the end of the exclusion consensus (line~\ref{line:exclude}).
Since the union of the $h'(d'_r)$ decided proposals contains more than enough processes to include, honest processes apply a deterministic function $\lit{choose}$ (line~\ref{lin:deter}) to the union of all decided proposals. 
This function restores the original committee size to $n$ by selecting the processes evenly from all decided proposals.
This guarantees (i)~a fair distribution of inclusions across all decisions, and (ii)~that 
the deceitful ratio does not increase
even if all included processes are deceitful. At the end, the excluded processes are punished by the application layer (line~\ref{line:punish}) and the new processes are included (lines~\ref{lin:deter}-\ref{line:instance-restarted}).

Honest processes from different partitions might find themselves at
different consensus instances at the moment they execute the
membership change. For this reason, even after the membership change
terminates, there is a transient period where honest processes may
receive blocks with certificates containing excluded processes, that
were decided and broadcast by other honest processes in a different
partition before they executed the membership change. Note, however,
that all certificates contain at least 1 honest process by
construction as long as $f<h$, and thus all honest processes eventually update their
committee and stop generating new certificates with excluded processes.

\end{enumerate}
\item[\ding{196}] {\bf Reconciliation:} Upon delivering a conflicting block with an associated valid certificate, the reconciliation starts by combining all transactions that were decided by distinct honest processes in the disagreement. These transactions are ordered through a deterministic function, whose simple example is a lexicographical order but can be made fair by rotating over the indices of the instances.
\end{enumerate}

Once the current instance $\Phi_k$ terminates, another instance $\Phi_{k+1}$ can start, even if it runs concurrently with a confirmation or a reconciliation at index $k$ or at a lower index.

\subsubsection{\componentlong (\component)}\label{ssec:blockchain}

We now present the \componentlong (\component) 
that builds upon \solution to merge the blocks from multiple branches of a blockchain when forks are detected. 
%
Once a fork is identified, the conflicting blocks are not discarded as it would be the case in classic blockchains when a double spending occurs, but they are merged. Upon merging blocks, \component also copes with conflicting transactions, as the ones of a payment system, by taking the funds of excluded processes to fund conflicting transactions.
We defer to Section~\ref{sec:zlbpayment} 
 the details of the amount processes must have on a deposit to guarantee this funding.

Similarly to Bitcoin~\cite{Nak08}, 
\component accepts transaction requests from a permissionless set of users.
In particular, this allows users to use different devices or wallets to issue distinct transactions withdrawing from the same account---a feature that is not offered in payment systems without consensus~\cite{CGK20}.
In contrast with Bitcoin, but similarly to recent blockchains~\cite{GHM17,CNG21}, our system does not incentivize all users to take part in 
trying to decide upon every block, instead a restricted set of permissioned processes have this responsibility for a given block. This is why \blockchain offers what is often called an open permissioned blockchain~\cite{CNG21}. Nevertheless, \solution can offer a permissionless blockchain with committee sortition~\cite{GHM17} without substantial modifications.

\paragraph{Guaranteeing consistency across processes}
By building upon the underlying \solution that resolves disagreements, \component
features a block merge to resolve forks, along with excluding detected faulty processes and including new processes.
%
A consensus instance may reach a disagreement, resulting in the creation of multiple branches or blockchain forks (Theorem~\ref{thm:nobranches}). 
%
%
\component builds upon the membership change of \solution in order to recover from forks. 
In particular, the fact that \solution excludes
$f_d$ deceitful processes each time a disagreement occurs guarantees that  
the ratio of deceitful processes $\updelta$ converges to a state where consensus is guaranteed %
(Theorem~\ref{thm:conv}). 
The maximum number of branches that can result from forks depends on the number $q$ of benign faults, the number $d$ of deceitful faults and the number of $t$ of Byzantine faults, as well as on the voting threshold, as was already 
shown for histories of SMRs~\cite{singh2009zeno}, and as we restate in Theorem~\ref{thm:nobranches}.

       \begin{algorithm}[ht]
      \caption{Block merge at process $p_i$
      } \label{alg:blkmerg}
     \smallskip
     \smallsize{
      \begin{algorithmic}[1]

		\Part{{\bf \smallsize State}}{
			\State $\Upomega$, a blockchain record with fields: 
			\State \T$\ms{deposit}$, an integer, initially $0$  \label{line:deposit}
			\State \T$\ms{inputs-deposit}$, a set of deposit inputs, initially in the first deposit \label{line:inputs-deposit}
			\State \T$\ms{punished-acts}$, a set of punished account addresses, initially $\emptyset$ \label{line:punished-accounts}
			\State \T$\ms{txs}$, a set of UTXO transaction records, initially in the genesis block 
			\State \T $\ms{utxos}$, a list of unspent outputs, initially in the genesis block
		}\EndPart
		
		\Statex \rule{0.45\textwidth}{0.4pt}
		
		\Part{\smallsize Upon receiving conflicting block $\ms{block}$} { \label{line:safe-binpropose} \Comment{merge block}
			 \For{$\ms{tx}$ \textbf{in} $\ms{block}$} \Comment{go through all txs}
        				\SmallIf{$\ms{tx}$ \textbf{not in} $\Upomega.\ms{txs}$}{} \Comment{check inclusion}
                                        \State $\lit{CommitTxMerge(\ms{tx})}$ \label{line:merge-invocation} \Comment{merge tx, go to line~\ref{line:merge}}
                                 \For{$\ms{out}$ \textbf{in} $\ms{tx.outputs}$}\Comment{go through all outputs}
       				\SmallIf{$\ms{out.account}$ \textbf{in} $\Upomega.\ms{punished-acts}$}{} \Comment{if punished}
        				\State $\lit{PunishAccount(\ms{out.account})}$ \Comment{punish also this new output}
        				\EndSmallIf
                                        \EndFor
        				\EndSmallIf
        			\EndFor
       			\State $\lit{RefundInputs()}$  \Comment{refill deposit, go to line~\ref{line:refund}}
        			\State $\lit{StoreBlock(\ms{block})}$ \Comment{write block in blockchain}
		}\EndPart

      		\Statex  \rule{0.45\textwidth}{0.4pt}
		
		\Part{\smallsize $\lit{CommitTxMerge(\ms{tx})}$}{ \label{line:merge}
        	 		\State $\ms{toFund} \gets 0$
        			\For{$\ms{input}$ \textbf{in} $\ms{tx.inputs}$}\Comment{go through all inputs}
        				\SmallIf{$\ms{input}$ \textbf{not in} $\Upomega.\ms{utxos}$}{}\Comment{not spendable, need to use deposit}
        					\State $\Upomega.\ms{inputs-deposit}.\lit{add(\ms{input})}$ \Comment{use deposit to refund}
        					\State $\Upomega.\ms{deposit}\gets \Upomega.\ms{deposit}-\ms{input.value}$\label{line:deposit-withdrawal} \Comment{deposit decreases in value}
        				\EndSmallIf
        				\SmallElse{} $\Upomega.\lit{consumeUTXO(\ms{input})}$\Comment{spendable, normal case}
        				\EndSmallElse
        			\EndFor\label{line:merge-end}
		}\EndPart
		
      		\Statex \rule{0.45\textwidth}{0.4pt}
		
		\Part{ \smallsize $\lit{RefundInputs()}$}{ \label{line:refund}
        			\For{$\ms{input}$ \textbf{in} $\Upomega.\ms{inputs-deposit}$}\Comment{go through inputs that used deposit}
        				\SmallIf{$input$ \textbf{in} $\Upomega.\ms{utxos}$}{}\Comment{if they are now spendable}
        					\State $\Upomega.\lit{consumeUTXO(}input\lit{)}$\Comment{consume them}
        					\State $\Upomega.\ms{deposit}\gets \Upomega.\ms{deposit}+\ms{input.value}$\Comment{and refill deposit}
        				\EndSmallIf       
			\EndFor
		}\EndPart
		
      \end{algorithmic}
		}
    \end{algorithm}

\paragraph{In memory transactions}\label{ssec:utxo}
\blockchain is a blockchain that inherits the same \emph{Unspent Transaction Output (UTXO)} model of Bitcoin~\cite{Nak08};
the balance of each account in the system is stored in the form of a UTXO table. 
In contrast with Bitcoin, the number of maintained UTXOs is kept to a minimum in order to allow in-memory optimizations. Each entry in this table is a UTXO that indicates some amount of coins that a particular account, the `output' has. When a transaction transferring from source accounts $s_1, ..., s_x$ to recipient accounts $r_1, ..., r_y$ executes, it checks the UTXOs of accounts $s_1, ..., s_x$. If the UTXO amounts for these accounts are sufficient, then this execution consumes as many UTXOs as possible and produces another series of UTXOs now outputting the transferred amounts to $r_1, ..., r_y$ as well as what is potentially left to the source accounts $s_1, ..., s_x$. Maximizing the number of UTXOs to consume helps keeping the table compact.
Each process can typically access the UTXO table directly in memory for faster execution of transactions.

\paragraph{Protocol to merge blocks} As depicted in Alg.~\ref{alg:blkmerg}, the state of the blockchain $\Upomega$ consists of a set of inputs $\ms{inputs-deposit}$ (line~\ref{line:inputs-deposit}), a set of account addresses $\ms{punished-acts}$ (line~\ref{line:punished-accounts}) that have been used by deceitful processes, a $\ms{deposit}$ (line~\ref{line:deposit}), 
that is used by the protocol, a set $\ms{txs}$ of transactions and a list $\ms{utxos}$ of UTXOs. 
The algorithm propagates blocks by broadcasting on the network and starts upon reception of a valid block that conflicts with a known block of the blockhain $\Upomega$ by trying to merge all transactions of the received block with those of the blockchain $\Upomega$ (line~\ref{line:merge-invocation}). This is done by invoking the function $\lit{CommitTxMerge}$ (lines~\ref{line:merge}--\ref{line:merge-end}) where the inputs get appended to the UTXO table and conflicting inputs are funded with the deposit (line~\ref{line:deposit-withdrawal}) of excluded processes. 
We explain in Section~\ref{sec:zlbpayment} 
how to build a payment system with a sufficient deposit 
to remedy successful disagreements.


\paragraph{Cryptographic techniques}
To provide authentication and integrity, transactions are signed using the Elliptic Curves Digital Signature Algorithm (ECDSA) with parameters \texttt{secp256k1}, as in Bitcoin~\cite{Nak08}. Each honest process assigns a strictly monotonically increasing sequence number to its transactions.
The network communications use gRPC between clients and processes and raw TCP sockets between processes, but all communication channels are encrypted through SSL. Finally, the exclusion protocol (Alg.~\ref{alg:recconsensus}) uses ECDSA for authenticating the sender of messages responsible for disagreements (i.e., for PoFs). Unlike ECDSA, threshold encryption cannot be used to trace back the faulty users as they are encoded in less bits than what is needed to differentiate users, and message authentication codes (MACs) are insufficient to provide this transferrable authentication~\cite{CJK12}.

\subsection{The \blockchainlong proofs}
In this section, we prove the properties of \blockchain to solve
\blockchainproblem depending on the voting threshold $h'$ used by the
exclusion and inclusion consensus. We also generalize results to the
voting threshold $h$ of ASMR consensus. Following, we discuss three options
for $h'$, and analyze their advantages and disadvantages, and discuss an
additional desirable property, which we call awareness.

\subsubsection{$\alpha$-Confirmation}
We show in Theorem~\ref{thm:alphaconf} that if a process delivers
$c>n-h+\alpha n$ distinct certificates, then either it confirms that
no coalition of size $\alpha n$ could have caused a disagreement, or
it finds a disagreement.
\begin{theorem}
  \label{thm:alphaconf}
Let $\sigma$ be a consensus protocol with voting threshold $h$, and
let honest process $p_i$ decide $v$ in an iteration of $\sigma$. If
honest process $p_i$ deliver certificates from $c>n-h+\alpha n$ distinct
processes, $\alpha\in[0,2/3]$, then $p_i$ either detects a disagreement or
$\alpha$-confirms $v$.
\end{theorem}
\begin{proof}
$p_i$ delivers certificates from $c>n-h+\alpha n$ processes, meaning that
$c-\alpha n > n-h$ are certificates delivered from honest processes.
As the total number of honest processes is $n-\alpha n$, then $x=n-
\alpha n - (c-\alpha n)=n-c$ are the number of honest processes from which
$p_i$ has not delivered a certificate. For some of these $x$ processes
to have decided $v'\neq v$ then $x+\alpha n\geq h\iff c\leq n-h+\alpha
n$. Thus, if $c>n-h+\alpha n$, either $p_i$ has already received a
certificate for $v'$, or else all honest processes decided on $v$, for
$\updelta \leq \alpha$. In the latter, this means that $v$ is
$\alpha$-confirmed.
\end{proof}

In contrast with Theorem~\ref{thm:alphaconf}, we show in
Theorem~\ref{thm:nobranches} the maximum number of disagreements
(or fork branches) $a$ that an adversary of size $d+t$ can cause in one
consensus instance.

\begin{theorem}[number of branches] Let $\sigma$ be a consensus
  \label{thm:nobranches}
protocol with voting threshold $h$. Suppose $d+t<h$ faulty processes
cause a disagreement, and let $a$ be the number of disagreeing
decisions. Then, $a\leq \frac{n-(d+t)}{h-(d+t)}$ for $d+t\geq
\frac{ah-n}{a-1}$.
\end{theorem} 
\begin{proof}
For $d+t$ faults to be able to create $a$ branches, then they must
reach the voting threshold $h$ with each of $a$ different disjoint
partitions of the honest processes. As these partitions are disjoint,
each contains $(n-(d+t))/a$ processes (assume $n$ divisible by $a$
w.l.o.g.). This means that $\frac{n-(d+t)}{a}+d+t\geq h \iff d+t\geq \frac{ah-n}{a-1}$ for the
attackers to be able to generate $a$ branches. The equivalent equation
in terms of the number of branches is $a\leq \frac{n-(d+t)}{h-(d+t)}$.
\end{proof}

\subsubsection{Exclusion and inclusion protocols}
\label{sec:membchangethresholds}
\begin{theorem}[Consensus of exclusion/inclusion protocol]
Let \blockchain execute with voting threshold $h$, being
$f_d=2h_0-n$ the minimum number of detected processes to start the
membership change. Then the exclusion and inclusion protocols of the membership change solves
consensus if their initial voting threshold $h_0'$ satisfies
$h'_0>\frac{d+t+n}{2}$ for safety and $h'_0\leq n-f+f_d$ for
liveness. \label{theorem:conexcinc}
\end{theorem}
\begin{proof}
  Honest processes start the exclusion protocol by locally excluding
$f_d$ processes that they detected as faulty through accountability. By
Basilic's accountability (Lemma~\ref{lem:acc}), these are at least
$f_d$ detected faulty processes. Let us thus w.l.o.g. assume that
exactly $f_d$ processes are detected and excluded (if it was more,
then consensus is even easier thanks to \myproperty). As such, we define $d'+t'=d+t-f_d$, and $n'=n-f_d$. The exclusion protocol executes an
instance of Basilic with voting threshold $h'_0$, but it actually starts with $h'(f_d)=h'_0-f_d$, since it starts when $f_d$ faulty processes are detected. Thus, this instance of Basilic solves consensus for $h'(f_d)> \frac{d'+t'+n'}{2}$ for safety and $h'(f_d)\leq n'-q-t'$ for liveness (Theorem~\ref{thm:consf}).

We thus consider first the safety bound. $h'(f_d)>
\frac{d'+t'+n'}{2}\iff h'(f_d)> \frac{d+t+n-2f_d}{2}\iff h'_0>
\frac{d+t+n}{2}$, since the membership change starts with the
advantage of having detected $f_d=2h_0-n$ faulty processes before
starting. For the liveness bound, notice that the minimum number of
Byzantine processes that will be detected at the start of the
membership change is $t-t'\geq f_d - d$. Thus, $h'(f_d)\leq n'-q-t' \iff h'(f_d)\leq n-f_d-q-t-d+f_d\iff h'(f_d)\leq n-f\iff h'_0\leq n-f+f_d$.
\end{proof}
\begin{theorem}[$\diamond$-Consensus of exclusion/inclusion protocol]
  Let \blockchain execute with initial voting threshold $h_0$, being
$f_d=2h_0-n$ the minimum number of detected processes to start the
membership change. Then the exclusion and inclusion protocols of the membership change solves
$\Diamond$-consensus if their voting threshold $h_0$ satisfies
$h'_0>d+t$ for safety and $h'_0\leq n-f+f_d$ for liveness.
\label{theorem:evconexcinc}
\end{theorem}
\begin{proof}
  Honest processes start the exclusion protocol by locally excluding
$f_d$ processes that they detected as faulty through accountability. By
Basilic's accountability (Lemma~\ref{lem:acc}), these are at least
$f_d$ detected faulty processes. The exclusion protocol executes an
instance of Basilic with voting threshold $h'$, which solves $\Diamond$-consensus
for $d'+t'<h'(f_d)$ for safety and $h'(f_d)\leq n'-q-t'$ for liveness (Theorem~\ref{thm:ec}), with $n'=n-f_d$ and $d'+t'=d+t-f_d$. 

We thus consider first the safety bound. $d'+t'<h'(f_d)\iff
h'(f_d)>d+t-f_d\iff h'_0>d+t$ by replacing $d'+t'$ by
$d+t-f_d$ and because the membership change starts with the advantage
of having detected $f_d=2h_0-n$ faulty processes before starting. For the
liveness bound, notice that
the minimum number of Byzantine processes that will be detected at the
start of the membership change is $t-t'\geq f_d - d$. Thus, $h'(f_d)\leq n'-q-t'\iff  h'(f_d)\leq n-f_d-q+f_d-t-d\iff h'(f_d)\leq n-f\iff h'_0\leq n-f+f_d$.
\end{proof}

For the standard voting threshold of the ASMR consensus of $h=2n/3$,
this means that there are two different optimal voting thresholds $h'$
for both the exclusion and inclusion protocols, depending on whether
we choose the membership change to solve consensus or eventual
consensus. These thresholds are $h'_0=7n/9=2n'/3$ for consensus and
$h'_0=2n/3=n'/2$ for eventual consensus. We discuss now these two
options, their advantages and drawbacks. We also propose an additional
threshold that is resilient-optimal for an additional property, known as
awareness.
\paragraph{Membership change solving eventual consensus}
The bound $h'_0=2n/3=n'/2$ means that the exclusion and inclusion
protocols solve eventual consensus, as shown by
Theorem~\ref{theorem:evconexcinc}. The advantage of this bound is that
the deceitful ratio $\updelta$ is optimal at $\updelta<2/3$. This is the
optimal value because for $\updelta=2/3=h/n$ the faulty processes can
cause a disagreement without even communicating with honest processes,
meaning that they can cause infinite disagreements, not satisfying
convergence. As such, for this voting threshold $h'$ the total number
of tolerated faults is $f<2n/3$ with
$q+t\leq n/3,\,d+t<2n/3$. Unfortunately, since the exclusion and inclusion
protocols solve only eventual consensus and not consensus, this means
that some processes may temporarily disagree on the processes to
exclude and to include. All processes will however eventually agree on
the same set to include and to exclude, as shown in
Section~\ref{sec:ec}. Furthermore, a disagreement on the exclusion
protocol is detrimental to the adversary, because honest processes
will eventually exclude even more faulty processes. Even a disagreement
of the inclusion protocol is detrimental to the adversary, since the
disagreement on the included processes requires the adversary to expose
even more faulty processes. These attackers could instead wait to expose
themselves as faulty during a disagreement on the ASMR consensus,
which is more beneficial to attackers. Nevertheless, we show
now a different voting threshold $h'$ that solves this vulnerability
of the membership change.

\paragraph{Membership change solving consensus} Setting a voting
threshold $h'_0=7n/9=2n'/3$ allows the exclusion and inclusion
protocols to solve consensus, as shown in
Theorem~\ref{theorem:conexcinc}. Compared to the previous scenario,
this voting threshold allows honest processes to be sure that they
agree on the decisions of the exclusion and inclusion protocols. This
means that if the inclusion protocol includes only honest processes,
then by the end of the membership change the adversary cannot cause
any more disagreements, and agreement is guaranteed from then on,
provided all honest processes have started the membership change.

The disadvantage of such an approach is that the total number of
tolerated faults for this threshold is $f<5n/9$ with $q+t\leq n/3,\,
d+t<5n/9$. Moreover, this voting threshold does not suffice to
guarantee that no disagreement is possible once the membership change
terminates, because some honest processes may not even be aware yet of the existence of a
membership change, and thus may still be using the
outdated committee with $f<5n/9$ faulty processes. For this reason, we
define a new property, which we call awareness.
\begin{definition}[Awareness]
  Suppose that the inclusion protocol only includes honest
processes. Suppose a membership change starts. Then, ZLB
satisfies awareness if all honest processes can fix $k>0$ such that
$\Phi_l$ will solve consensus $\forall l\geq k$ during the static period of the adversary.
\end{definition}
The definition of awareness is strictly stronger than that of
convergence in that it does not suffice for honest processes to know
that eventually they will solve consensus, but they must be aware of
when they stop just solving eventual consensus and start solving
consensus (provided that the inclusion protocol does not restate the
deceitful ratio back to where it was prior to the membership
change). Awareness is also strictly stronger than
$\alpha$-confirmation of the membership change, because awareness also
guarantees that the remaining honest processes that have not yet even
heard of the membership change, and are thus still deciding blocks
with the outdated committee, cannot decide with the outdated committee
after some honest processes terminate the membership change.

\begin{theorem}
  \label{theorem:awareness}
  \blockchain solves awareness if $d+t<h_0+h'_0-n$.
\end{theorem}
\begin{proof}
  Let $O$ be the set of honest processes that have not yet heard of a
membership change. If $|O|+d+t\geq h_0$ then processes in $O$ are
enough to terminate consensus instance $\Phi_k$ with decision $v$,
$k>0$. Let $O'$ be
the set of honest processes that have started the membership
change. Then if $|O'|+d+t\geq h'_0$ the membership change can
terminate, and since $h'_0\geq h_0$ by construction, once processes in $O'$ can terminate the membership change, they can also terminate consensus instance $\Phi_k$ with decision $v',\,v'\neq v$. Thus, we calculate for
which values of $f$ it is impossible for both $|O|+d+t\geq h_0$ and
$|O'|+d+t\geq h'_0$ to be met. By solving the system of equations, for
both to be possible then $d+t\geq h_0+h'_0-n$, which means that for
$d+t< h_0+h'_0-n$ either processes in $O$ can terminate $\Phi_k$ deciding $v$, or processes in $O'$ can terminate deciding $v'$, but not both.
\end{proof}

By Theorem~\ref{theorem:awareness}, a voting threshold
$h'_0=7n/9=2n'/3$, while solving consensus of the exclusion and
inclusion protocols for $f<5n/9$, only satisfies awareness for
$f<4n/9$. Instead, setting a voting threshold $h'_0=5n/6=5n'/9$ solves
consensus of the exclusion and inclusion protocols for $f<n/2$ with
$q+t\leq n/3,\, d+t<n/2$, and awareness for the same adversary. We discuss
these three starting settings of ZLB and compare them with the state
of the art in Section~\ref{sec:comptableZLB}.

\subsubsection{ZLB proofs of LLB}
In this section, we show that
\blockchain solves \blockchainproblem, regardless of the three
possible starting parameters that we showed in the previous section.
\begin{lemma}
  \label{lem:convergenceaux}
  If $d+t<\min(h_0,h'_0)$ and $h'_0\leq n-f+f_d$, then every
disagreement in ZLB leads to a membership change whose inclusion and
exclusion protocols eventually solves consensus.
\end{lemma}
\begin{proof}
  If $d+t\geq h_0$ then faulty processes can cause disagreements without
communicating with honest processes, meaning that disagreements are not
detected and the membership change does not start. Thus, it follows
that $d+t<h_0$. Theorem~\ref{theorem:evconexcinc} shows that $d+t<h'_0$
and $h'_0\leq n-f+f_d$ for the membership change to solve eventual
consensus.
\end{proof}

\begin{theorem}
  \label{thm:conv}
  \blockchain satisfies convergence for $f=d+q+t$ total faults if
$q+t\leq n-h_0$, $d+t<\min(h_0,h'_0)$ and $h'_0\leq n-f+f_d$.
\end{theorem}
\begin{proof}[Proof]
  By Lemma~\ref{lem:convergenceaux} every membership change solves
eventual consensus. The remaining bound $q+t\leq n-h_0$ 
follows from the fact that the ASMR consensus must at least solve eventual
consensus (Lemma~\ref{lem:aac}).  If $d+t<2h_0-n$ from the start then there
is no disagreements (Theorem~\ref{thm:consf}) and thus
convergence is guaranteed. Instead for
$2h_0-n\leq d+t<\min(h'_0,h_0)$ and if $f\leq n-h'_0+f_d$, by
Lemma~\ref{lem:convergenceaux} every disagreement leads to a
membership change that solves eventual consensus.
The inclusion protocol does not increase the deceitful ratio, 
 since the inclusion protocol does not include more processes than the
number of excluded processes by the exclusion protocol (thanks to the
deterministic function $\lit{choose}$) and all excluded
processes are faulty.

As the inclusion consensus decides at least $h'(d_r)=h'_0-d_r$ proposals and
$d+t< h'_0$ (because we implement Basilic to solve SBC by deciding the union of all proposals with associated bit decided to 1), it follows that some proposals from honest
processes will be decided. As the pool of joining candidates is finite
and no process is included more than once, then in the worst case all faulty
processes from the pool have been included at least once, and from then on all honest
processes propose to include only honest processes from the pool. At
this point, the deceitful ratio will decrease in every new membership change, within a static period of the slowly-adaptive adversary.

Some inclusion consensus will thus eventually lead to a deceitful ratio
$\updelta n<2h_0-n$ and consensus is satisfied from then on. Let
$\Phi_k$ be the first ASMR consensus such that $\updelta n<2h_0-n$. All previous iterations $k'<k$ solve eventual consensus because $d+t<h_0$ and $q+t\leq n-h_0$ (Theorem~\ref{thm:ec}).
\end{proof}
\begin{corollary}
  \label{cor:llb}
  \blockchain solves \blockchainlongproblem with $h_0$ for $q+t\leq n-h_0$ and $d+t<h_0$ for any $h'_0$ satisfying $d+t<h'_0$ and $h'_0\leq n-f+f_d$.
\end{corollary}
\begin{proof}
  For $d+t<h_0$, $q+t\leq n-h_0$, by Theorem~\ref{thm:ec} ASMR
consensus solves eventual consensus, satisfying termination. If $d+t<2h_0-n$, then by Theorem~\ref{thm:consf} ASMR consensus solves consensus, satisfying agreement. Finally, convergence is shown in Theorem~\ref{thm:conv}.
\end{proof}

\subsection{Comparative table}
\label{sec:comptableZLB}
We show in Table~\ref{tab:zlbcomparison} a comparison of \blockchain
with the three voting thresholds of the membership changed
mentioned. Notice that Basilic with initial voting threshold
$h_0=\frac{2n}{3}$ is the only one to solve $\Diamond$-consensus
against more than a supermajority of faults, thanks to characterizing
them in the BDB model. We represent three settings for ZLB depending
on the initial voting threshold $h'_0$ of the membership change, but
with the same initial voting threshold of \solution consensus set to
$h_0=\frac{2n}{3}$. These are the three settings that we discussed in
Section~\ref{sec:membchangethresholds}. In any of the three cases,
notice however that Basilic and ZLB are the only to both solve
consensus for a resilient-optimal number of faults of $3t<n$ in the
BFT model, and also solve eventual consensus for a greater number of
faults than $3t<n$. Furthermore, only the three settings of ZLB solve
$\Diamond$-consensus for a total number of faults $3f\geq n$ of which
up to $t=t_\ell$ are Byzantine (as noted by the table's footnotes) while
simultaneously solving consensus for the resilient-optimal bound of
$3t<n$ in partial synchrony. The difference of the three settings of
ZLB lie in the awareness column, as discussed in
Section~\ref{sec:membchangethresholds}. Some works are represented in
multiple rows~\cite{ebbnflow,MNR19} because their tolerance to faults and assumptions
varies depending on their starting configuration.

\begin{table*}[tp]
  \hspace{-6em}
  \footnotesize{
  \begin{tabular}{l|c|cc|cc|cc|cccc}  
    \hline
    & & \multicolumn{2}{c|}{Consensus} & \multicolumn{2}{c|}{$\diamond$-Consensus} &&&&&\\
    Blockchain & N. & Byz. & Total & Byz. & Total & \blockchainproblem & Awareness & Acc. & Slashing & Z.l. & Act.\\
    \hline
    \cite{Nak08,Eth2} & S. & 0 & 0 & $2t<n$ &$2f<n$ & \no & \no &\yes \cite{Eth2} & \yes \cite{Eth2} & \no & \no\\
    \cite{BBC19,CNG21} & P. & $3t<n$ & $3f<n$ & $3t<n$ & $3f<n$ &\no &\no & \no & \no & \no   & \no\\
     \cite{ebbnflow} (P)& P. & $3t<n$ & $3f<n$ & $3t<n$ & $3f<n$ &\no &\no & \no & \no & \no   & \no\\
    \cite{ebbnflow} (S) & S. & $2t<n$ & $2f<n$ & $2t<n$ & $2f<n$ & \no &\no &\no & \no & \no   & \no\\
    \cite{MNR19} (1) & P. & $3t<n$ & $3f<n$ & $3t<n$ & $3f<n$ &\no & \no & \no &\no & \no   & \no\\
    \cite{MNR19} (2) & P. & 0 & $f<\frac{2n}{3}$ & 0 & $f<\frac{2n}{3}$ &\no & \no &\no  & \no & \no  & \no\\
    \cite{dbft, Fac19}. & P. & $3t<n$ & $3f<n$ & $3t<n$ & $3f<n$ &\no & \no &\no & \yes & \no   & \no\\
    \cite{BKM18,forensics,shamis2021pac} & P. & $3t<n$ & $3f<n$ & $3t<n$ & $3f<n$ &\no &\no & \yes & \yes  & \no  & \no\\
    \cite{CGG21} & P. & $3t<n$ & $3f<n$ & $3t<n$ & $3f<n$ & \no &\no &\yes & \no & \no   & \no\\
    \cite{anceaume2020finality} & P. & $0$ & $0$  & $2t<n$ & $2f<n$ & \no& \no &\yes & \no & \no  & \no\\
    \hline
    Basilic ($h_0=\frac{2n}{3})$ & P. & $3t<n$ & $f<\frac{2n}{3}$$^\dagger$ & $3t<n$ & $f<n$$^\ddagger$ & \no & \no & \yes & \no & \no & \yes\\
    ZLB ($h'_0=\frac{7n}{9}$) & P. & $3t<n$ & $f<\frac{2n}{3}$$^\dagger$ & \multicolumn{2}{l|}{$3(q+t)<n$, $d+t<\frac{5n}{9}^{\mathsection}$} & \yes$^\mathparagraph$ & $d+t<\frac{4n}{9}$ & \yes & \yes & \yes & \yes\\
    ZLB ($h'_0=\frac{2n}{3}$) & P. & $3t<n$ & $f<\frac{2n}{3}$$^\dagger$ & \multicolumn{2}{l|}{$3(q+t)<n,$ $d+t<\frac{2n}{3}^\mathsection$} & \yes & \no & \yes & \yes & \yes& \yes\\
    ZLB ($h'_0=\frac{5n}{6}$) & P. & $3t<n$ & $f<\frac{2n}{3}$$^\dagger$ & \multicolumn{2}{l|}{$3(q+t)<n,$ $2(d+t)<n^{\mathsection}$} & \yes$^\mathparagraph$ & $2(d+t)<n$ & \yes & \yes & \yes& \yes\\
    \hline
  \end{tabular}
  \begin{flushleft}
  $\dagger$ Actually, $3(d+t)<n$ and $3(q+t)<n$, meaning that if $f=\ceil{\frac{2n}{3}}-2$ then $t=0$ (Figure~\ref{fig:fig32})\\
  $\ddagger$ Actually, $d+t<\frac{2n}{3}$ and $3(q+t)<n$, meaning that if $f=n-2$ then $t=0$ (Theorem~\ref{thm:ec})\\
  $\mathsection$ In addition to $\ddagger$ as it implements on top of Basilic\\
  $\mathparagraph$ Membership change solves consensus not just $\Diamond$-consensus
  \end{flushleft}  
  }
  \caption[Comparative table of ZLB with previous work.]{Comparative table of ZLB with previous work, where N. means
the network assumption (S. for synchrony and P. for partial
synchrony), Byz. means Byzantine faults tolerated, Acc.
accountability, Z.l. zero loss, and Act. active accountability.}
  \label{tab:zlbcomparison}
\end{table*}
\subsection{Experimental evaluation}
\label{sec:exper}
This section answers the following: Does \blockchain offer practical\ performance in a geo-distributed environment?
When $f<n/3$, how does \solution perform compared to the HotStuff state machine replication that inspired Facebook Libra~\cite{BBC19} and the recent fast Red Belly Blockchain~\cite{CNG21}?
What is the impact of large scale coalition attacks on the recovery of \solution?  
We defer the evaluation of a zero-loss payment application to Section~\ref{sec:zlbpayment}. 

\mypar{Selecting the right blockchains for comparison.}
As we offer a solution for open networks, we cannot rely on the synchrony assumption made by 
other blockchains~\cite{GHM17}. 
As we need to reach consensus, we have to assume an unknown bound on the delay of messages~\cite{DLS88}, and do not compare against randomized blockchains~\cite{MSC16,DRZ18,GLT20,LZT20} whose termination is yet to be proven~\cite{TG19}.
This is why we focus our evaluation on partially synchronous blockchains.
We thus evaluated Facebook Libra~\cite{BBC19}, however, its performance was limited to 11 transactions per second, seemingly due to its Move VM overhead. Hence, we omit these results here and focus on its raw state machine replication (SMR) algorithm, HotStuff and its available C++ code that was previously shown to lower communication complexity of traditional BFT SMRs~\cite{YMR19} (we use the unchanged original implementation in its default configuration~\cite{YMR19b}).
We also evaluate the recent scalable Red Belly Blockchain~\cite{CNG21} (RBB), and the Polygraph protocol~\cite{CGG21} as it is, as far as we know, the only implemented accountable consensus protocol. Nevertheless, this protocol does not tolerate more than $n/3$ failures as it cannot recover after detection. 

\mypar{Geodistributed experimental settings.}
We deploy the four systems 
in two distributed settings of c4.xlarge Amazon Web Services (AWS) instances equipped with 4 vCPU and 7.5\,GiB of memory: (i)~a LAN with up to 100 machines 
and (ii)~a WAN with up to 90 machines.
We evaluate \blockchain with a number of failures $f$ up to $\lceil
\frac{2n}{3} \rceil - 1$, however, when not specified we fix
$f=d=\lceil 5n/9\rceil-1$ and $q=0$. Since the impact of selecting a
different $h'_0\in[2n/3,n]$ is negligible in terms of throughput, we
fix for this section $h'_0=7n/9$. Notice that for this threshold we
can actually tolerate $d<2n/3$ deceitful faults, provided that they
are all detected at the start of the membership change,
i.e. $f_d=2n/3$. This can happen if all attackers collude together to
maximize the number of branches that they can cause a disagreement
for, as we do in our attack.
All error bars represent the 95\% confidence intervals and the plotted values are averaged over 3 to 5 runs.  
All transactions are $\sim{400}$-byte Bitcoin transactions with ECDSA signatures~\cite{Nak08}.

\begin{figure}[t]
  \centering
  \includegraphics[width=.8\textwidth]{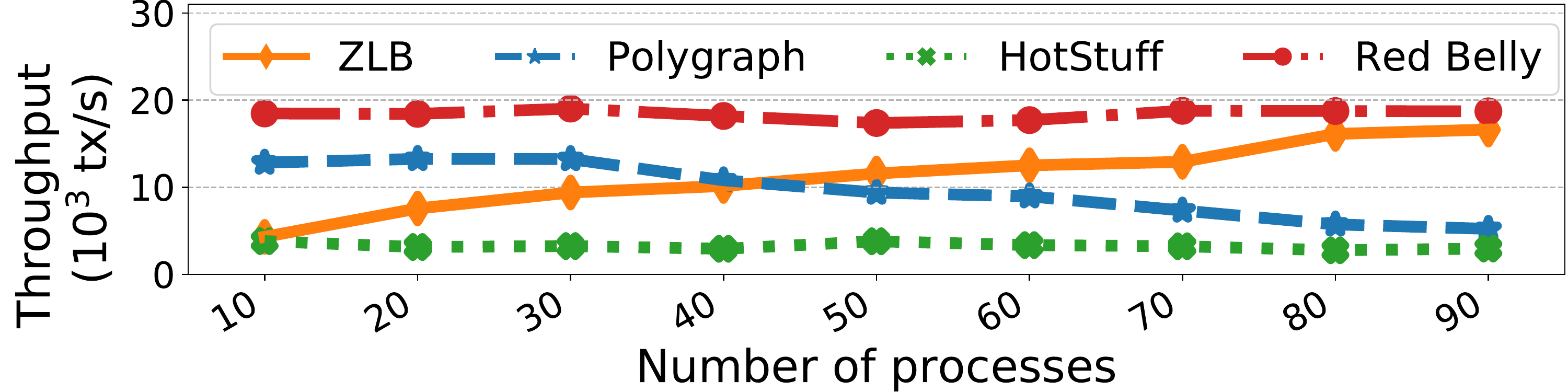}
  \caption[Throughput of \blockchain compared to recent works.]{Throughput of \blockchain 
  compared to that of Polygraph~\cite{CGG21}, HotStuff~\cite{YMR19} and Red Belly Blockchain~\cite{CNG21}.} 
  \label{fig:fig1}
  \vspace{-1.5em}
\end{figure}

\subsubsection{\blockchain vs. HotStuff, Red Belly and Polygraph}
Figure~\ref{fig:fig1} compares the performance of \blockchain, RBB, Libra and Polygraph deployed over 5 availability zones of 2 continents, California, Oregon, Ohio, Frankfurt and Ireland (exactly like the Polygraph experiments~\cite{CGG21}).
For \blockchain, we only represent the decision throughput that reaches $16,626$ tx/sec at $n=90$ as the confirmation throughput is similar ($16,492$ tx/sec). As only \blockchain tolerates $f\geq n/3$, we fix $f=0$ for this comparison. 

First, Red Belly Blockchain offers the highest throughput.
As expected it outperforms \blockchain due to its lack of accountability: it does not require messages to piggyback certificates to detect PoFs. 
Both solutions solve SBC so that they decide more transactions (txs) as the number of proposals enlarges and use the same batch size of $10,000$ txs per proposal.
As a result \solution scales pretty well: the cost of tolerating $f\geq n/3$ failures even appears negligible at $90$ processes.

Second, HotStuff offers the lowest throughput even if it does not verify transactions. 
Note that HotStuff is benchmarked with its dedicated clients in their default configuration, they transmit the proposal to all servers to save bandwidth by having servers exchanging only a digest of each transaction. 
The performance is explained by the fact that HotStuff decides one proposal per consensus instance (i.e. one batch of $10,000$ txs), regardless of the number of submitted transactions, which is confirmed by previous observations~\cite{VG19}.
By contrast, \blockchain becomes faster as $n$ increases to outperform
HotStuff by $5.6\times$ at $n=90$, thanks to the superblock
optimization that allows \blockchain to decide multiple
proposals at once per instance of its multi-valued consensus~\cite{CNG21}.

Finally, Polygraph is faster at small scale than \blockchain, because Polygraph's distributed verification and reliable broadcast implementations~\cite{CGG21} are not accountable, performing less verifications. 
After 40 processes, Polygraph becomes slower because of our optimizations: e.g., its RSA verifications are larger than our ECDSA signatures and consume more bandwidth. 

\begin{figure}[t]
  \centering
  \includegraphics[width=.8\textwidth]{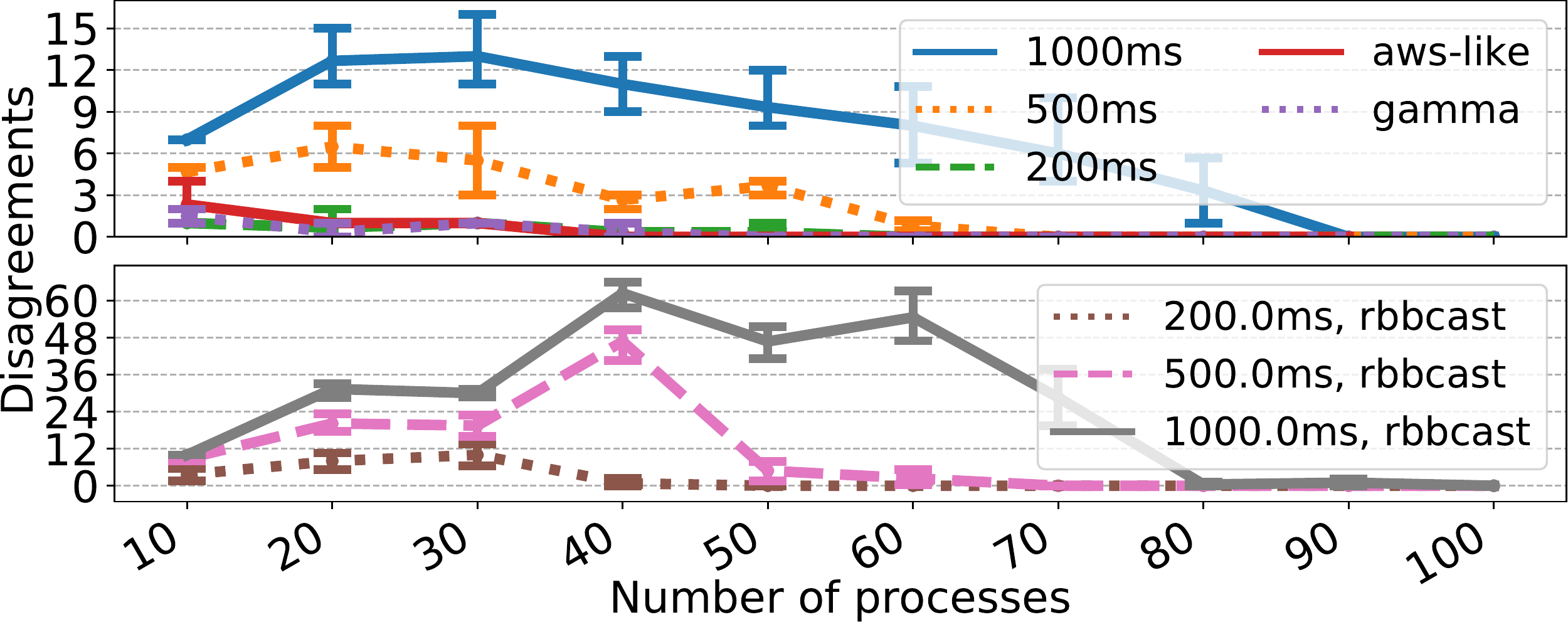}
  \caption[Disagreements caused by attackers with artificial delays.]{Disagreeing decisions for various uniform delays and for delays generated from a Gamma distribution and a distribution that draws from observed AWS latencies, when equivocating while voting for a decision (top), and while broadcasting the proposals (bottom), for $f=d=\lceil 5n/9\rceil-1$.}
  \label{fig:fig2}
  \vspace{-1.5em}
\end{figure}

\subsubsection{Scalability of \blockchain despite coalition attacks}
To evaluate \blockchain under failures, we implemented the following two possible coalition attacks. In the solution to the SBC problem (Def.~\ref{def:sbc}), faulty processes can form a coalition of $f\geq n/3$ processes to lead honest processes to a disagreement by sending conflicting messages, with one of two coalition attacks:
\begin{enumerate}[leftmargin=*,wide=\parindent]
\item {\bf Reliable broadcast attack:} faulty processes misbehave during the reliable broadcast by sending different proposals to different partitions, leading honest processes to end up with distinct proposals at the same index $k$. For example, faulty processes send block $b_a$ with transaction $tx_a$ to a subset $A$ of honest processes, while block $b_b$ with conflicting transaction $tx_b$ to a subset $B$ of honest processes, $A\cap B=\emptyset$, both at the same index $k$. 
\item  {\bf Binary consensus attack:}  faulty processes vote for each binary value in each of two partitions for the same binary consensus leading honest processes to decide different bits in the same index of their bitmask, where deciding $1$ (resp. 0) at bitmask index $k$ means to include (resp. not include) proposal at index $k$ in ZLB. For example, faulty processes send messages to decide $1$ and $0$ to a subset of honest processes $A$, while they send messages to decide $0$ and $1$ to a subset $B$ of honest processes, with $A\cap B=\emptyset$, on the binary consensus instances associated to block $b_a$ with transaction $tx_a$ and block $b_b$ with conflicting transaction $tx_b$, respectively.
\end{enumerate}

Note that faulty processes do not benefit from combining these attacks: If two honest processes deliver different proposals at index $k$, the disagreement comes from them outputting 1 at the corresponding binary consensus instance. Similarly, forcing two honest processes to disagree during the $k$-th binary consensus only makes sense if they both have the same corresponding proposal at index $k$.

To disrupt communications between partitions of honest processes, we inject random communication delays between partitions based on the uniform and Gamma distributions, and the AWS delays obtained in previously published measurements traces~\cite{mukherjee1992dynamics,crovella1995dynamic,CNG21}.   
(Attackers communicate normally with each partition.) 

Fig.~\ref{fig:fig2}(top) depicts the amount of disagreements as the number of distinct proposals decided by honest processes, caused by the binary consensus attack.
First, we select uniformly distributed delays between the two partitions of 200, 500 and 1000 milliseconds.
Then, we select delays following a Gamma distribution with parameters taken from~\cite{mukherjee1992dynamics,crovella1995dynamic} and a distribution that randomly samples the fixed latencies previously measured between AWS regions~\cite{CNG21}. 
We automatically calculate the maximum amount of branches that the size of deceitful faults can create (i.e., 3 branches for $d<5n/9$), we then create one partition of honest processes for each branch, and we apply these delays between any pair of partitions.

Interestingly, we observe that our agreement property is scalable: the greater the number of processes (maintaining the deceitful ratio), the harder for attackers to cause disagreements. This scalability phenomenon is due to an unavoidable increase of the communication latency between attackers as the scale enlarges, which gives relatively more time for the partitions of honest processes to detect the deceitful processes, hence limiting the number of disagreements. With more realistic network delays (Gamma distribution and AWS latencies) that are lower in expectation than the uniform delays,  
deceitful processes can barely generate a single disagreement.  
This confirms the scalability of our system.

Fig.~\ref{fig:fig2}(bottom) depicts the amount of disagreements under the reliable broadcast attack. The number of disagreements is substantially higher during this attack than during the binary consensus attack. However, it drops faster as the system enlarges, 
because the attackers expose themselves earlier.
\begin{figure}[tp]

  \subfloat[]{
        \includegraphics[width=.49\textwidth]{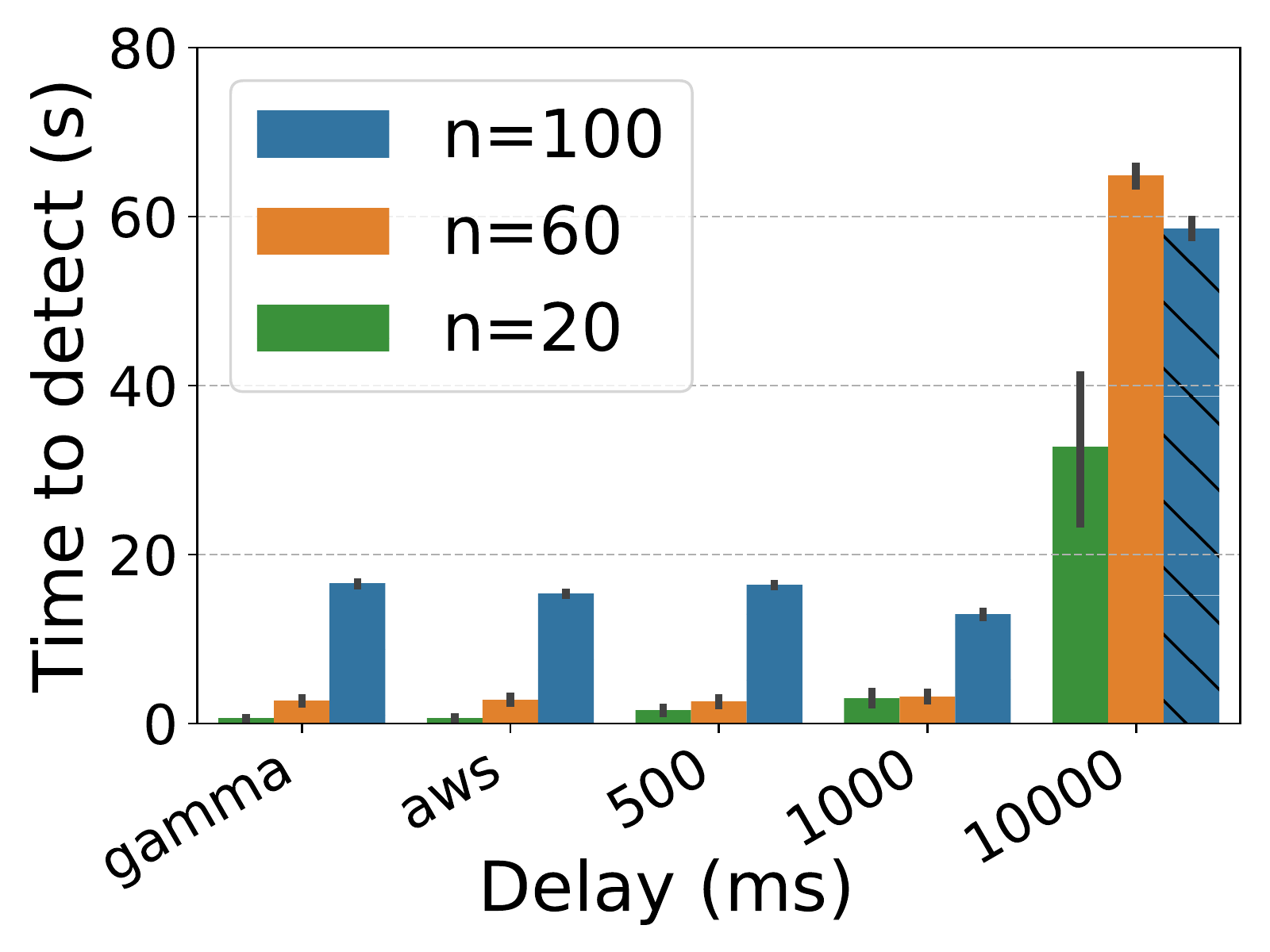}}\hfill
        \subfloat[]{
        \includegraphics[width=.49\textwidth]{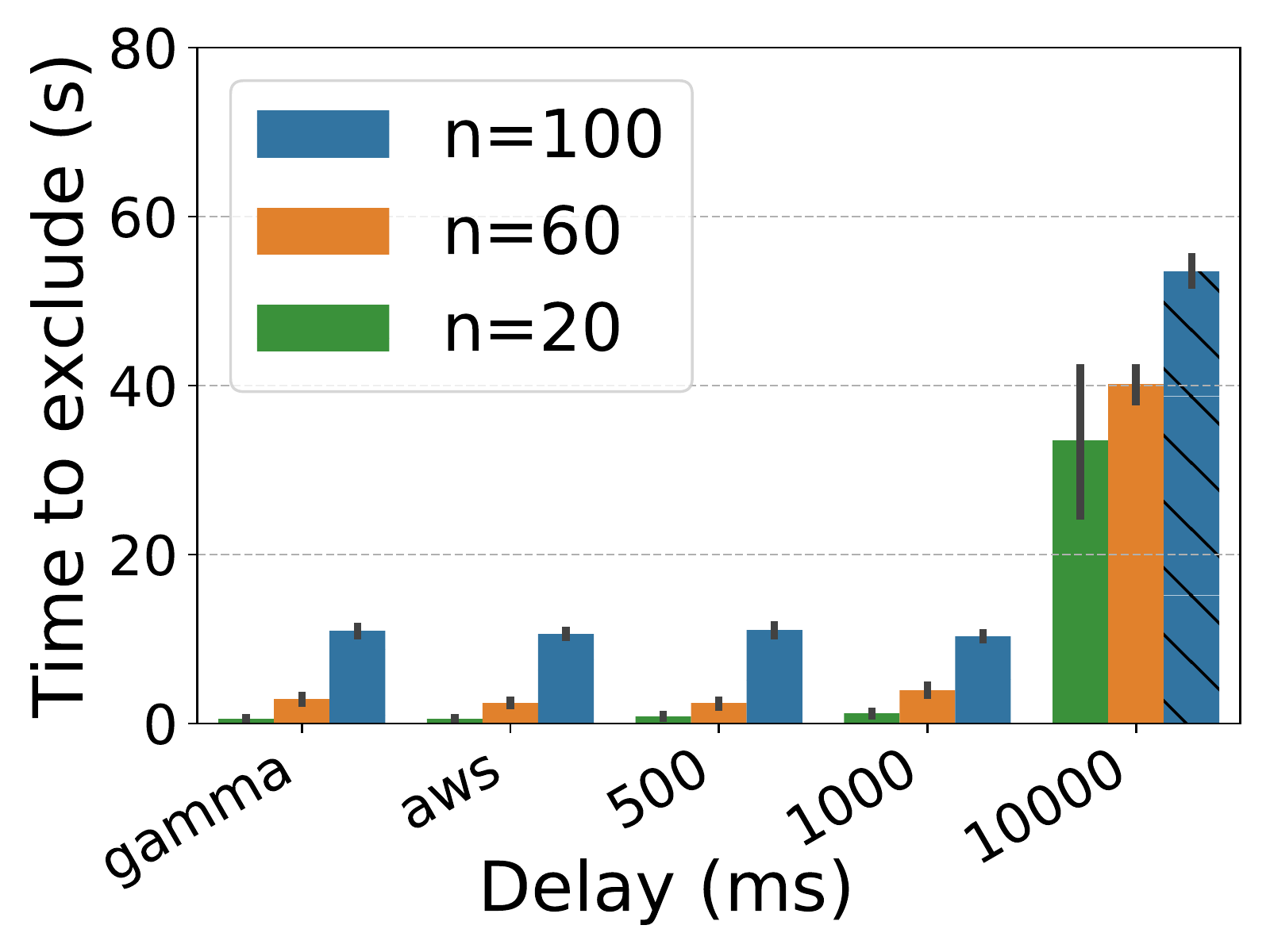}}\hfill
      \subfloat[]{
        \includegraphics[width=.49\textwidth]{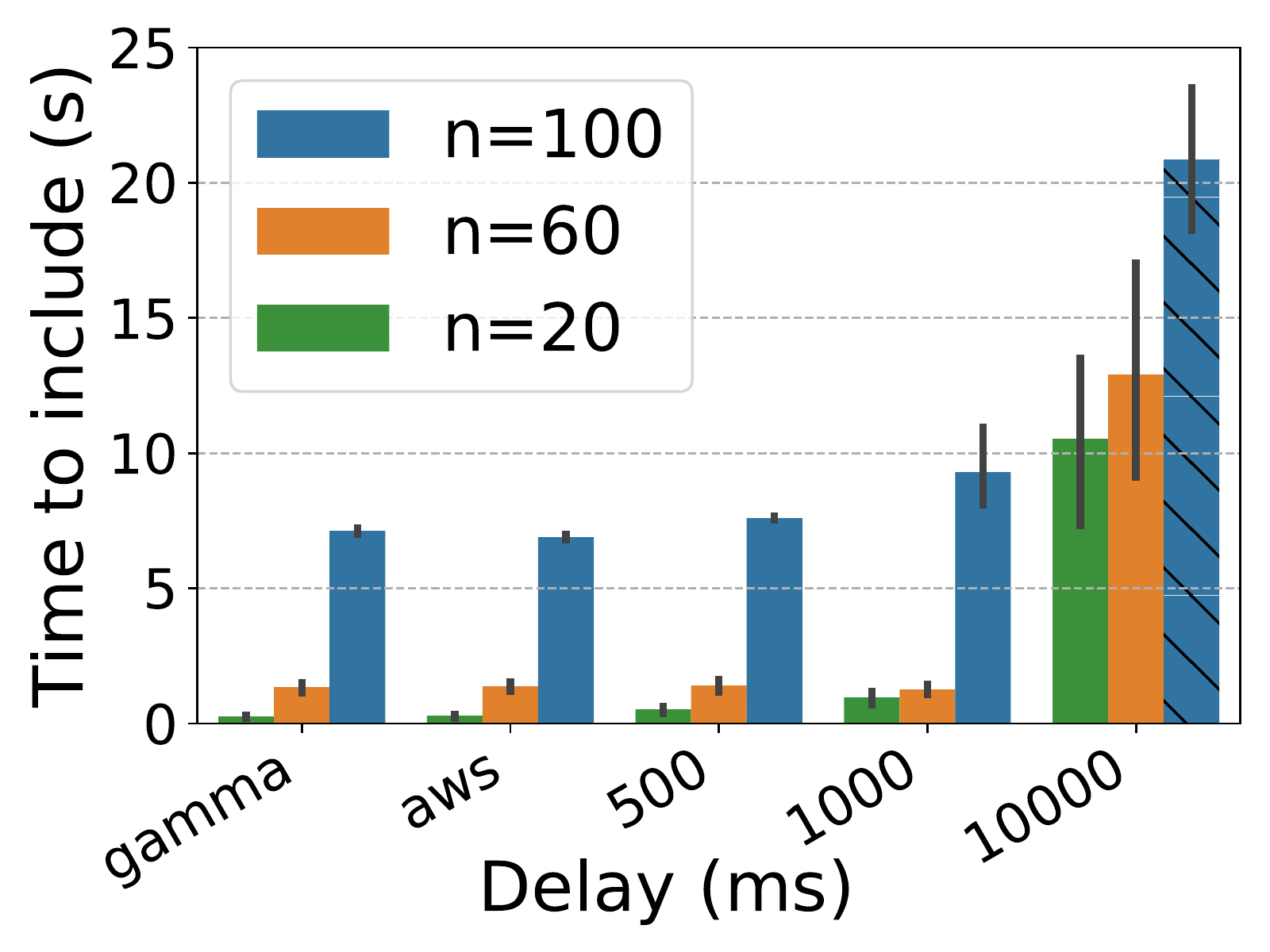}}\hfill
        \subfloat[]{
        \includegraphics[width=.49\textwidth]{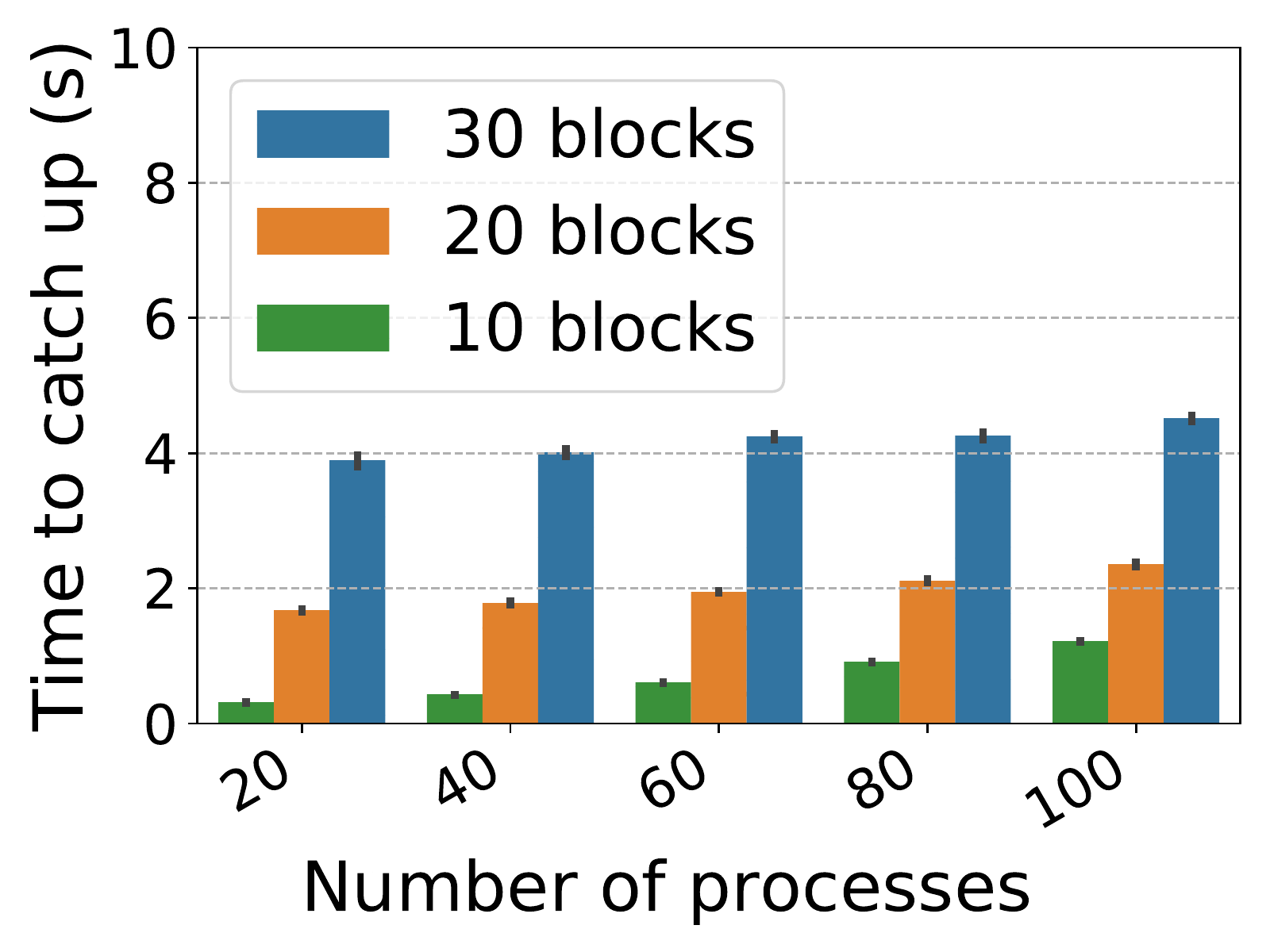}}\hfill
      \caption[Time to detect and exclude deceitful processes, include new processes, per delay distribution and number of processes; and catch up per number of blocks and processes.]{(Left to right, top to bottom) Time to detect $\ceil{\frac{n}{3}}$ deceitful processes, exclude them, include new processes, per delay distribution and number of processes; and catch up per number of blocks and processes, with $f=d=\lceil 5n/9\rceil-1$.}
      \label{fig:fig7}
    \end{figure}
    \subsubsection{Disagreements due to failures and delays}
We now evaluate the impact of even larger coalitions and delays on \blockchain. We measure the number of disagreements as we increase the deceitful ratio and the partition delays in a system from 20 to 100 processes. 
Note that these delays could be theoretically achieved with man-in-the-middle attacks, but are notoriously difficult on real blockchains due to direct peering between the autonomous systems of mining pools~\cite{EGJ18}.

\begin{figure}[t]
  \centering
  \includegraphics[width=.8\textwidth]{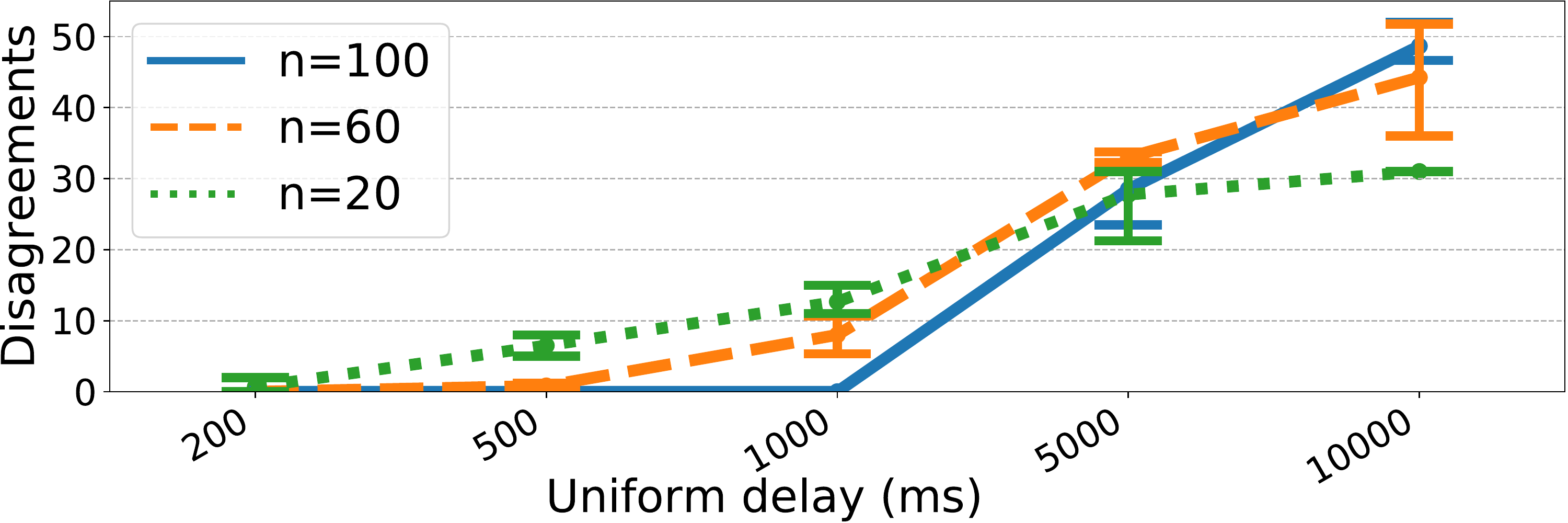}
  \caption[Disagreements caused by attackers with catastrophic, artificial delays.]{Disagreeing decisions for various catastrophic uniform delays with the binary consensus attack, for $f=d=\lceil 5n/9\rceil-1$.}
  \label{fig:catdelays}
  \vspace{-1.5em}
\end{figure}
While \blockchain is quite resilient to attacks for realistic but not catastrophic delays (Fig.~\ref{fig:fig2}), attackers can try to attack when the network collapses for a few seconds between regions. Our experiments, shown in Figure~\ref{fig:catdelays}, show that attackers can reach up to 52 disagreeing proposals for a uniform delay of 10 seconds between partitions of honest processes for the binary consensus attack, and up to 33 disagreements for a uniform delay of 5 seconds, with $n=100$. Further tests showed that the reliable broadcast attack reaches up to 165 disagreeing proposals with a 5-second uniform delay.
\subsubsection{Time to merge blocks and change members} 
To have a deeper understanding of the cause of \blockchain delays, we measured the time needed 
to merge blocks and to change members by replacing deceitful processes by new ones.
We show here the times to locally merge two blocks for different sizes assuming the worst case: all transactions conflict. This is the time taken in the worst case because processes can merge proposals that they receive concurrently (i.e., without halting consensus). Our experiments show that the times to merge two blocks of $100$, $1000$, and $10000$ transactions are $0.55$, $4.20$ and $41.38$ milliseconds, respectively.
It is clear that this time to merge blocks locally is negligible compared to the time it takes to run the consensus algorithm.

Figure~\ref{fig:fig7} shows the time to detect $f_d$ deceitful (top left), and to run the exclusion (top right) and inclusion (bottom left) consensus, for a variety of delays and numbers of processes. The time to detect reflects the time from the start of the attack until honest processes detect the attack:
If the first $f_d$ deceitful processes are forming a coalition together and cause a disagreement, then the times to detect the first deceitful and the first $f_d$ deceitful processes overlap.
(We detect all at the same time.) The time to exclude (57 seconds) is significantly larger than to include (21 seconds) for large communication delays, due to the proposals of the exclusion consensus carrying PoFs and leading processes to execute a time consuming cryptographic verification. With shorter communication delays, performance becomes practical.
Finally, Figure~\ref{fig:fig7} (bottom right) depicts the time to catch up depending on the number of proposals (i.e., blocks). As expected, this time increases linearly with the number of processes, due to the catchup requiring to verify larger certificates, but it remains practical at $n=100$ processes. The advantage of using certificates is that processes can catch up by just verifying certificates, instead of having to verify all transactions in the block that the certificate refers to.

\section{A Zero-Loss payment application}
\label{sec:zlbpayment}
In this section, we describe how \blockchain can be used to implement a \emph{zero-loss payment system} where no honest process loses any coin.  
The key idea is to request the consensus processes to deposit a sufficient amount of coins in order to spend, in case of an attack, the coins of deceitful processes to avoid any honest process loss.

 \subsection{Assumptions}
In order to measure the expected impact of a coalition attack succeeding with probability $\rho$ 
in forking \blockchain by leading a consensus to a disagreement, we first need to make the following assumptions:
\begin{enumerate}[leftmargin=*,wide=\parindent] 
 \item {\bf Fungible assets.}
We assume that users can transfer assets (like coins) that are \emph{fungible} in that one unit is interchangeable and indistinguishable from another of the same value. An example of a fungible asset is a cryptocurrency.

This zero-loss payment system can also work with non-fungible tokens
(NFTs) and smart contracts, with the exception that one of the two
recipients of the same NFT (or of disagreeing states) will see their
NFT taken back (or their returned state reverted) in exchange for a
previously agreed-upon reimbursement for the inconvenience.
\item {\bf Deposit refund.}
  To limit the impact of one successful double spending on a block, \blockchain keeps the deposit for a number of blocks $\mconf$, before returning it. A transaction should not be considered \textit{final} (i.e. irreversible) until it reaches this blockdepth $\mconf$. We call thus $\mconf$ the \textit{finalization blockdepth}. Attackers can fork into $a$ branches, and try to spend multiple times an amount $\mathfrak{G}$ (per block), which we refer to as the \textit{gain}, obtaining a maximum gain of $(a-1)\mathfrak{G}$. 
  Each honest process can calculate the gain by summing up all the outputs of all transactions in their decided block. Additionally, processes can limit the gain to an upper-bound by design, discarding blocks whose sum of outputs exceeds the bound, or they can allow the gain to be as much as the entire circulating supply of assets. The \textit{deposit} $\mathfrak{D}$ is a factor of the gain, i.e., $\mathfrak{D}= b\cdot \mathfrak{G}$. The goal is for every coalition to have at least $\mathfrak{D}$ deposited, and since every coalition has at least size $\ceil{n/3}$, this means that each process must deposit an amount $3b\mathfrak{G}/n$.

%
 \item {\bf Network control restriction.}
Once faulty processes select the disjoint subsets (i.e., the partitions) of honest processes to suffer the disagreement, 
we need to prevent faulty processes from communicating infinitely faster than honest processes in different partitions.
%
More formally, let $X_1$ (resp. $X_2$) be the random variables that indicate the time it takes for a message between two processes within the same partition (resp. two honest processes from different partitions). We have $E(X_1) / E(X_2) > \varepsilon$, for some $\varepsilon > 0$. Note that the definition of $X_1$ also implies that it is the random variable of the communication time of either two honest processes of the same partition or two faulty processes. This probabilistic synchrony assumption is similar to that of other blockchains (e.g. Bitcoin) that guarantee exponentially fast convergence, a result that also holds for \blockchain under the same assumptions. In the following, we show an analysis focusing on the attack on each consensus iteration, considering a successful disagreement if there is a fork in a single consensus instance, even for a short period of time. We discuss in Section~\ref{sec:probsyncdisc} the use of a random beacon for committee sortition in order to satisfy zero loss in a partially synchronous communication network. 
\end{enumerate}

\subsection{Theoretical analysis}
\label{sec:theoreticalanalysis}
 We show that attackers always fund at least as much as they steal.
 For ease of exposition, we consider that a membership change starts
before the deposit is refunded or does not start. Therefore, the
attack represents a Bernoulli trial that succeeds with probability
$\rho$ (per block) that can be derived from $\varepsilon$.
Out of one attack attempt, the attackers may gain up to $(a-1)\mathfrak{G}$ coins by forking into $a$ branches, or lose at least $\mathfrak{D}$ coins as a punishment, which can be used to fund the stolen funds from successful attacks. 

We introduce the random variable $Y$ that measures the number of
attempts for an attack to succeed and follows a geometric
distribution with mean $E(Y)=\frac{1-\hat{\rho}}{\hat{\rho}}$, where
$\hat{\rho} = 1 - \rho$ is the probability that the attack fails. Thus, we
define the expected gain of attacking:
$\mathcal{G}(\hat{\rho}) =(a-1)\cdot (\mathds{P}(Y>\mconf)\cdot \mathfrak{G})$
, and the expected punishment as:
$  \mathcal{P}(\hat{\rho}) =\mathds{P}(Y\leq \mconf)\cdot \mathfrak{D}$.
We can then define the expected \textit{deposit flux} per attack
attempt as the difference
$\xi=\mathcal{P}(\hat{\rho})-\mathcal{G}(\hat{\rho})$. Theorem~\ref{thm:zeroloss} shows the values for which
\blockchain yields zero loss.
\begin{theorem}[Zero-Loss Payment System]\label{thm:zeroloss}
Let 
$\rho$ be the probability of success of an attack per block, $\mathfrak{D}$ the minimum deposit per coalition expressed as a factor of the upper-bound on the gain $\mathfrak{D}=b \mathfrak{G}$, and $\mconf$ the finalization blockdepth to return the deposit.
If $g(a,b,\rho,\mconf)=(1-\rho^{\mconf+1})b-(a-1)\rho^{\mconf+1}\geq\,0$ then \blockchain implements a zero-loss payment system.
\end{theorem}
\begin{proof}
Recall that the maximum gain of a successful attack is $\mathfrak{G}
\cdot (a-1)$, and the expected gain $\mathcal{G}(\hat{\rho})$ and punishment $\mathcal{P}(\hat{\rho})$ for the attackers in a
disagreement attempt are as follows:
\begin{align*} \mathcal{G}(\hat{\rho}) =&(a-1)\cdot (\mathds{P}(Y>
\mconf)\cdot \mathfrak{G}) =(a-1)\cdot (\rho^{\mconf+1}\cdot \mathfrak{G}),\\
\mathcal{P}(\hat{\rho}) =&\mathds{P}(Y\leq \mconf)\cdot \mathfrak{D} =
(1-\rho^{\mconf+1})\mathfrak{D}=(1-\rho^{\mconf+1})b\mathfrak{G}.
  \end{align*}

Thus the deposit flux $\xi=\mathcal{P}(\hat{\rho})-\mathcal{G}(\hat{\rho})$:
$$\xi=\big((1-\rho^{\mconf+1})b-(a-1)\rho^{\mconf+1}\big)\mathfrak{G}=g(a,b,\rho,\mconf)\mathfrak{G}.$$

If $\xi< 0$ then a cost of $\mathcal{G}(\hat{\rho})-\mathcal{P}(\hat{\rho})$ is incurred to the system, otherwise the punishment is enough to fund the conflicts. Since the gain is non-negative $\mathfrak{G}\geq 0$, it follows that $g(a,b,\rho,\mconf)\geq 0$ for $\xi\geq 0$, obtaining zero loss.
\end{proof}



 
\subsubsection{Finalization blockdepth and deposit size}
Setting $c=\frac{b}{a-1+b}$, we can either calculate the probability $\rho\leq c^{\frac{1}{\mconf+1}}$ of success for an attack that \blockchain tolerates given a finalization blockdepth $\mconf$, or a needed finalization blockdepth $\mconf\geq \frac{\log(c)}{\log(\rho)}-1$ for a probability $\rho$ to yield zero loss, once we fix the deposit $\mathfrak{D}$ and upper-bound the gain $\mathfrak{G}$. For example, for $\updelta=0.5$ then $a=3$, and for a probability $\rho=0.55$, a finalization blockdepth of $\mconf=4$ blocks guarantees zero loss even if the deposit is a tenth of the maximum gain $\mathfrak{D}=\mathfrak{G}/10$, but with $\rho=0.9$ then $\mconf=28$. Whereas $a$ increases polynomially with $\rho$, it increases exponentially as the deceitful ratio $\updelta$ approaches the asymptotic limit $2/3$, leading to $\mconf=37$ blocks for $\updelta=0.6$, while $\mconf=46$ for $\updelta=0.64$, or
$\mconf=58$ for $\updelta=0.66$, with $\rho=0.9$ and $\mathfrak{D}=\mathfrak{G}/10$. 

\begin{figure}[t]
  \centering
  \includegraphics[width=.8\textwidth]{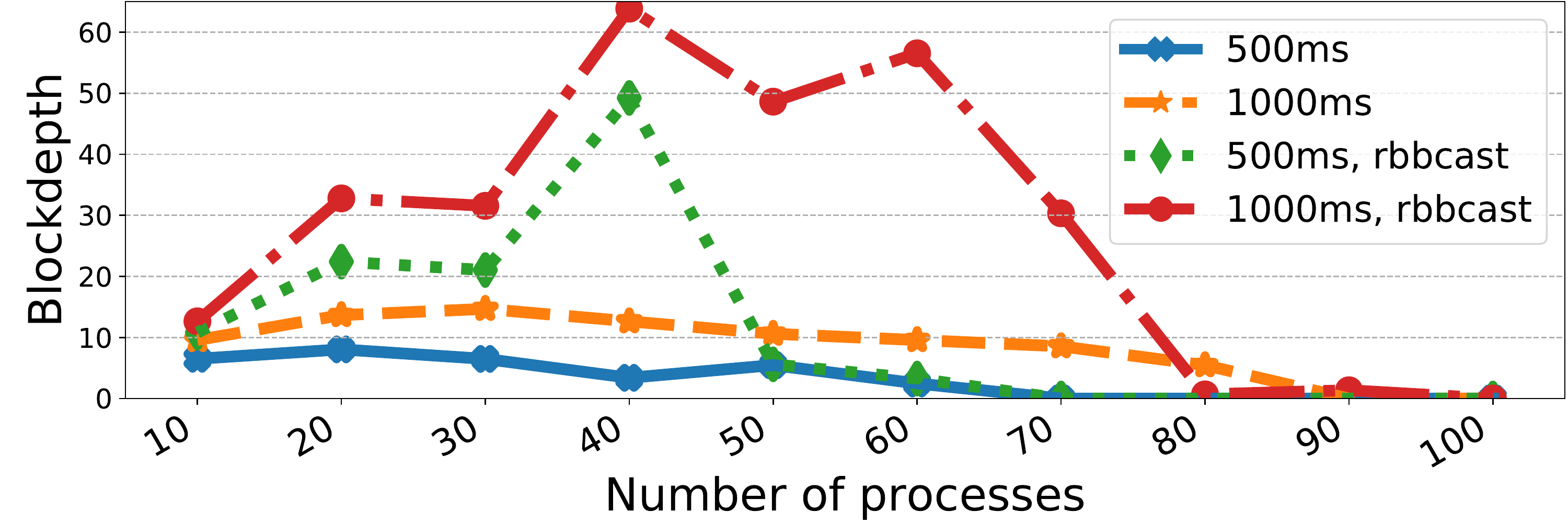}
  \caption[Required finalization blockdepth $\mconf$ for zero loss.]{Minimum finalization blockdepth $\mconf$ to obtain zero loss for $\mathfrak{D}=\mathfrak{G}/10$, $f=d=\lceil 5n/9\rceil-1$ and $q=0$.} 
  \vspace{-1em}
  \label{fig:fig6}
\end{figure}
 

\subsubsection{Experimental evaluation of the payment system}
Taking the experimental results of Section~\ref{sec:exper} and based on our aforementioned theoretical analysis, Figure~\ref{fig:fig6} depicts the minimum required finalization blockdepth $\mconf$ for a variety of uniform communication delays for $\mathfrak{D}=\mathfrak{G}/10$, $f=d=\lceil 5n/9\rceil-1$ and $q=0$. Again, we can see that the finalization blockdepth decreases with the number of processes, confirming that the zero loss property scales well. 
Additionally, small uniform delays yield zero loss at smaller values of $\mconf$, with all of them yielding $\mconf<5$ blocks for $n>80$. 
%
%
Although omitted in the figure, our experiments showed that even for a uniform delay of 10 seconds, setting $\mconf=50$ blocks (resp. $\mconf=168$ blocks) still yields zero loss in the case of a binary consensus attack (resp. reliable broadcast attack). 
Nevertheless, if the network performs normally, \blockchain will support large values of $f$, and will actually benefit from attacks, obtaining more than enough funds to cover the stolen amount.




\subsection{Discussion on probabilistic synchrony}
\label{sec:probsyncdisc}
We assumed probabilistic synchrony in Section~\ref{sec:zlbpayment} in order to introduce a probability of failure of an attack per consensus iteration. In partial synchrony, since the committee remains static until fraudsters are identified, the adversary can successfully perform an attack with probability of success $\rho=1$. There are, however, other factors that could influence the probability $\rho$ even in partial synchrony. For example, considering a blockdepth $\mconf\geq 1$, the implementation of a random beacon~\cite{GHM17} that replaces the committee in every iteration can decrease the probability of success of an attack. In such a case, the probability of an attack succeeding depends on the probability that the random beacon selects enough processes of the coalition (and enough of each of the partitions of honest processes) for $\mconf+1$ consecutive iterations, so that the coalition is able to perform the attack for $\mconf$ additional blocks. The design and proof of a random beacon that tolerates coalitions of sizes greater than $t_\ell$ is part of our future work.

\section{Conclusion}
\label{sec:zlbsum}
In this work, we presented the BDB failure model. We then presented
the Basilic class of protocols, that is resilient-optimal for the
problem of consensus in both the BFT and BDB model, and optimal in the
communication complexity, thanks to the \myproperty property that
states that deceitful behavior does not prevent liveness. We also
showed that the Basilic class of consensus protocols solves the
$\Diamond$-consensus problem in the BDB model for $d+t<h_0$ and
$q+t\leq n-h_0$, with $h_0\in(n/2,n]$ being the initial voting
threshold.

Following, we introduced \blockchain, the first blockchain that
tolerates a majority of faults. To this end, we first defined the
\blockchainproblem problem, to then detail \blockchain and prove
\blockchain's correctness. Additionally, we built and evaluated
\blockchain against a majority of attackers, and compared it with
previous works, offering competitive performance. We finally presented
a zero-loss payment application built on top of \blockchain that
guarantees that no honest process or user loses any fund from
temporary disagreements.

Our main future direction involves applying the advances presented in
this work to the random beacon problem, in order to propose a random
beacon protocol in the model here presented, stronger than the general
BFT model, in order to rotate the committee and remove the
probabilistic synchrony assumption from our zero-loss payment
application.
\bibliographystyle{plain}
\bibliography{journal-basilic-zlb}

\appendix
\section{ZLB Additional Results}
\label{sec:zlbappendixres}
\subsection{Number of branches and deceitful ratio}
We show in Figure~\ref{fig:plot23} the minimum deceitful ratio (left)
and number of deceitful processes (right) for the attackers to be able
to cause a disagreement into at least $a$ branches, for a voting
threshold of $h=2n/3$. It is specially interesting observing that a
coalition of less than half of the system cannot perform anything else
than a double-spending, while a deceitful ratio of $\updelta<11/18$ can
at most perform a sextuple-spending, while being significantly close,
at only $1/18$ of distance, to the threshold value of $2/3$. Notice
also that $a$ can only range from $1$ to $n/3+1$, a value that is
taken when $t+d = t_s=\ceil{2n/3}-1$ and each of the $n/3+1$ honest
processes belong to a different partition.
\begin{figure}[!hptb]
 \centering
 \subfloat[]{
   \includegraphics[width=.482\textwidth]{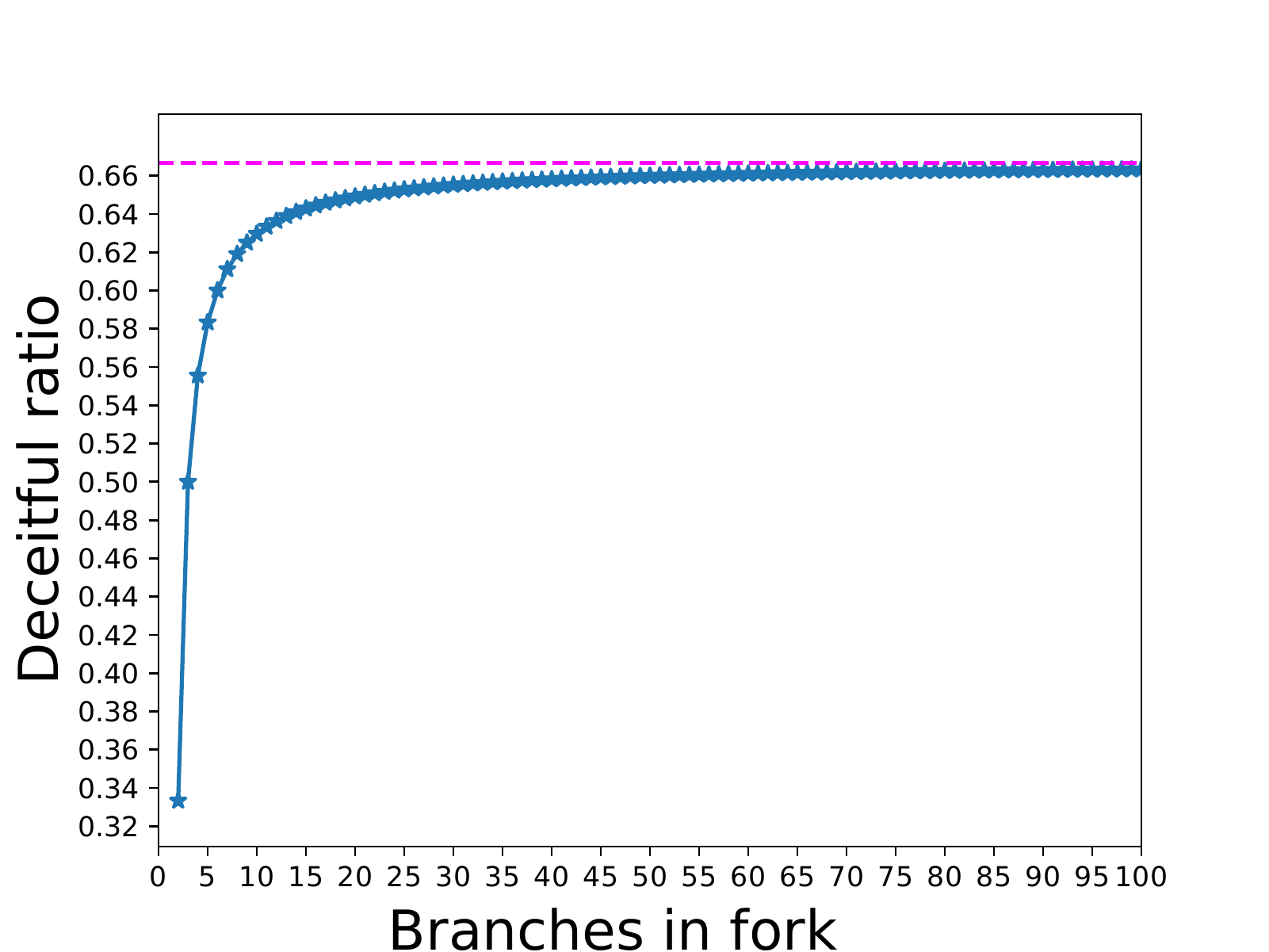}
 }\hfill
 \subfloat[]{
   \includegraphics[width=.482\textwidth]{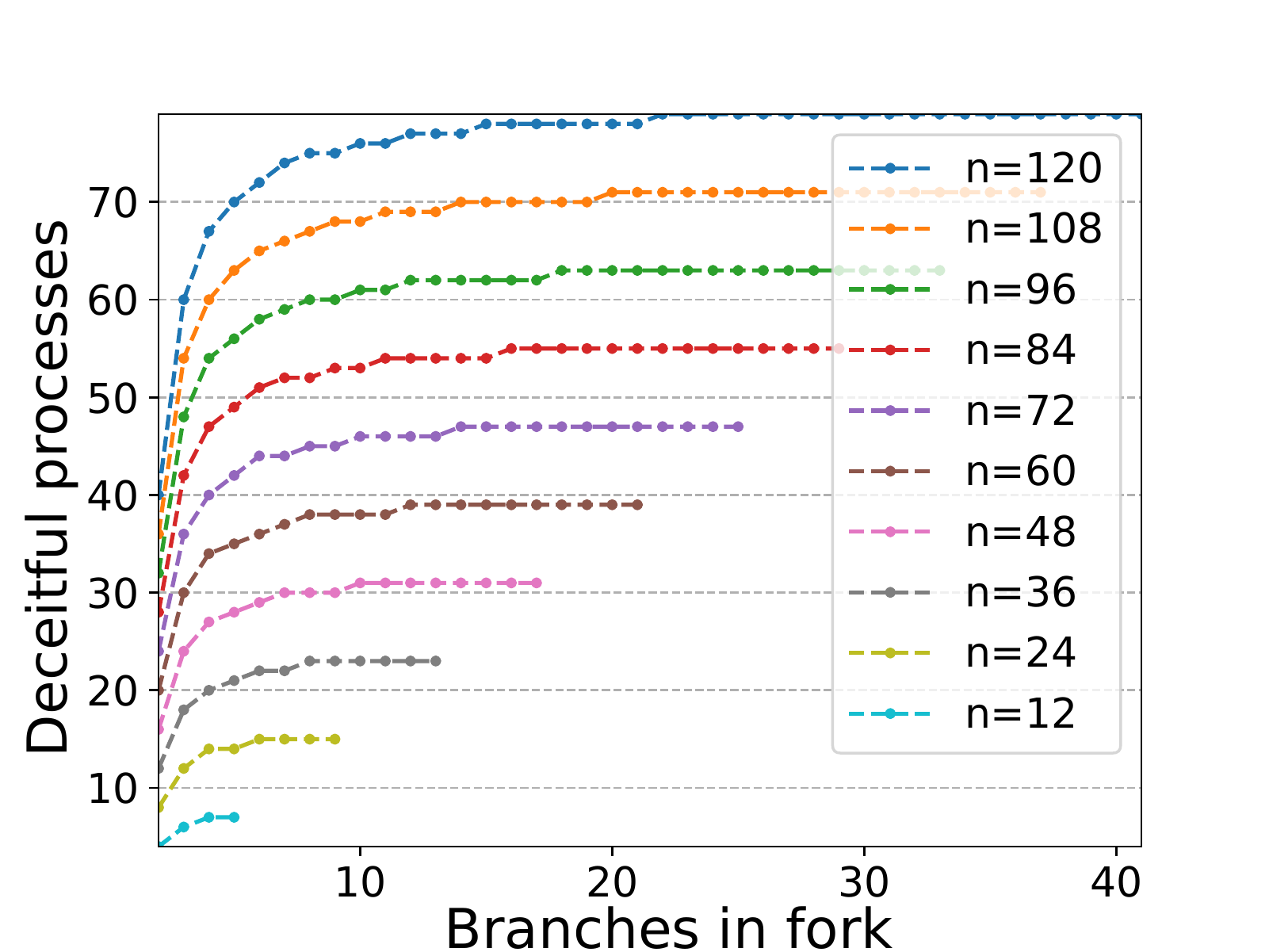}
 }
 \caption[Minimum deceitful ratio $\updelta$ and number of
deceitful processes required for a number of branches $a$ in a
blockchain fork for voting threshold $h=2n/3$.]{Minimum deceitful ratio $\updelta$ (left) and number of
deceitful processes (right) required for a number of branches $a$ in a
blockchain fork for voting threshold $h=2n/3$.}
 \label{fig:plot23}
\end{figure}
\subsection{Fixed superblock size}

Figure~\ref{fig:big} shows the throughput of \blockchain in a large WAN of up to 300 AWS c5.4xlarge instances, for a total size of all proposals fixed to $200,000$ transactions. We include two throughput results of our implementation, one for decisions and one for confirmations, for which we assume the maximum $f$ possible, i.e., all replies must be received before confirming a value. We can see that the throughput of confirmed transactions is slightly lower, given that every process must wait to receive a certificate from every single other process, increasing the impact of slow processes. The performance decreases as the number of processes increases, mainly due to the increase in size of certificates. We omitted the confirmation throughput in Figure~\ref{fig:fig1} as it was a negligible amount of time more than decisions for that setting, mainly due to the bottleneck being the validation of transactions, less noticeable for a fixed total size of transactions, as Figure~\ref{fig:big} shows.
\begin{figure}[h]
 \centering
 \includegraphics[width=.8\textwidth]{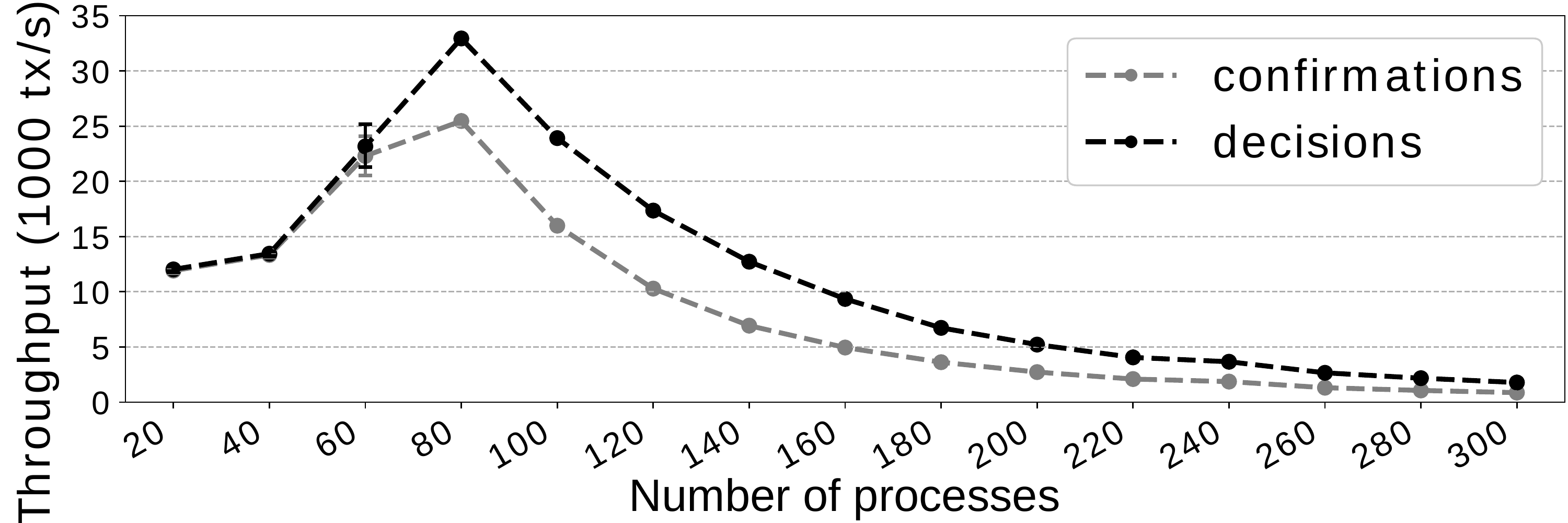}
 \includegraphics[width=.8\textwidth]{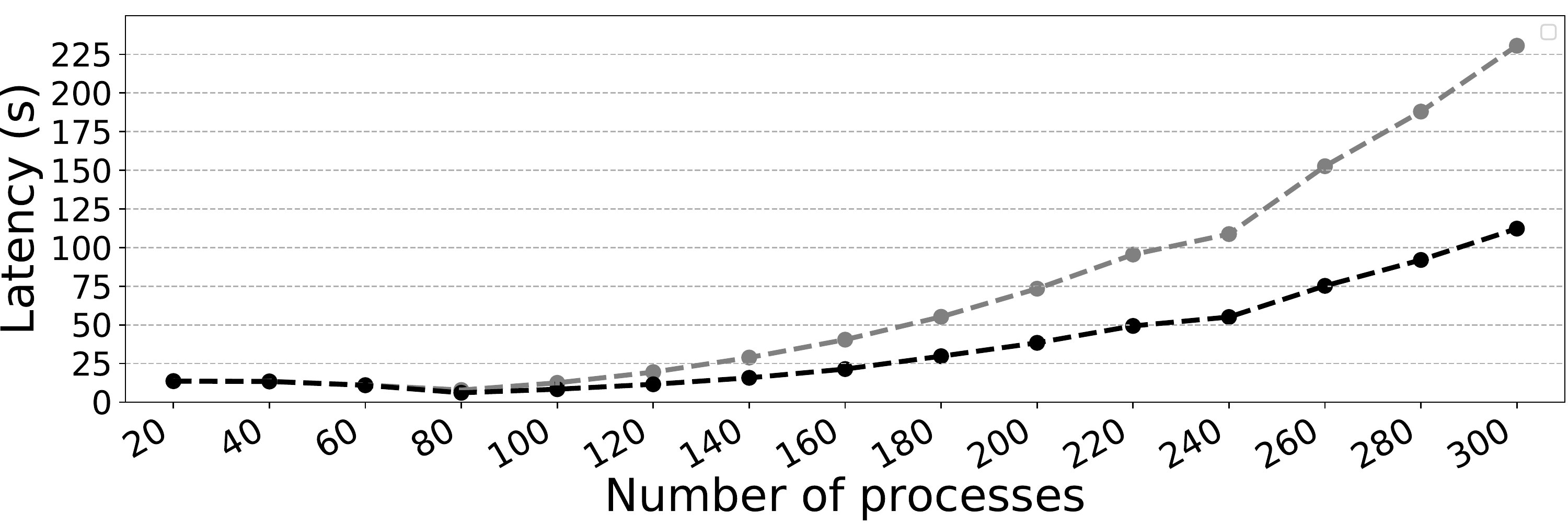}
 \includegraphics[width=.8\textwidth]{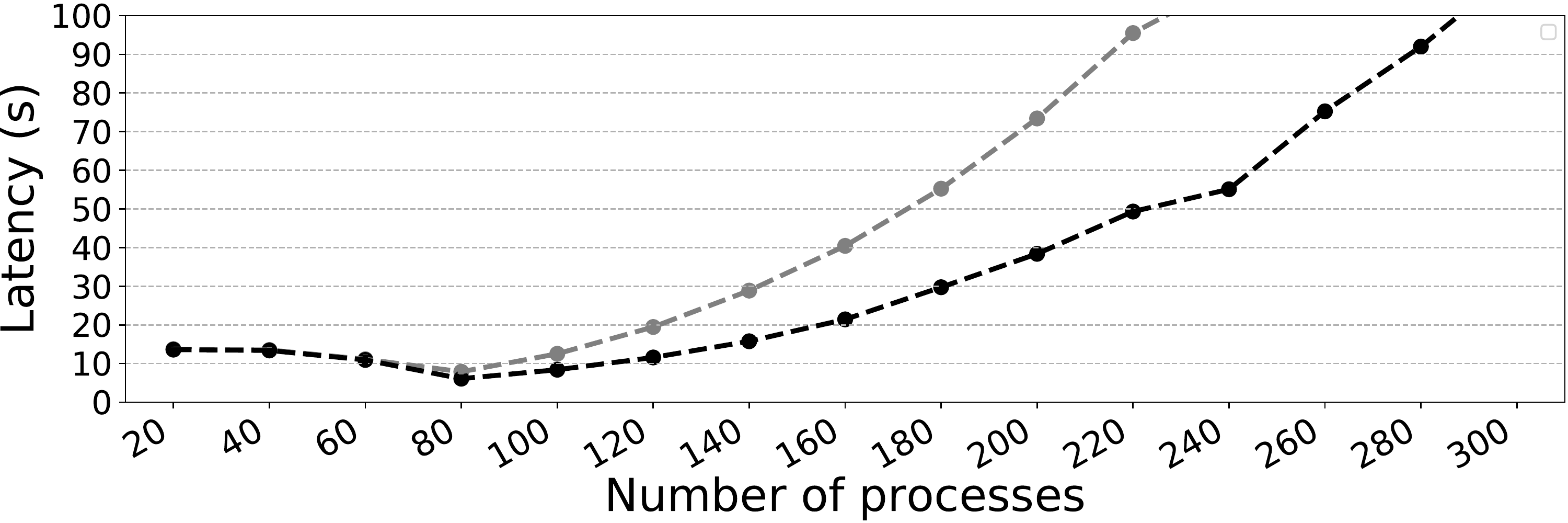}
 \caption{Throughput (top) and latency (center, bottom) of \blockchain for decided proposals, compared to confirmed proposals, for a size of the superblock fixed to $200,000$ transactions.}
 \label{fig:big}
\end{figure}
\subsection{Bitmask of binary consensus attack}
The binary consensus attack can maximize disagreements by leading all branches to a different bitmask. However, a disagreement on a bit associated with a proposal broadcast by an honest process might not contribute to the specific attack intended by attackers (e.g. double spending). We show in Figure~\ref{fig:minmax} however that if attackers do not maximize the disagreeing bits across branches, the number of disagreements decreases, even if they expose themselves in less binary consensus instances, for the binary consensus attack.
\begin{figure}[h]
 \centering
 \includegraphics[width=.8\textwidth]{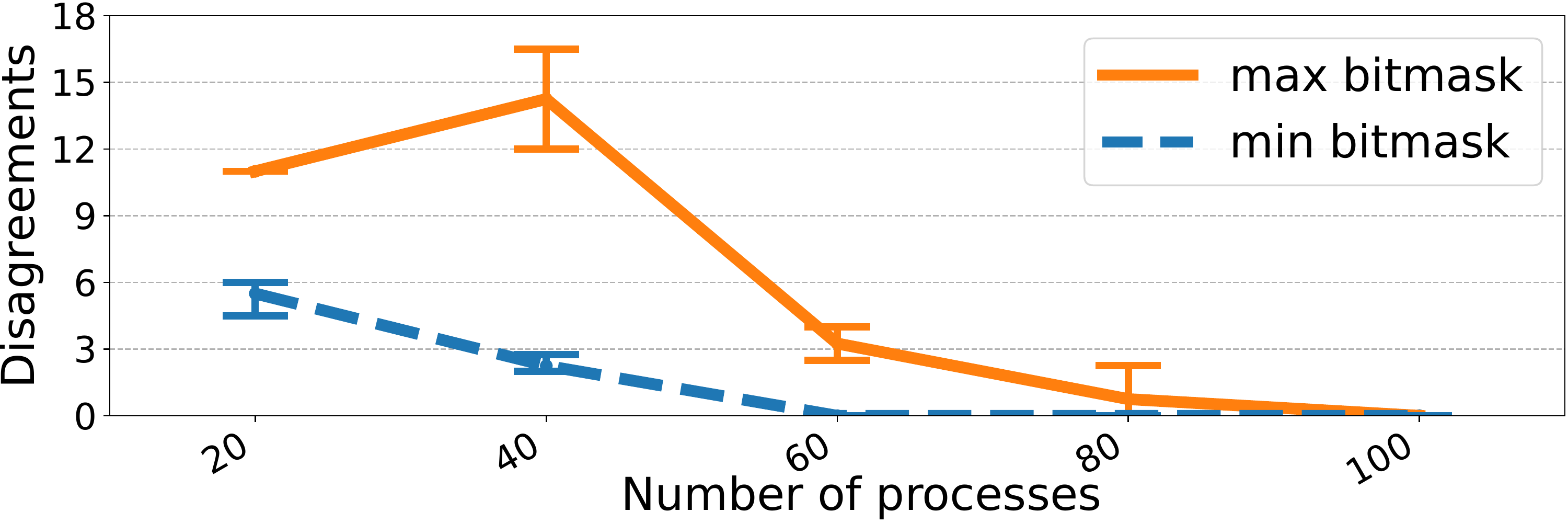}
 \caption[Disagreeing decisions for the binary consensus attack for a minimal and a maximal disagreement in number of bits per iteration.]{Disagreeing decisions for the binary consensus attack for a minimal disagreement per bitmask of one bit or a maximal of all bits, for $f=d=\ceil{5n/9}-1$.}
 \label{fig:minmax}
\end{figure}
\end{document}